\documentclass{article}

\usepackage{natbib}

\usepackage{tikz}
\usetikzlibrary{fit}					% fitting shapes to coordinates
\usetikzlibrary{backgrounds}	% drawing the background after the foreground

\usepackage{url}
\usepackage{amsmath}
\usepackage{amssymb}
\usepackage{amsbsy}
\usepackage{siunitx}
\usepackage{mathtools}
\usepackage{array}
\usepackage{setspace}

\usepackage{amsthm}
\usepackage{cancel}
\usepackage{bbm}

\usepackage{todonotes}
\newcommand{\norm}[1]{\left\lVert#1\right\rVert}

\newcommand{\asto}{\overset{a.s.}{\to}}
\newcommand{\iid}{\overset{i.i.d.}{\sim}}

\theoremstyle{plain}
\newtheorem{theorem}{Theorem}

\newcommand{\real}{\mathbb{R}}

\DeclareMathOperator*{\argmax}{arg\,\!max}
\DeclareMathOperator*{\var}{Var}
\newcommand{\E}{\mathbb{E}}
\newcommand{\N}{\mathbb{N}}

\usepackage{paralist}

\usepackage{thmtools}
\usepackage{thm-restate}

\usepackage{caption}
\usepackage{subcaption}
\usepackage{booktabs}
\usepackage{adjustbox}
\usepackage{microtype}

\usepackage{wrapfig}
\usepackage{tabularx}

\usepackage{times}
\usepackage[bottom]{footmisc}
\usepackage{listings}
\usepackage{color}
\usepackage{xcolor}
\usepackage{textcomp}
\usepackage{xspace}

\definecolor{darkgreenClj}{rgb}{0.25,.5,0.25}
\definecolor{blueClj}{rgb}{0,0.33,0.66}
\definecolor{redClj}{rgb}{0.66,0.0,0.0}
\definecolor{purpleClj}{rgb}{0.33,0,0.66}
\definecolor{cyanClj}{rgb}{0.0,0.5,0.5}
\definecolor{orangeClj}{rgb}{0.75,0.35,0.0}
\definecolor{grayClj}{rgb}{0.4,0.4,0.4}
\lstnewenvironment{code}[2]{\lstset{caption=#1,label=#2}}{}

\usepackage{abbreviations}

\usepackage{multicol}

\usepackage[accepted]{icml2018}
\usepackage{hyperref}

\usepackage[title]{appendix}
\usepackage{titlesec}
\titlespacing\section{0pt}{4pt plus 1pt minus 2pt}{-2pt plus 1pt minus 0pt}
\titlespacing\subsection{0pt}{4pt plus 1pt minus 2pt}{-2pt plus 1pt minus 0pt}
\titlespacing\subsubsection{0pt}{4pt plus 1pt minus 2pt}{-2pt plus 1pt minus 0pt}

\icmltitlerunning{On Nesting Monte Carlo Estimators}

\begin{document}
	
	\twocolumn[
	\icmltitle{On Nesting Monte Carlo Estimators}
	
	
	
	\icmlsetsymbol{equal}{*}
	
	\begin{icmlauthorlist}
		\icmlauthor{Tom Rainforth}{ox}
		\icmlauthor{Robert Cornish}{ox,ox2}
		\icmlauthor{Hongseok Yang}{kaist}
		\icmlauthor{Andrew Warrington}{ox2}
		\icmlauthor{Frank Wood}{ubc}
	\end{icmlauthorlist}
	
	\icmlaffiliation{ox}{Department of Statistics, University of Oxford}
	\icmlaffiliation{ox2}{Department of Engineering, University of Oxford}
	\icmlaffiliation{ubc}{Department of Computer Science, University of British Columbia}
	\icmlaffiliation{kaist}{School of Computing, KAIST}
	
	\icmlcorrespondingauthor{Tom Rainforth}{rainforth@stats.ox.ac.uk}
	
	\icmlkeywords{Monte Carlo, nested Monte Carlo, nested inference}
	
	\vskip 0.3in
	]
	\printAffiliationsAndNotice{}	
\setlength{\abovedisplayskip}{3pt}
\setlength{\belowdisplayskip}{3pt}
\setlength{\abovedisplayshortskip}{3pt}
\setlength{\belowdisplayshortskip}{3pt}
%
%\vspace{-15pt}

%\thispagestyle{empty}

\begin{abstract}
%	\vspace{-3pt}
Many problems in machine learning and statistics involve nested expectations and thus do not permit
conventional Monte Carlo (MC) estimation.  For such problems, one must nest estimators, such that terms in
an outer estimator themselves involve calculation of a separate, nested, estimation.  We investigate
the statistical implications of nesting MC estimators, including 
cases of multiple levels of nesting, and establish the conditions under which they
converge. We derive corresponding rates of convergence
and provide empirical evidence that these rates are observed in practice.
We further establish a number of pitfalls that can arise from na\"{i}ve nesting of MC estimators,
provide guidelines about how these can be avoided, 
%Our results show that whenever an outer estimator depends nonlinearly on an inner
%estimator, then the number of samples used in \emph{both} the inner and outer estimators
%must, in general, be driven to infinity for convergence.  
and lay out novel methods for reformulating certain classes of nested expectation problems
into single expectations, leading to improved convergence rates. % compared with na\"{i}ve NMC estimation.
We demonstrate the applicability of our work by using our results to develop
a new estimator for discrete Bayesian experimental design problems and derive error bounds
for a class of variational objectives.
\end{abstract}

\vspace{-16pt}

% !TEX root =  main.tex

\section{Introduction}
\label{sec:intro}

%Although interesting alternatives have recently been suggested \cite{briol2015probabilistic}, MC integration is almost exclusively used to calculate the expectation, given the generated samples.
%the method used in practise for calculating these expectations, given the generated samples, is almost exclusively MC integration.  
%The calculation of expectations using MC can be considered intertwined with that of MC inference.  
%After all, convergence rates that are often quoted for MC inference schemes, actually correspond to the convergence of the final MC integration estimate.

Monte Carlo (MC) methods are used throughout the quantitative sciences.
For example, they have become a ubiquitous means of carrying out approximate Bayesian inference \citep{doucet2001introduction,gilks1995markov}.
%, generating approximate samples from a posterior from which an expectation can be estimated.
%From simplistic Metropolis Hastings approaches to state-of-the-art algorithms such as the bouncy particle sampler \cite{bouchard2015bouncy} and interacting particle Markov chain MC \cite{rainforth2016interacting},
%%pseudo-marginal Hamiltonian MC \cite{lindsten2016pseudo}, 
%The aim of these methods is always the same: to generate approximate samples from a posterior, from which an expectation can be calculated.  
%Although interesting alternatives have been suggested recently \citep{briol2015probabilistic}, 
The convergence of MC estimation
 has been considered extensively in the
literature \citep{durrett2010probability}.  However, the implications
arising from the \emph{nesting} of MC schemes, where terms in the integrand depend on the
result of separate, nested, MC estimators, is generally less well known.
%By nesting, we refer to problems in which terms in our MC integration are themselves a function of an intractable expectation and must be estimated using a separate, nested, MC scheme.  
%This problem is distinct to that of conjoining of Monte Carlo schemes, such as is done in sequential Monte Carlo \cite{smith2013sequential}, where the inference is broken down into a separate parts, but the overall estimation is still a single Monte Carlo integration.
This paper examines the convergence of such nested Monte Carlo (NMC) methods.  

Nested expectations occur in wide variety of problems from portfolio risk management
\citep{gordy2010nested} to stochastic control \citep{belomestny2010regression}. 
In particular, 
simulations of agents that reason about other agents often include nested expectations.
Tackling such problems requires some form of nested estimation scheme like NMC.

A common class of nested expectations is
doubly-intractable inference problems~\citep{murray2006mcmc,liang2010double}, where
the likelihood is only known up to a parameter-dependent normalizing constant. This can occur, for
example, when nesting
probabilistic programs~\citep{mantadelis2011nesting,le2016nested}.  Some problems
are even multiply-intractable, such that they require multiple levels of nesting to
encode~\citep{stuhlmuller2014reasoning}.  
Our results can be used to show that changes are required to the approaches currently
employed by probabilistic programming systems to ensure consistent estimation for such
problems~\citep{rainforth2017thesis,rainforth2017nestpp}.
%
%The accompanying systems
%make implicit assumptions on the associated back-end inference engines remaining valid in this
%setting, something which we will later call into question.

%, and completely-intractable
%problems~\citep{wilkinson2013approximate} where the likelihood cannot be evaluated at all.  Such cases 
%are often tackled using ABC methods, whereby simulation is used to approximate the likelihood~\citep{csillery2010approximate}.
%Here
%the normalization constant itself represents an intractable expectation
%nested within the overall expectation.

The expected information gain used in
Bayesian experimental design \citep{chaloner1995bayesian} 
%\citep{chaloner1995bayesian,sebastiani2000maximum}
requires
the calculation of an entropy of a marginal distribution and therefore the
expectation of the logarithm of an expectation.  By extension, any Kullback-Leibler
divergence where one of the terms is a marginal distribution also involves a nested expectation.  Hence, our results have important implications for 
relaxing mean-field assumptions, or using different bounds, in variational
inference \citep{hoffman2015stochastic,naesseth2017variational,maddison2017filtering} and deep generative models
\citep{burda2015importance,le2017auto}.
%Here the nonlinearity provided by the logarithm prevents a simple reformulation to a
%single expectation in the general case, and thus presents conventional MC estimation.

Certain nested estimation problems can be tackled by pseudo-marginal methods
\citep{beaumont2003estimation,andrieu2009pseudo,andrieu2010particle}.
These consider inference problems where the likelihood is intractable, 
%such as
%when it originates from an Approximate Bayesian Computation (ABC)
%\citep{csillery2010approximate}, 
but can be estimated unbiasedly.
% or when sequential Monte Carlo \cite{smith2013sequential} is used to approximate a high
% dimensional distribution.  
From a theoretical perspective, they
reformulate the problem in an extended space with auxiliary variables that
are used to represent the stochasticity in the likelihood computation, enabling the
problem to be expressed as a single expectation.

Our work goes beyond this by considering cases in which a non-linear mapping is
applied to the output of the inner expectation, 
(e.g.~the logarithm in the experimental design example), 
prohibiting such reformulation.
We demonstrate that the construction of consistent NMC algorithms is possible,
establish convergence rates, and provide empirical evidence that these rates are observed in practice.
Our results show that whenever an outer estimator depends non-linearly on an inner
estimator, then the number of samples used in \emph{both} the inner and outer estimators
must, in general, be driven to infinity for convergence. We extend our results
to cases of repeated nesting and show that the optimal NMC convergence rate 
is $O(1/T^{\frac{2}{D+2}})$ where $T$ is the total
number of samples used in the estimator and $D$ is the nesting depth (with $D=0$ being conventional MC), whereas
na\"{i}ve approaches only achieve a rate of $O(1/T^{\frac{1}{D+1}})$.  
We further lay out methods for reformulating certain classes of nested expectation problems
into a single expectation, allowing usage of conventional MC estimation 
schemes with superior convergence rates than na\"{i}ve NMC.  Finally, we use our results
to make application-specific advancements in Bayesian experimental design and variational auto-encoders.
%and probabilistic programming. 
%use one of these results 
%to derive a new estimation scheme for Bayesian experimental design with discrete outputs that has 
%a superior convergence rate to existing approaches.

%More specifically, we show
%that NMC typically converges at a rate $O(1/N_0 +(\sum_{k=1}^{D} 1/N_k)^2)$ where $N_k$
%is the number of samples used for each call of the estimator at depth $k$.  Given the
%computational cost of NMC is $O(T)$ where $T=\prod_{k=0}^{D} N_k$, we show that the optimal
%convergence rate for NMC is $O(1/T^{\frac{2}{D+2}})$ where $D$ is the maximum nesting depth.
%Though this means that NMC produces exponentially worse estimates as $D$ increases, this
%rate is still substantially faster than the $O(1/T^{\frac{1}{D+1}})$ rate that is achieved
%by na\"{i}vely schedules for setting $N_k$.

\subsection{Related Work}

Though the convergence of NMC has previously received little attention within the machine learning literature,
a number of special cases having been investigated in other fields, sometimes under the
name of \textit{nested simulation}~\citep{longstaff2001valuing,belomestny2010regression,gordy2010nested,broadie2011efficient}.
%NMC with non-linear mappings of the inner expectation has been previously considered in
%the financial statistics literature, for example in the pricing of American
%options \citep{longstaff2001valuing}. 
While most of this literature focuses on particular application-specific non-linear 
mappings, a convergence bound for a wider range of problems was shown by \citet{hong2009estimating} 
and recently revisited in the context of rare-event problems by~\citet{fort2016mcmc}.  The latter paper further
considers the case where samples in the outer estimator originate from a Markov chain.
Compared to this previous work, ours is the first to consider multiple levels of nesting, applies
to a wider range of non-linear mappings, and provides more precise convergence rates.
By introducing new results, outlining special cases, providing empirical assessment, and examining
specific applications, we provide a unified investigation and practical guide nesting MC
estimators in a machine learning context.  We begin to realize the potential significance of this
by using our theoretical results to make advancements in a number of specific application
areas.

Another body of literature related to our work is in the study of the convergence of Markov chains
with approximate transition kernels~\cite{rudolf2015perturbation,alquier2016noisy,medina2016stability}.
The analysis in this work is distinct, but complementary, to our own, focusing on the impact of a
known bias on an MCMC chain, whereas our focus is more on the quantifying this bias.  Also
related is the study of techniques for variance reduction, such as multilevel 
MC~\cite{heinrich2001multilevel,giles2008multilevel}, and bias reduction, such as 
the multi-step Richardson-Romberg method~\cite{pages2007multi,lemaire2017multilevel} and 
Russian roulette sampling~\citep{lyne2015russian}, many of which are applicable in a NMC
context and can improve performance.
% for certain problems.
%
%We build on these results and outline the opportunities and pitfalls of nesting Monte Carlo
%estimators in a machine learning context.

%
%Our proofs apply to a more general class of problems than existing approaches, placing, for example, weaker restrictions
%on the non-linear mapping.  Notably, we provide the first convergence bounds and proof of consistency for 
%cases of multiple levels of estimator nesting, a result of particular importance to the probabilistic 
%programming community.
%We further provide theoretical results showing that any MC
%estimator using imperfect nested estimates is, in general, biased.

%We also lay out novel methods for reformulating certain classes of nested expectation problems
%into a single expectation, allowing the use of conventional MC estimation 
%schemes with superior convergence rates than na\"{i}ve use of NMC.  Finally, we use one of these results 
%to derive a new estimation scheme for Bayesian experimental design with discrete outputs that has 
%a superior convergence rate to existing approaches.

%\tom{Policy search?}

%We now arrive at the crux of this paper: we aim to establish in what nested inference scenarios one can guarentee convergence; we demonstrate that even when a general purpose scheme convergences, it must be biased; and we provide upper bounds on the convergence rate that such a scheme can achieve, showing that it decreases exponentially with the depth of nesting.

%
%\todo[inline]{Highlight some of the non-intuitive elements of the results.}
% !TEX root =  main.tex

\section{Problem Formulation}
\label{sec:prob-form}

The key idea of MC is that the expectation of an arbitrary function 
$\lambda \colon \mathcal{Y} \rightarrow \mathcal{F} \subseteq \real$ under a probability distribution $p(y)$ for its input $y \in \mathcal{Y}$ can be approximated using:
\begin{align}
I &= \mathbb{E}_{y\sim p(y)} \left[\lambda(y)\right] \displaybreak[0] \\
\label{eq:MC}
&\approx \frac{1}{N} \sum_{n=1}^{N} \lambda(y_n) \quad \text{where} \quad y_n \iid p(y).
\end{align}
In this paper, we consider the case that $\lambda$ is itself intractable, defined only in terms of a functional mapping of an expectation. Specifically, $\lambda(y) = f(y,\gamma(y))$
where we can evaluate $f \colon \mathcal{Y} \times \Phi \rightarrow \mathcal{F}$ exactly for a given $y$ and $\gamma (y)$, but $\gamma(y)$ is the output of the following 
intractable expectation of another variable $z \in \mathcal{Z}$:
\begin{subequations}
	\label{eq:gamma}
	\begin{align}
	\label{eq:gamma_1}
	\text{either}\quad
	\gamma(y) &=  \mathbb{E}_{z\sim p(z | y)} \left[\phi(y,z)\right] \\
	\label{eq:gamma_2}
	\text{or} \quad \gamma(y) &= \mathbb{E}_{z\sim p(z)} \left[\phi(y,z)\right]
	\end{align}
\end{subequations}
depending on the problem, with $\phi \colon \mathcal{Y} \times \mathcal{Z} \rightarrow \Phi$.
All our results apply to both cases, but we will focus on~\eqref{eq:gamma_1} for clarity.
Estimating $I$ involves computing an integral over $z$ for each value of $y$ in the outer integral.
We refer to the approach of tackling both integrations using MC
as \emph{nested Monte Carlo} (NMC):
\begin{subequations}
\label{eq:nested-mc}
\begin{align}
\hspace{-4pt} I =  \mathbb{E}\left[f(y,\gamma(y))\right] 
&\approx I_{N,M} = \frac{1}{N}  \sum_{n=1}^{N} f(y_n,(\hat{\gamma}_M)_n) \label{eq:nested-outer}
\intertext{where $y_n \iid p(y)$ \text{and}}
(\hat{\gamma}_M)_n &= \frac{1}{M}  \sum_{m=1}^{M}  \phi(y_n,z_{n,m}) \label{eq:nested-inner}
\end{align}
\end{subequations}
where each $z_{n,m} \sim p(z | y_n)$ are independently sampled.
In Section~\ref{sec:convergence} we will build on this further by considering cases with multiple
levels of nesting, where calculating $\phi(y,z)$ involves computation of an intractable (nested) expectation.
%Note that for the experimental design example given in~\eqref{eq:exp-design},
%$f(y,\gamma(y)) = \log \gamma(y)$, $\gamma(y)$ is of the form~\eqref{eq:gamma_2}, and
%$\phi(y,z)$ is the likelihood function  $p(y|z,x)$.
%
%The rest of this paper proceeds as follows. In Section~\ref{sec:special_cases}, we consider
%special cases that allow recovery of the standard MC convergence rate.
%In Section~\ref{sec:convergence}, we establish convergence results for $I_{N,M}$ given a
%general class of $f$ and extend this to cases of repeated nesting. In Section~\ref{sec:bias}, we establish results demonstrating the
%inevitable bias of most possible general-purpose NMC schemes. Finally, in Section~\ref{sec:empirical}, 
%we present empirical results demonstrating the applicability of NMC and suggesting that our theoretical 
%convergence rates are observed in practice.  Appendices are provided in the supplementary material,
%containing, for example, the majority of the proofs.

% !TEX root =  main.tex

\section{Convergence of Nested Monte Carlo}
\label{sec:convergence}
%
%Since we cannot always unravel the target estimation using one of the special cases in the previous section, we must resort to NMC (or another nested estimation scheme) in order to compute $I$ in general. 
%Our aim here is to show that
We now show that approximating $I \approx I_{N,M}$ is in principle possible, at least when $f$ is
well-behaved. In particular, we establish a convergence rate of 
the mean squared error of $I_{N,M}$ and prove a form of almost sure convergence to
$I$.  We 
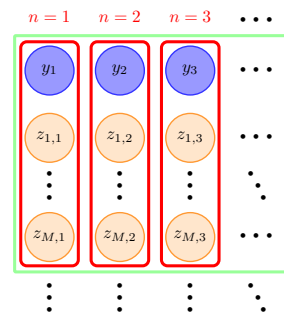
\begin{wrapfigure}{r}{0.25\textwidth}
	\vspace{-12pt}
	\centering 
	\resizebox{.25\textwidth}{!}{
		\tikzstyle{smc}=[circle,
                                    thick,
                                    minimum size=1cm,
                                    draw=blue!80,
                                    fill=blue!40]
\tikzstyle{csmc}=[circle,
                                    thick,
                                    minimum size=1cm,
                                    draw=orange!80,
                                    fill=orange!20]
\tikzstyle{iter} = [text=red]

\tikzstyle{background}=[rectangle,
                                                draw=green!40,
                                                line width=0.0625cm,
                                                inner sep=0.2cm,
                                                rounded corners=0.25mm]

\tikzstyle{iterback}=[rectangle,
draw=red,
line width=0.0625cm,
inner sep=0.075cm,
rounded corners=1.25mm]

\begin{tikzpicture}[>=latex,text height=0.375ex,text depth=0.0625ex]
  \matrix[row sep=0.35cm,column sep=0.35cm] {
    %&
    \node (n1) [iter]{$n=1$}; &
    \node (n2) [iter]{$n=2$}; &
    \node (n3) [iter]{$n=3$}; &
%    \node (n4)[iter]{$n=4$}; &
%    \node (n5) [iter]{$n=5$}; &
%    \node (n6) [iter]{$n=6$}; &
%    \node (n7) [iter]{$n=7$}; &
%    \node (n8)[iter]{$n=8$}; &
    \node (n9) []{\tiny $\bullet \; \bullet \; \bullet$}; &
    \\
        \node (y1) [smc]{$y_1$}; &
        \node (y2) [smc]{$y_2$}; &
        \node (y3) [smc]{$y_3$}; &
%        \node (y4)[smc]{$y_4$}; &
%        \node (y5) [smc]{$y_5$}; &
%        \node (y6) [smc]{$y_6$}; &
%        \node (y7) [smc]{$y_7$}; &
%        \node (y8)[smc]{$y_8$}; &
        \node (y9) []{\tiny $\bullet \; \bullet \; \bullet$}; &
        \\
        \node (z_11) [csmc]{$z_{1,1}$};      &
        \node (z_12) [csmc]{$z_{1,2}$};      &
        \node (z_13) [csmc]{$z_{1,3}$};     &
%        \node (z_14) [csmc]{$z_{1,4}$};      &
%        \node (z_15) [csmc]{$z_{1,5}$};     &
%        \node (z_16) [csmc]{$z_{1,6}$};      &
%        \node (z_17) [csmc]{$z_{1,7}$};      &
%        \node (z_18) [csmc]{$z_{1,8}$};     &
        \node (z_19) []{\tiny $\bullet \; \bullet \; \bullet$};      &
        \\
        \node (d11) []{\tiny $\begin{matrix}
        	\bullet \\
        	\bullet \\
        	 \bullet \\
        	\end{matrix}
        $}; &
        \node (d12) []{\tiny $\begin{matrix}
        	\bullet \\
        	\bullet \\
        	\bullet \\
        	\end{matrix}
        	$}; &
        \node (d13) []{\tiny $\begin{matrix}
        	\bullet \\
        	\bullet \\
        	\bullet \\
        	\end{matrix}
        	$}; &
%        \node (d14) []{\tiny $\begin{matrix}
%        	\bullet \\
%        	\bullet \\
%        	\bullet \\
%        	\end{matrix}
%        	$}; &
%        \node (d15) []{\tiny $\begin{matrix}
%        	\bullet \\
%        	\bullet \\
%        	\bullet \\
%        	\end{matrix}
%        	$}; &
%        \node (d16) []{\tiny $\begin{matrix}
%        	\bullet \\
%        	\bullet \\
%        	\bullet \\
%        	\end{matrix}
%        	$}; &
%        \node (d17) []{\tiny $\begin{matrix}
%        	\bullet \\
%        	\bullet \\
%        	\bullet \\
%        	\end{matrix}
%        	$}; &
%        \node (d18) []{\tiny $\begin{matrix}
%        	\bullet \\
%        	\bullet \\
%        	\bullet \\
%        	\end{matrix}
%        	$}; &
        \node (d19) []{\tiny $\begin{matrix}
        	\bullet \; \; \; \; \\
        	\; \; \bullet \; \; \\
        	\; \; \; \; \bullet \\
        	\end{matrix}
        	$}; &
        \\
        \node (z_M1) [csmc]{$z_{M,1}$};      &
        \node (z_M2) [csmc]{$z_{M,2}$};      &
        \node (z_M3) [csmc]{$z_{M,3}$};     &
%        \node (z_M4) [csmc]{$z_{M,4}$};      &
%        \node (z_M5) [csmc]{$z_{M,5}$};     &
%        \node (z_M6) [csmc]{$z_{M,6}$};      &
%        \node (z_M7) [csmc]{$z_{M,7}$};      &
%        \node (z_M8) [csmc]{$z_{M,8}$};     &
        \node (z_M9) []{\tiny $\bullet \; \bullet \; \bullet$};      &
        \\
        \node (d21) []{\tiny $\begin{matrix}
        	\\\\
        	\bullet \\
        	\bullet \\
        	\bullet \\
        	\end{matrix}
        	$}; &
        \node (d22) []{\tiny $\begin{matrix}
        	\\\\
        	\bullet \\
        	\bullet \\
        	\bullet \\
        	\end{matrix}
        	$}; &
        \node (d23) []{\tiny $\begin{matrix}
        	\\\\
        	\bullet \\
        	\bullet \\
        	\bullet \\
        	\end{matrix}
        	$}; &
%        \node (d24) []{\tiny $\begin{matrix}
%        	\\
%        	\bullet \\
%        	\bullet \\
%        	\bullet \\
%        	\end{matrix}
%        	$}; &
%        \node (d25) []{\tiny $\begin{matrix}
%        	\\
%        	\bullet \\
%        	\bullet \\
%        	\bullet \\
%        	\end{matrix}
%        	$}; &
%        \node (d26) []{\tiny $\begin{matrix}
%        	\\
%        	\bullet \\
%        	\bullet \\
%        	\bullet \\
%        	\end{matrix}
%        	$}; &
%        \node (d27) []{\tiny $\begin{matrix}
%        	\\
%        	\bullet \\
%        	\bullet \\
%        	\bullet \\
%        	\end{matrix}
%        	$}; &
%        \node (d28) []{\tiny $\begin{matrix}
%        	\\
%        	\bullet \\
%        	\bullet \\
%        	\bullet \\
%        	\end{matrix}
%        	$}; &
        \node (d29) []{\tiny $\begin{matrix}
        	\\\\
        	\bullet \; \; \; \; \\
        	\; \; \bullet \; \; \\
        	\; \; \; \; \bullet \\
        	\end{matrix}
        	$}; &
        \\
    };
    
%    \node (conv) []{$\E[f(y,\frac{1}{M}\sum_{m=1}^{M} \phi(y,z_m))]$}
%
%    \path[->] (d19) -- (conv.east);
        %(A) edge[thick] (x)	% 

    \begin{pgfonlayer}{background}
    \node [iterback,
    fit=(y1) (z_M1)] {};
        \node [iterback,
        fit=(y2) (z_M2)] {};
            \node [iterback,
            fit=(y3) (z_M3)] {};
%                        \node [iterback,
%                        fit=(y4) (z_M4)] {};
%                                    \node [iterback,
%                                    fit=(y5) (z_M5)] {};
%                                                \node [iterback,
%                                                fit=(y6) (z_M6)] {};
%                                                            \node [iterback,
%                                                            fit=(y7) (z_M7)] {};
%                                                                        \node [iterback,
%                                                                        fit=(y8) (z_M8)] {};
    \end{pgfonlayer}
        ;
    \begin{pgfonlayer}{background}
    \node [background,
    fit=(y1) (z_M1) (z_M9)] {};
    \end{pgfonlayer}
\end{tikzpicture}
	}
	\caption{Informal convergence representation \label{fig:conv-rep}}
	\vspace{-15pt}
\end{wrapfigure}
further generalize our convergence rate to apply to the case of multiple
levels of estimator nesting.

%\subsection{Strong Consistency}

Before providing a formal examination of the convergence of NMC, we first provide 
intuition about how we might expect to construct a convergent NMC estimator.  Consider the
diagram shown in Figure~\ref{fig:conv-rep}, and suppose that we want our error to be
less than some arbitrary $\varepsilon$.  Assume that $f$ is sufficiently smooth 
that we can choose $M$ large enough to make
$\left|I-\E\left[f(y_n,(\hat{\gamma}_M)_n)\right]\right| < \varepsilon$
(we will characterize the exact requirements for this later).  For this fixed
$M$, we have a standard MC estimator on an extended space $y,z_1,\dots,z_M$ such that each
sample constitutes one of the red boxes.  As we take $N\rightarrow \infty$, i.e. taking
all the samples in the green box, this estimator converges such that $I_{N,M}
\to \E\left[f(y_n,(\hat{\gamma}_M)_n)\right]$ as $N \to \infty$ for fixed $M$.  As we can
make $\varepsilon$ arbitrarily small, we can also achieve an arbitrarily small error.

More formally, convergence bounds for NMC have previously been shown by~\citet{hong2009estimating}.  
Under the assumptions that each $\left(\hat{\gamma}_M\right)_n$ is Gaussian distributed
(which is often reasonable due to the central limit theorem) and that $f$ is thrice differentiable
other than at some finite number of points, they show that it is possible to achieve a
converge rate of $O(1/N+1/M^2)$.  We now show that these assumptions can be relaxed to only requiring
$f$ to be Lipschitz continuous, at the expense of weakening the bound. % this bound to $O(1/N+1/M)$.
%\vspace{-6pt}
\begin{restatable}{theorem}{theRate} \label{the:Rate}
	If $f$ is Lipschitz continuous and $f(y_n, \gamma(y_n)), \phi(y_n, z_{n,m}) \in
	L^2$, the mean squared error of $I_{N,M}$ converges to $0$ at rate $O\left(1/N +
	1/M\right)$.
\end{restatable}
\vspace{-10pt}
\begin{proof}
The theorem follows as a special case of Theorem~\ref{the:Repeat}.  For exposition, a more accessible proof for this
particular result is also provided in Appendix~\ref{sec:app:rate_single} in the supplement.
\end{proof}
\vspace{-6pt}
%\begin{remark}
%This result can be carried over to cases where there are finite numbers of points for which $f$ is
%not Lipschitz continuous by decomposing the estimator into separate truncated functions as
%per~\cite{hong2009estimating}.
%\end{remark}
Inspection of the convergence rate above shows that, given a total number of samples
$T=MN$, our bound is tightest when $N\propto M$, with a
corresponding rate $O(1/\sqrt{T})$ (see Appendix~\ref{sec:app:opt-conv}). 
When the additional assumptions of~\citet{hong2009estimating}
apply, this rate can be lowered to $O(1/T^{2/3})$ by setting $N \propto M^2$.  We
will later show that this faster convergence rate can be achieved whenever $f$ is
continuously differentiable, see also~\cite{fort2016mcmc}.

%We note that a similar result to 
%Theorem~\ref{the:Rate} was recently independently derived by~\citet{fort2016mcmc}, who further
%consider the case where the $y_n$ are generated by a Markov chain.
%%in work published shortly after an earlier version of our own~\citep{rainforth2016pitfalls}.
%%\citet{fort2016mcmc} 
%They also show, as we will in Theorem~\ref{the:Repeat}, that the faster $O(1/N+1/M^2)$ convergence 
%rate of~\citet{hong2009estimating} can be achieved  whenever $f$ is continuously differentiable.
%They further provide more precise forms for the error bounds, but these are less tight than those 
%we provide in Theorem~\ref{the:Repeat} and do not apply to cases of repeated nesting.

These convergence rates suggest
that, for most $f$, it is necessary to increase not only the total number of samples, $T$,
but also the number of samples used for each evaluation of the inner estimator, $M$, to achieve convergence.
Further, as we show in Appendix~\ref{sec:bias}, the estimates produced by NMC are, in general, biased.  
This is perhaps easiest to see by noting that as $N\to\infty$, the variance of the
estimator must tend to zero by the law of large numbers, but our bounds remain non-zero for
any finite $M$, implying a bias.

%and that these results hold when the $y_n$ are generated by a valid Markov chain, instead of being i.i.d..

\subsection{Minimum Continuity Requirements}

We next consider the question of what is the minimal requirement on $f$ to ensures some form of
convergence? For a given $y_1$, we
have that $(\hat{\gamma}_M)_1=\frac{1}{M}\sum_{m=1}^{M} \phi(y_1,z_{1,m})\rightarrow\gamma(y_1)$ 
almost surely as $M \rightarrow \infty$, because the left-hand side is a MC estimator. If $f$ is continuous
around $y_1$, this also implies $f(y_1,(\hat{\gamma}_M)_1) \rightarrow
f(y_1,\gamma(y_1))$.  Our candidate requirement is that this holds in
expectation, i.e. that it holds when we incorporate the effect of the outer estimator.
More precisely, we define $(\epsilon_M)_n = \left|f(y_n, (\hat{\gamma}_M)_n) -
f(y_n,\gamma(y_n))\right|$ and require that $\E\left[(\epsilon_M)_1\right] \to 0$ as $M
\to \infty$ (noting that $(\epsilon_M)_n$ are i.i.d. and so
$\E\left[(\epsilon_M)_1\right] = \E\left[(\epsilon_M)_n\right], \forall n\in\N$). Informally, this ``expected continuity''
requirement is weaker than uniform continuity (and much weaker than Lipschitz continuity)
as it allows (potentially infinitely
many) discontinuities in $f$.  More formally we have the following result.
\vspace{-10pt}
\begin{restatable}{theorem}{theConsistent} \label{the:Consistent}
	For $n \in \N$, let 
	\[
	(\epsilon_M)_n = \left|f(y_n, (\hat{\gamma}_M)_n) - f(y_n, \gamma(y_n))\right|.
	\]
  Assume that $\E\left[(\epsilon_M)_1\right] \to 0$  as $M \to \infty$. Let $\Omega$ be
  the sample space of our underlying probability space, so that $I_{\tau_\delta(M),M}$
  forms a mapping from $\Omega$ to $\mathbb{R}$. Then, for every $\delta > 0$,
  there exists a measurable $A_\delta \subseteq \Omega$ with $\mathbb{P}(A_\delta) <
  \delta$, and a function $\tau_\delta : \N \to \N$ such that, for all $\omega\not\in
  A_\delta$,
	\[ 
		I_{\tau_\delta(M),M}(\omega) \asto I\quad\mbox{as}\quad M \to \infty.
	\]
\end{restatable}
\vspace{-10pt}
\begin{proof}
	See Appendix~\ref{sec:app:consistent}.
\end{proof}
\vspace{-6pt}
As well as providing proof of a different form of convergence to any existing results, this
result is particularly important because many, if not most, functions are not Lipschitz
continuous due to the their behavior in the limits.  For example, even the function $f(y,\gamma(y)) = \left(\gamma(y)\right)^2$
is not Lipschitz continuous because the derivative is unbounded as $\left|\gamma(y)\right|\rightarrow\infty$,
whereas the vast majority of problems will satisfy our weaker requirement of $\E\left[(\epsilon_M)_1\right] \to 0$.
%\begin{theorem} \label{the:Consistent}
%  For $n \in \N$, let 
%  \[
%          (\epsilon_M)_n = \left|f(y_n, (\hat{\gamma}_M)_n) - f(y_n, \gamma(y_n))\right|.
%  \]
%  If~~$\E\left[(\epsilon_M)_1\right] \to 0$  as $M \to \infty$, then
%  there exists a $\tau : \N \to \N$ such that $I_{\tau(M),M} \asto I$ as $M \to \infty$.
%\end{theorem}
%%\begin{proof}
%%  See Section~\ref{sec:app:conv-proof} in the Appendices. \todo{Add discussion
%%  characterising some instances of when $\epsilon_M \to 0$}
%%\end{proof}

%
%\subsection{Convergence Rate}
%We now refine this by also establishing the convergence rate as below.
%
%

\subsection{Repeated Nesting and Exact Bounds}

We next consider the case of multiple levels of nesting. As previously explained, 
this case is particularly important for analyzing probabilistic programming
languages.
To formalize what we mean by arbitrary nesting, we first assume some fixed integral depth
$D > 0$, and real-valued functions $f_0, \cdots, f_D$.
We then define
\begin{align*}
  &\gamma_D\left(y^{(0:D-1)}\right) = \E \left[f_D\left(y^{(0:D)}\right) \middle| y^{(0:D-1)}\right] \quad \text{and} \displaybreak[0] \\
  &\gamma_k(y^{(0:k-1)}) = \E \left[f_k\left(y^{(0:k)}, \gamma_{k+1}\left(y^{(0:k)}\right)\right) \middle| y^{(0:k-1)}\right], \displaybreak[0]
\end{align*}
for $0 \leq k < D$, where $y^{(k)} \sim p\left(y^{(k)}|y^{(0:k-1)}\right)$. 
Note that our single nested case corresponds to the setting of $D=1$, $f_0 = f$, $f_1 = \phi$, $y^{(0)}=y$,
$y^{(1)}=z$, $\gamma_0 = I$, and $\gamma_1 = \gamma$. Our goal is to
estimate $\gamma_0 = \E \left[f_0\left(y^{(0)},
\gamma_1\left(y^{(0)}\right)\right)\right]$. To do so we will use the following NMC scheme:
\begin{align*}
  I_D & \left(y^{(0:D-1)}\right) = \frac{1}{N_D} \sum_{n=1}^{N_D} f_D\left(y^{(0:D-1)}, y^{(D)}_n\right) \quad \text{and} \displaybreak[0] \\
   I_k &\left(y^{(0:k-1)}\right) \\
  &= \frac{1}{N_k} \sum_{n=1}^{N_k} f_k\left(y^{(0:k-1)}, y^{(k)}_n, I_{k+1}\left(y^{(0:k-1)}, y^{(k)}_n\right)\right) \displaybreak[0]
\end{align*}
for $0 \leq k \le D-1$, where each $y^{(k)}_n \sim p\left(y^{(k)}|y^{(0:k-1)}\right)$ is drawn
independently. Note that there are multiple values of $y^{(k)}_n$ for each possible $y^{(0:k-1)}$
and that $I_k\left(y^{(0:k-1)}\right)$ is still a random variable given  $y^{(0:k-1)}$.

We are now ready to provide our general result for the convergence bounds that applies to cases of
repeated nesting, provides constant factors (rather than just using big $O$ notation), and
shows how the bound can be improved if the additional assumption of continuous differentiability holds.
%We note that for the particular case of a single nesting without continuous differentiability,
%then the bound is exact is the sense that, in addition to providing the required constant factors,
%there are no higher-order terms that are only dominated asymptotically.
%\begin{restatable}{theorem}{theRepeat} \label{the:Repeat}
%  Assume that $f_0, \cdots, f_D$ are all Lipschitz continuous, and 
%  that 
%  \[
%          f_k\left(y^{(0:k)}, \gamma_{k+1}\left(y^{(0:k)}\right)\right), I_k\left(y^{(0:k-1)}\right) \in L^2 
%  \]
%  for $0 \leq k \leq D$. Then, the mean squared error $\E \left[\left(\gamma_0 - I_0\right)^2\right]$ 
%  converges to $0$ at rate $O\left(\sum_{k=0}^D
%  \frac{1}{N_k}\right)$.
%\end{restatable}
%\vspace{-5pt}
\begin{restatable}{theorem}{theRepeat}
	\label{the:Repeat} %\label{the:biggie}
	If $f_0, \cdots, f_D$ are all Lipschitz continuous in their second input with Lipschitz 
	constants
	\[K_k := \sup_{y^{(0:k)}} \left| \frac{\partial f_k\left(y^{(0:k)},\gamma_{k+1}(y^{(0:k)})\right)}{\partial \gamma_{k+1}}\right|,
	\]
	for all $k \in 0,\dots,D-1$ and if 
	\begin{align*}
	\varsigma_{k}^2  
	&:=\E \left[\left(f_k\left(y^{(0:k)},\gamma_{k+1}
	\left(y^{(0:k)}\right) \right)-\gamma_k\left(y^{(0:k-1)}\right)\right)^2\right] \displaybreak[0]\\
	&< \infty \quad \forall k\in 0,\dots,D
	\end{align*}
	then
	\begin{align}
	\label{eq:bound-lip}
	\hspace{-5pt} \E \left[\left(I_0 - \gamma_0\right)^2\right] \le
	\frac{\varsigma_{0}^2}{N_0} +
	\sum_{k=1}^{D} \left(\prod_{\ell=0}^{k-1} K_{\ell}^2\right)
	\frac{\varsigma_{k}^2}{N_{k}}+ O(\epsilon)
	\end{align}
	where $O(\epsilon)$ represents asymptotically dominated terms.
	
	If $f_0, \cdots, f_D$ are also continuously differentiable with 
	second derivative bounds 
	\[
	C_k := \sup_{y^{(0:k)}} \left|\frac{\partial^2 f_k\left(y^{(0:k)},\gamma_{k+1}(y^{(0:k)})\right)}{\partial \gamma^2_{k+1}}\right|
	\]
	then this mean square error bound can be tightened to
	\begin{align}
	\begin{split}
	\label{eq:bound-cont}
	\E &\left[\left(I_0 - \gamma_0\right)^2\right] \le 
	\frac{\varsigma_0^2}{N_0}+ \\
	&\left(
	\frac{C_0 \varsigma_{1}^2}{2 N_{1}}
	+\sum_{k=0}^{D-2}  \left(\prod_{d=0}^{k} K_{d}\right)
	\frac{C_{k+1} \varsigma^2_{k+2}}{2 N_{k+2}}
	\right)^2 + O(\epsilon).
	\end{split}
	\end{align}
	For a single nesting, we can further characterize $O(\epsilon)$ giving
	\begin{align}
	\E  &\left[\left(I_0 - \gamma_0\right)^2\right]  \le \frac{\varsigma^2_0}{N_0}+\frac{4 K_{0}^2 \varsigma_1^2}{N_0 N_{1}}
	+\frac{2 K_{0}\varsigma_{0} \varsigma_1}{N_{0} \sqrt{N_1}}+\frac{K_0 ^2 \varsigma_1^2}{N_1}
	\end{align}
	\vspace{-5pt}
	\begin{align}
	\begin{split}
	\E &\left[\left(I_0 - \gamma_0\right)^2\right]  \le \frac{\varsigma^2_0}{N_0}+\frac{C_0 ^2 \varsigma_1^4}{4 N_1^2}\left(1+\frac{1}{N_0}\right)\\
	&+\frac{K_{0}^2 \varsigma_1^2}{N_0 N_{1}}
	+\frac{2 K_{0}\varsigma_1}{N_{0} \sqrt{N_1}}\sqrt{\varsigma_0^2+\frac{C_0 ^2 \varsigma_1^4}{4 N_1^2}}
	+ O\left(\frac{1}{N_1^3}\right)
	\end{split}
	\label{eq:cont-single}
	\end{align}
	for when the continuous differentiability assumption does not hold and 
	holds respectively.
\end{restatable}
\vspace{-10pt}
\begin{proof}
	See Appendix~\ref{sec:app:repeat}.
\end{proof}
\vspace{-6pt}
These results give a convergence rate of $O(\sum_{k=0}^{D} 1/N_k)$
when only Lipschitz continuity holds and $O(1/N_0 +(\sum_{k=1}^{D} 1/N_k)^2)$ when
all the $f_k$ are also continuously differentiable.  
As estimation requires drawing $O(T)$ samples where $T= \prod_{k=0}^{D} N_k$, 
the convergence rate will rapidly diminish with repeated nesting.  More precisely,
as shown in Appendix~\ref{sec:app:opt-conv}, the optimal convergence rates are
$O(1/T^{\frac{1}{D+1}})$ and $O(1/T^{\frac{2}{D+2}})$ respectively for the two cases, both of which
imply that the rate diminishes exponentially with $D$.

%\input{result}
% !TEX root =  main.tex

\section{Special Cases}
\label{sec:special_cases}

We now outline some special cases where it is possible to achieve a
convergence rate of $O(1/N)$ in the mean square error (MSE) as per conventional MC estimation.  
Establishing these cases is important because it identifies for which problems we can use conventional results,
when we can achieve an improved convergence rate, and what precautions we must take to ensure this.
We will focus on single nesting instances, but note that all results still apply to repeated nesting scenarios because they can
be used to ``collapse'' layers and thereby reduce the depth of the nesting.

\subsection{Linear $f$}
\label{sec:linear_case}

Our first special case is that $f$ is linear in its second argument, i.e. $f(y,\alpha v + \beta w) = 
\alpha f(y,v)+ \beta f(y,w)$.
Here the problem can be rearranged to a single expectation, a well-known result which forms the basis for pseudo-marginal, 
nested sequential MC \citep{andersson2015nested}, and certain ABC methods~\citep{csillery2010approximate}. 
Namely we have
\begin{align}
I
& = \mathbb{E}_{y \sim p(y)}\left[f\left(y,\mathbb{E}_{z\sim p(z|y)}\left[\phi(y,z)\right]\right)\right] \displaybreak[0] \nonumber \\
& = \mathbb{E}_{y \sim p(y)}\left[ \mathbb{E}_{z\sim p(z|y)}\left[f(y,\phi(y,z))\right]\right] \displaybreak[0] \nonumber \\
& \approx\frac{1}{N} \sum_{n=1}^{N} f(y_n,\phi(y_n,z_n))
\end{align}
where $(y_n, z_n) \sim p(y)p(z|y)$ if $\gamma(y)$ is of the form of~\eqref{eq:gamma_1} and
$y_n \sim p(y)$ and $z_n \sim p(z)$ are independently drawn if $\gamma(y)$ is of the form of~\eqref{eq:gamma_2}.

% or the use of nested queries \cite{goodman2008church,rainforth2016nips} in probabilistic
% programming when the outer query does not depend on any nonlinear functional mapping of
% the marginal probability from the inner query.  
%In these scenarios, nesting provides a convenient means of expressing the problem, but is
%not a fundamental component of its solution.

%Another important case occurs when $\gamma(y)$ is independent of $y$. \todo{Isn't this
%just saying $\gamma$ is constant?}  In this case, we can estimate $\tilde{\gamma} =
%\gamma(y)$ and $(I(f) \,|\, \tilde{\gamma})$ separately, each converging at
%$O(1/\sqrt{N})$, presuming that the same number of samples $N$ is used for both. Provided
%that $f(y,v)$ is Lipschitz continuous in $v$, the root mean squared error of the
%estimation for $I(f)$ will also have same convergence rate.
%\footnote{The proof for this
%follows the same lines as used to link the error in $\gamma$ to $f$ in
%Theorem~\ref{the:Consistent}.} \todo{Do we want this footnote?}

% !TEX root =  main.tex

\subsection{Finite Possible Realizations of $y$}
\label{sec:discrete}

Our second case is if $y$ must take one of finitely many values $y_1, \cdots, y_C$, then it is possible to
use another approach to ensure the same convergence rate as standard MC.
The key observation is to note that in this case we can convert the nested problem
\eqref{eq:MC} into $C$ separate non-nested problems
\begin{align}
	\label{eq:disc-I}
         I = \sum_{c=1}^C P(y = y_c) \, f(y_c, \gamma(y_c))
\end{align}
which can then be estimated using
\begin{align}
        I_N  &= \sum_{c=1}^C (\hat{P}_N)_c \, (\hat{f}_N)_c \quad \text{where}
        \label{eq:IN} \displaybreak[0] \\
        \label{eq:PN}
        P(y = y_c) &\approx (\hat{P}_N)_c = \frac{1}{N} \sum_{n=1}^N \mathbbm{1}(y_n = y_c)  \displaybreak[0] \\
        \label{eq:fn}
        \hspace{-5pt} f(y_c, \gamma(y_c)) &\approx (\hat{f}_N)_c = f\left(y_c, \frac{1}{N} \sum_{n=1}^N \phi(y_c, z_{n,c})\right)
\end{align}
with $y_n \iid p(y)$ and $z_{n,c} \sim p(z|y_c)$ (or $z_{n,c} \sim p(z)$  
if using the formulation in \eqref{eq:gamma_2}).  Note the critical point that
each $z_{n,c}$ is independent of $y_n$ as each $y_c$ is a constant.
We can now show the following result which, though intuitively straightforward, 
requires care to formally prove.
%Note that the same set of $y_n$'s is used for every
%$(\hat{P}_N)_c$, calculating $I_N$ requires $N$ samples of $y$, 
%$CN$ samples of $z$ (each
%drawn independently to the samples of $y$), $CN$ evaluations of $\phi$, and $C$ evaluations of $f$.
\begin{restatable}{theorem}{thefiniteres}
	\label{the:finite-res}
  If $f$ is Lipschitz continuous, then the mean squared error of $I_N  = \sum_{c=1}^C (\hat{P}_N)_c \, (\hat{f}_N)_c$ 
  as an estimator for $I$ as per~\eqref{eq:disc-I} converges at rate $O(1/N)$.
\end{restatable}
\vspace{-12pt}
\begin{proof}
	See Appendix~\ref{sec:app:finite-res}.
\end{proof}

\subsection{Products of Expectations}
\label{sec:products}

We next consider the scenario, which occurs for many latent variables models and probabilistic
programming problems, where $\gamma (y)$ is equal to the product of 
multiple expectations, rather than just a single expectation as per~\eqref{eq:gamma_1}.
That is,
\begin{equation}
\label{eq:prod-mc}
I = \mathbb{E}_{y\sim p(y)}\left[ f\left(y,\prod_{\ell=1}^{L}\mathbb{E}_{z_{\ell}\sim p(z_{\ell}|y)}\left[\psi_{\ell}(y,z_{\ell})\right]\right) \right].
\end{equation}
Because the $z_{\ell}$ will not in general be independent, we cannot trivially
rearrange~\eqref{eq:prod-mc} to a standard nested estimation by moving
the product within the expectation.
Our insight is that the required rearrangement can instead be achieved
by introducing new random variables $\{z'_{\ell}\}_{\ell=1:L}$ such that each 
$z'_{\ell}|y \sim p(z_\ell | y)$ and the $z'_{\ell}$ are independent of one another. 
This can be achieved by, for example, taking $L$ independent samples from the joint $Z_{\ell} \iid p(z_{1:L} | y)$ 
and using the $\ell^{\mathrm{th}}$ such draw for the $\ell^{\mathrm{th}}$ dimension of $z'$, i.e.
setting $z'_{\ell}= \{Z_{\ell}\}_{\ell}$.
For every $y \in \mathcal{Y}$ we now have
%Consider the set of random variables $\{Z_{\ell}\}_{\ell=1:L}$  such that each $Z_{\ell} \sim p(z_\ell | y)$,
%noting that by definition, the $Z_{\ell}$ are independent of one another.
%Now define the set of variables  $\{z'_{\ell}\}_{\ell=1:L}$ such that $z'_{\ell}= \{Z_{\ell}\}_{\ell}$.
%We can now see that each $z'_{\ell} | y \sim p(z_{\ell} | y)$, but unlike $z_{\ell}$, the $z'_{\ell}$ must be independent of one another and thus
\begin{align}
 \prod_{\ell=1}^{L}&\mathbb{E}_{z_{\ell}\sim p(z_{\ell}|y)}[\psi_{\ell}(y,z_{\ell})] 
 =  \prod_{\ell=1}^{L}\mathbb{E}_{z'_{\ell}\sim p(z'_{\ell}|y)}[\psi_{\ell}(y,z'_{\ell})] \displaybreak[0] \nonumber \\
  &= \mathbb{E}_{\{z'_{\ell}\}_{\ell=1:L} \sim p(\{z'_{\ell}\}_{\ell=1:L}|y)}\left[ \prod_{\ell=1}^{L} \psi_{\ell}(y,z'_{\ell}) \right]
\end{align}
which is a single expectation on an extended space and shows that \eqref{eq:prod-mc} 
fits the NMC formulation.
Furthermore, we can now show that if $f$ is linear, the MSE of the NMC estimator~\eqref{eq:prod-mc}
converges at the standard MC rate $O(1/N)$, provided that $M$ remains fixed.  
%More formally we can state the following result.

\begin{restatable}{theorem}{theprod}
	\label{the:prod}
	Consider the NMC estimator
	\begin{align*}
	I_{N} &= \frac{1}{N}\sum_{n=1}^N f\left(y_n,\prod_{\ell=1}^{L} \frac{1}{M_{\ell}} \sum_{m=1}^{M_{\ell}} \psi_{\ell}(y_n,z_{n,\ell,m}')\right)
	%&= \frac{1}{N}\sum_{n=1}^N \prod_{\ell=1}^{L} \frac{1}{M_{\ell}} \sum_{m=1}^{M_{\ell}} f(y_n,\psi_{\ell}(y_n,z_{n,\ell,m}')) \\
	\end{align*}
	where each $y_n \in \mathcal{Y}$ and $z_{n,\ell,m}' \in \mathcal{Z}_{\ell}$ are independently drawn from 
	$y_n \sim p(y)$ and $z_{n,\ell,m}' | y_n \sim p(z_{\ell} | y_n)$, respectively. If $f$ is linear, 
        the estimator converges almost surely to $I$,
	with a convergence rate of $O(1/N)$ in the mean square error for any fixed choice of $\{M_{\ell}\}_{\ell = 1:L}$.
\end{restatable}
\vspace{-12pt}
\begin{proof}
See Appendix~\ref{sec:app-prod}.
\end{proof}
\vspace{-8pt}

As this result holds in the case $L=1$, an important consequence is that
whenever $f$ is linear, the same convergence rate is achieved regardless of whether we
reformulate the problem to a single expectation or not, provided that the number of samples
used by the inner estimator is fixed.  
%Conversely, as we will show in Section~\ref{sec:convergence},
%if $f$ is nonlinear, it will be necessary for the number of samples in \emph{both} the inner and the 
%outer estimators to increase in order to achieve the convergence of~\eqref{eq:nested-mc}. 

\begin{figure*}[t!]
	\centering
	\begin{subfigure}[b]{0.49\textwidth}
		\centering
		\includegraphics[width=0.9\textwidth,trim={1.5cm 0 3.5cm 0},clip]{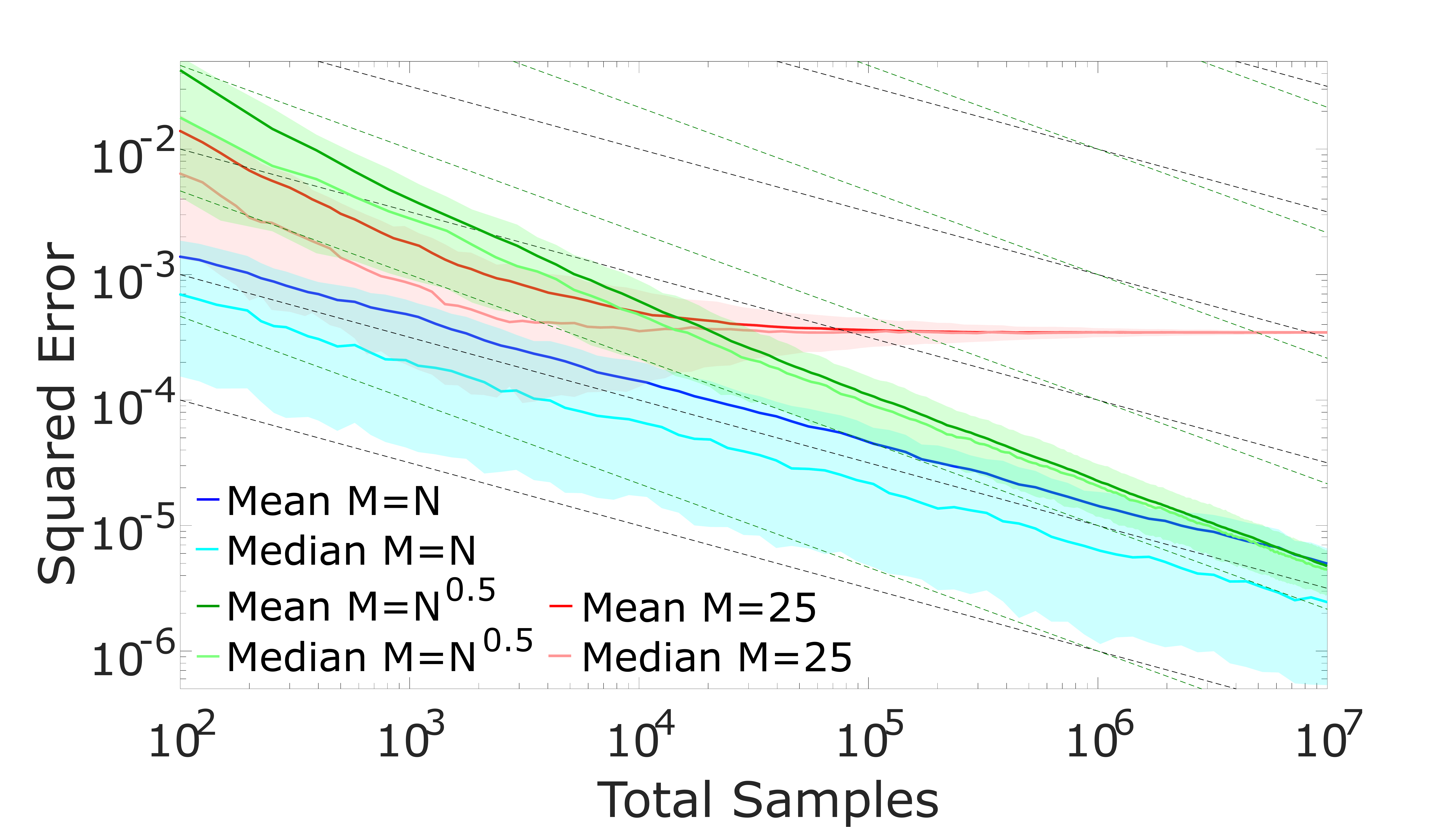}
		\caption{Convergence with increasing sample budget. \label{fig:emprical-conv}}
	\end{subfigure}
	\begin{subfigure}[b]{0.49\textwidth}
		\centering
		\includegraphics[width=0.9\textwidth,trim={1.5cm 0 3.5cm 0},clip]{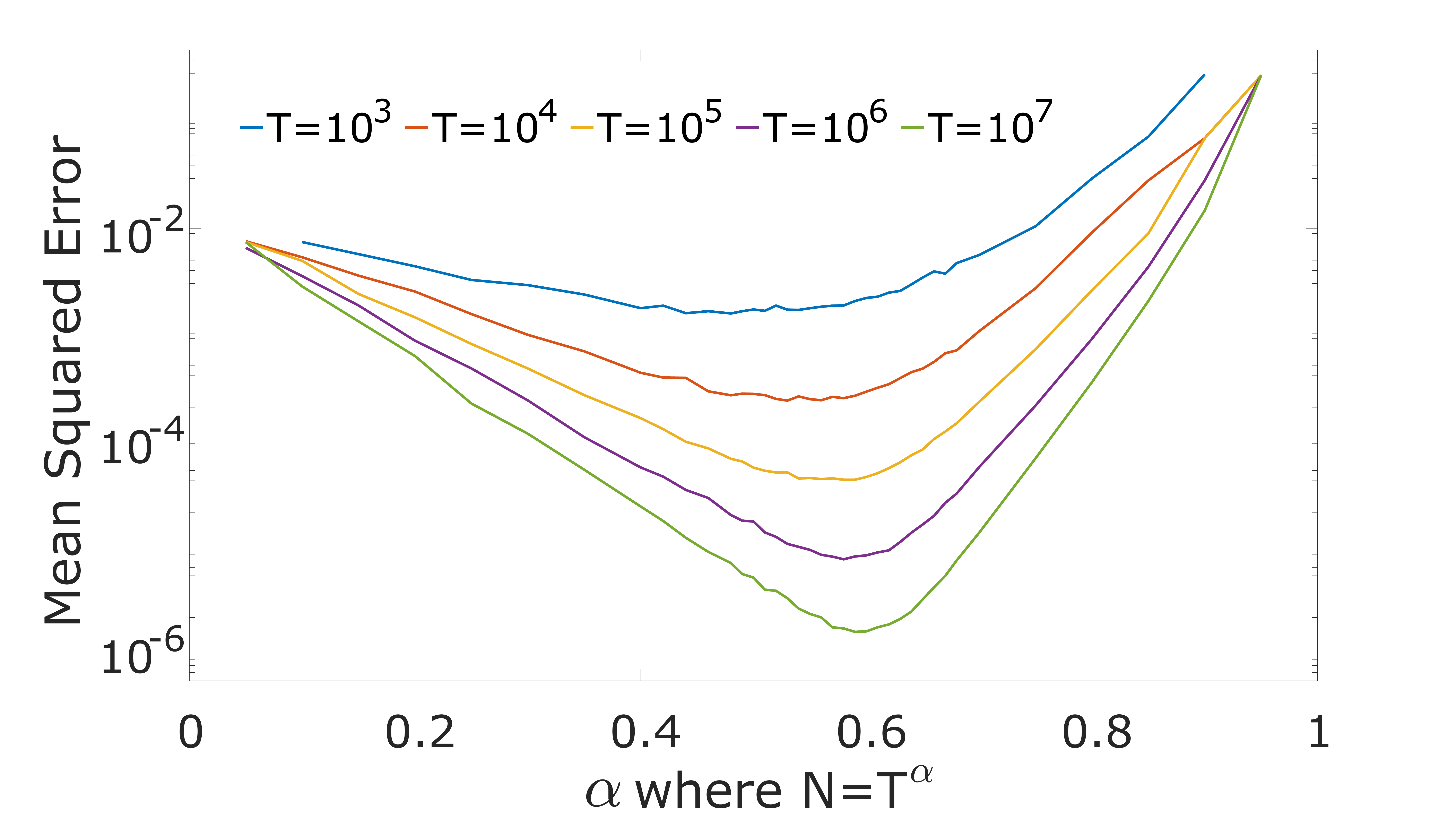}
		\caption{Final error for different $T$ and $N$.\label{fig:tau_sweep}}
	\end{subfigure}
	\vspace{-6pt}
	\caption{Empirical convergence of NMC for~\eqref{eq:model}.  [Left] 
		convergence in total samples for different ways of setting $M$ and $N$.  
		Results are averaged over 1000 independent runs, while shaded regions give the 25\%-75\% quantiles. The theoretical convergence rates, namely
		$O(1/\sqrt{T})$ and $O(1/T^{2/3})$ for setting $N\propto M$ and $N\propto M^2$
		respectively, are observed (see the dashed black and green lines respectively for reference).
		The fixed $M$ case converges at the standard MC error rate of $O(1/T)$ but to a 
		biased solution.
		[Right] final error for different total sample budgets
		as a function of $\alpha$ where $N=T^{\alpha}$ and $M=T^{1-\alpha}$ iterations are used for the outer
		and inner estimators respectively.  This shows that even though $\alpha=\frac{2}{3}$ is the
		asymptotically optimal allocation strategy, this is not the optimal solution for
		finite $T$. Nonetheless, as $T$ increases, the optimum value of $\alpha$ increases,
		starting around $0.5$ for $T=10^3$ and reaching around $0.6$ for $T=10^7$. \vspace{-20pt}}
\end{figure*}	

\subsection{Polynomial $f$}
\label{sec:polynomial}

Perhaps surprisingly, whenever $f$ is of the form
\begin{align}
  f(y,\gamma(y)) &= g(y) \, \gamma(y)^{\alpha}
\end{align}
where $\alpha \in \mathbb{Z}_{\geq 0}$,
then it is also possible to construct a standard MC estimator by building on the ideas
introduced in Section~\ref{sec:products} and those of \cite{goda2016computing}.
The key idea is
\begin{align}
\left(\E \left[z\right]\right)^2 = \E \left[z\right]\E \left[z'\right] = \E \left[z z'\right]
\end{align}
where $z$ and $z'$ are i.i.d.  Therefore, assuming appropriate integrability requirements,
we can construct the following non-nested MC estimator:
\begin{align}
&\E\left[g(y) \, \gamma(y)^{\alpha}\right] = 
 \E\left[g(y) \, 
\prod_{\ell=1}^{\alpha}
\E_{z_{\ell}\sim p(z | y)}\left[\phi(y,z_{\ell}) | y\right]\right] \nonumber \displaybreak[0] \\
% &\quad\quad=\sum_{j=1}^{J} \E\left[g_j(y) \E\left[\prod_{\ell=1}^{\alpha_j} \phi(y,z_{j,\ell}) \middle|y\right]\right] \nonumber \displaybreak[0] \\
 &=\E\left[g(y)\prod_{\ell=1}^{\alpha} \phi(y,z_{\ell})\right] \nonumber \displaybreak[0]
\approx \frac{1}{N} \sum_{n=1}^{N} g(y_n)\prod_{\ell=1}^{\alpha} \phi(y_n,z_{n,\ell})
\end{align}
where we independently draw each $z_{n,\ell} | y_n \sim p(z | y_n)$.

% !TEX root =  main.tex

\section{Empirical Verification}
\label{sec:empirical}

The convergence rates proven in Section~\ref{sec:convergence} are only
\emph{upper bounds} on the worst-case performance. We will now
examine whether these convergence rates are tight in practice, investigate what happens
when our guidelines are not followed, and outline some applications of our results.

\subsection{Simple Analytic Model}
\label{sec:simple}

We start with the following analytically calculable problem
\begin{subequations}
\label{eq:model}
\begin{align}
 y &\sim \mathrm{Uniform}(-1,1), \\
z &\sim \mathcal{N}(0,1), \\
\phi(y,z) & = \sqrt{2/\pi}\exp\left(-2(y-z)^2\right), \\
f(y,\gamma(y)) & = \log (\gamma (y)) = \log(\E_{z}[\phi(y,z)]).
\end{align}
\end{subequations}
for which $I=\frac{1}{2}\log\left(\frac{2}{5\pi}\right)-\frac{2}{15}$.
Figure~\ref{fig:emprical-conv} shows the corresponding empirical convergence obtained by
applying~\eqref{eq:nested-mc} to~\eqref{eq:model} directly. It shows that, for this
problem, the theoretical convergence rates from Theorem~\ref{the:Repeat} are indeed
realized.
%\footnote{Note the continuous differentiability assumption is satisfied as $\gamma(y)\neq0$
%or $\infty$ for any $y\in[-1,1]$.}
The figure also demonstrates the danger of not increasing
$M$ with $N$, showing that the NMC estimator converges to an incorrect solution when $M$
is held constant.  Figure~\ref{fig:tau_sweep} shows the effect of varying $N$ and $M$ for various
fixed sample budgets $T$ and demonstrates that the asymptotically optimal strategy can be suboptimal
for finite budgets.

\subsection{Planning Cancer Treatment}
\label{sec:cancer}

We now introduce a real-world example to show the applicability of NMC in a scenario
where the solution is not analytically tractable and conventional MC is insufficient.
Consider a treatment center assessing a new policy for planning cancer treatments, subject to a budget. 
Clinicians must decide on a patient-by-patient basis whether to administer chemotherapy in the
hope that their tumor will reduce in size sufficiently to be able to perform surgery at a later date.
A treatment is considered to have been successful if the size of the tumor drops below a threshold value in a fixed time window.
The clinicians have at their disposal a simulator for the evolution of tumors with time,
parameterized by both observable values, $y$, such as tumor size, and unobservable values, $z$, such as the patient-specific response to treatment.
Given a set of input parameters, the simulator deterministically returns a binary response $\phi(y,z)\in\left\lbrace 0,1\right\rbrace $, with $1$ indicating a successful treatment.
To estimate the probability of a successful treatment for a given patient, the clinician must calculate the expected
success over these unobserved variables, namely $\E_{z\sim p(z|y)} [\phi(y,z)]$ where $p(z|y)$ represents a probabilistic
model for the unobserved variables, which could, for example, be constructed based on empirical data.
The clinician then decides whether to go ahead with the treatment for that
patient based on whether the calculated probability of success exceeds a certain threshold $T_{\mathrm{treat}}$.

\begin{figure}[t]
	\centering
	\includegraphics[width=0.45\textwidth,trim={1.5cm 0 3.5cm 0},clip]{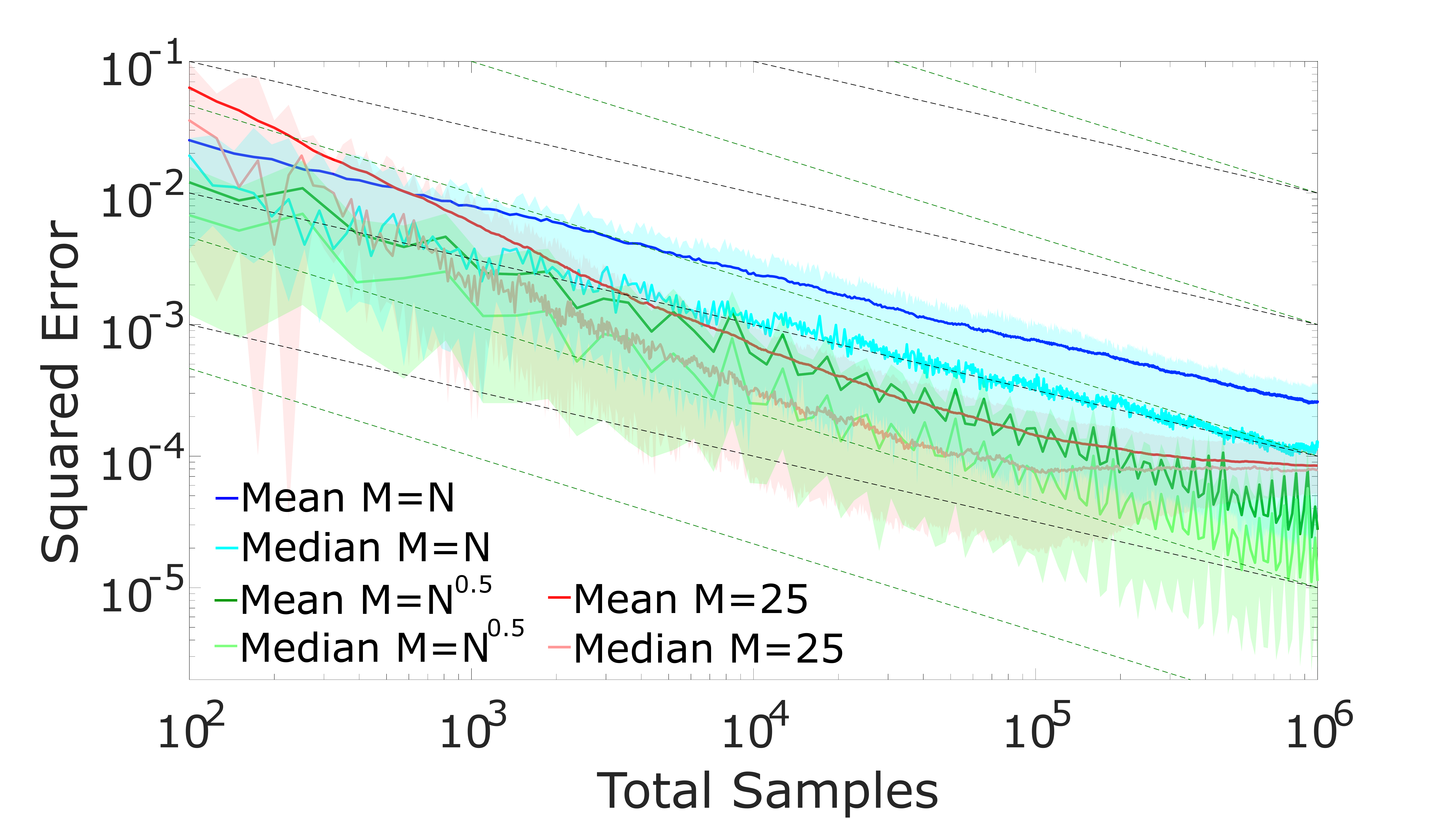}
	\caption{Convergence of NMC for cancer simulation.
		A ground truth estimate was calculated
		using a single run with $M=10^5$ and $N=10^5$.  Experimental setup and conventions are as per Figure~\ref{fig:emprical-conv}
		and we again observe the expected convergence rates. When $M=\sqrt{N}$ an interesting fluctuation behavior is observed.  
		Further testing suggests that this originates because the bias of the estimator depends in
		a fluctuating manner on the value of $M$ as the binary output of $\phi(y,z)$ creates a quantization
		effect on the possible estimates for $\hat{\gamma}$.  This effect is also observed for the $M=N$ case
		but is less pronounced. \vspace{-20pt} \label{fig:emperical-conv-cancer}}
\end{figure}

The treatment center wishes to estimate the expected number of patients that will be treated for a given $T_{\mathrm{treat}}$ so that it can minimize this threshold without exceeding its budget.
To do this, it calculates the expectation of the clinician's decisions to administer 
treatment, giving the complete nested expectation for calculating the number of treated patients as
\begin{equation}
	\label{eq:cancer}
I(T_{\mathrm{treat}}) = \E_{} \left[\mathbb{I}\left(\E_{z\sim p(z|y)} [\phi(y,z)]>T_{\mathrm{treat}}\right)\right],
\end{equation}
where the step function $\mathbb{I}(\cdot > T_{\mathrm{treat}})$ imposes a non-linear
mapping, preventing conventional MC estimation. Full details on $\phi$, $p(y)$, and $p(z|y)$ are 
given in Appendix~\ref{sec:cancer_sim_app}.

To verify the convergence rate, we repeated the analysis from Section~\ref{sec:simple} for~\eqref{eq:cancer} 
at a fixed value of $T_{\mathrm{treat}}=0.35$. 
The results, shown in Figure~\ref{fig:emperical-conv-cancer}, again verify the theoretical rates. 
By further testing different values of $T_{\mathrm{treat}}$, we found $T_{\mathrm{treat}} = 0.125$ 
%(i.e. go ahead with treatment if the probability of success is 0.125 or greater) 
to be optimal under the budget.

\begin{figure}[t]
	\centering
	\includegraphics[width=0.45\textwidth,trim={1.5cm 0 3.5cm 0},clip]{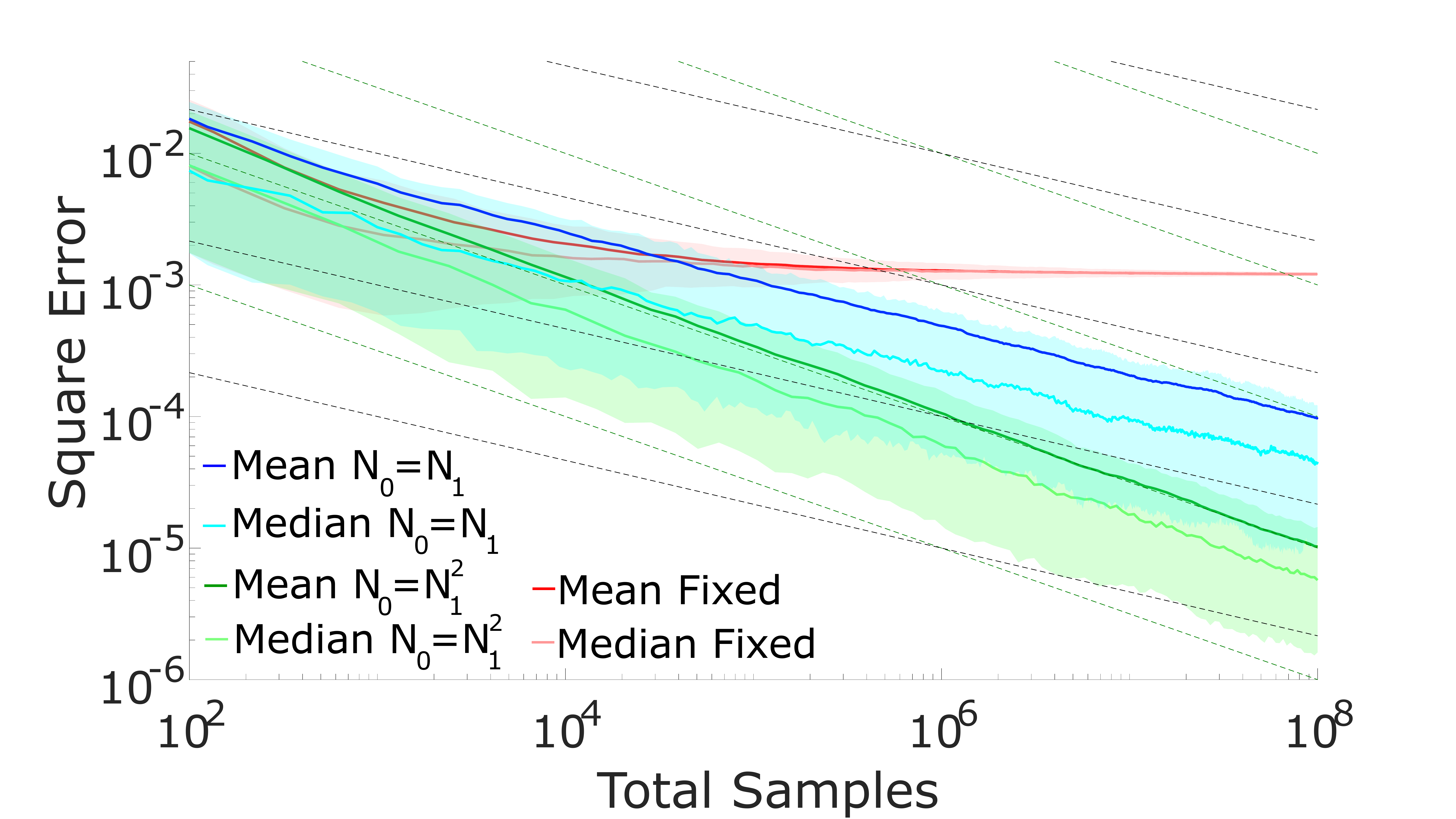}
	\vspace{-5pt}
	\caption{Empirical convergence of NMC to~\eqref{eq:repeat-nest} for an increasing total sample budget
		$T=N_0 N_1 N_2$.  Setup and conventions as per Figure~\ref{fig:emprical-conv}.
		Shown in red is the convergence with a fixed $N_2=5$ and $N_0=N_1^2$, which we see gives a biased solution. Shown in blue is the convergence
		when setting $N_0=N_1=N_2$, which we see converges at the expected $O(T^{-1/3})$ rate.  Shown in green is the convergence when
		setting $N_0=N_1^2=N_2^2$ which we see again gives the theoretical convergence rate, namely $O(T^{-1/2})$.
		\vspace{-5pt}\label{fig:multi-nest}}
\end{figure}

\vspace{-4pt}

% !TEX root =  main.tex

\subsection{Repeated Nesting}
\label{sec:exp-repeat-app}

We next consider some simple models with multiple levels of nesting, starting with
\begin{subequations}
	\label{eq:repeat-nest}
\begin{align}
y^{(0)} \sim \mathrm{Uniform}(0,1), \;&\;
y^{(1)} \sim \mathcal{N}(0,1), \;\;
y^{(2)} \sim \mathcal{N}(0,1), \displaybreak[0] \nonumber\\ 
f_0 \left(y^{(0)}, \gamma_1\left(y^{(0)}\right)\right)&= \log \gamma_1\left(y^{(0)}\right) \displaybreak[0] \\ 
\begin{split}
f_1 \left(y^{(0:1)}, \gamma_2\left(y^{(0:1)}\right)\right)&= \\
\exp\Bigg(-\frac{1}{2}\bigg(y^{(0)}- &y^{(1)}-\log \gamma_2\left(y^{(0:1)}\right)\bigg)\Bigg)  
\end{split} \displaybreak[0]  \\
f_2 \left(y^{(0:2)}\right)=&\exp\left(y^{(2)}-\frac{y^{(0)}+y^{(1)}}{2}\right)
\end{align}
\end{subequations}
which has analytic solution $I=-3/32$. 
The convergence plot shown in Figure~\ref{fig:multi-nest} demonstrates that the
theoretically expected convergence
%\footnote{Strictly speaking the assumptions of Theorem~\ref{fig:multi-nest} are
%	not actually satisfied for this model and the later modifications
%	 because $K_1 = C_1 = \infty$.  However, we see convergence behavior
%	as if the assumptions were satisfied.  This highlights the importance of Theorem~\ref{the:Consistent}
%	because $f_1$ does satisfy this weaker assumption.}
behaviors are observed for different methods of setting 
$N_0, N_1$, and $N_2$.

We further investigated the empirical performance of different strategies for choosing $N_0, N_1$, $N_2$ under a 
finite fixed budget $T=N_0N_1N_2$.  In particular, we looked to establish the optimal empirical
setting under the fixed budget $T=10^6$ for the model described in~\eqref{eq:repeat-nest} and 
a slight variation where $y^{(0)}$ is replaced with $y^{(0)}/10$, for which the ground
truth is now $I=39/160$.
Defining $\alpha_1$ and $\alpha_2$ such that $N_0 = T^{\alpha_1}$, 
$N_1 = T^{\alpha_2(1-\alpha_1)}$, and $N_2 = T^{(1-\alpha_1)(1-\alpha_2)}$, we ran a Bayesian optimization
algorithm, namely BOPP~\citep{rainforth2016nips}, to optimize the log MSE,
$\log_{10} \left(\E \left[(I_0(\alpha_1,\alpha_2)-\gamma_0)^2\right]\right)$, with respect to 
$(\alpha_1,\alpha_2)$.  For each tested $(\alpha_1,\alpha_2)$, the MSE was estimated using 
$1000$ independently generated samples of $I_0$ and we allowed a total of $200$ such tests.
%To investigate how sensitive the empirically optimal 
%allocation is to the problem at hand, we considered both the model described in~\eqref{eq:repeat-nest} and 
%a slight variation where $y^{(0)}$ is replaced with $y^{(0)}/10$, for which the ground
%truth is now $I=39/160$.
We found respective optimal values for $(\alpha_1,\alpha_2)$ of  $(0.53,0.36)$ and $(0.38,0.45)$.
%\footnote{Note that there is inevitably a small bias introduced
%	by the fact that each $N_k$ has to be a whole number.}  
By comparison, the asymptotically optimal
setup suggested by our theoretical results is $(0.5,0.5)$,
showing that the finite budget optimal allocation can vary significantly
from the asymptotically optimal solution and that it does so in a problem dependent manner.
% We also see that the relative optimally can change quite significantly with
%changes to the target problem.  
%For example, the second modification meant it was preferable to run less iterations of
%the outer estimator, perhaps because it reduced the relative importance of $y^{(0)}$ compared to the other variables
%in $f_1$ and $f_2$.  
%Further work is required to
%establish for concrete empirical guidelines or to investigate whether methods that adaptively allocate the number of
%%samples might be feasible.  In the meantime, the asymptotically optimal choice (presuming continuous differentiability) 
%of $N_0 \propto N_1^2 \propto \dots \propto N_D^2$ seems to a reasonable practical choice on average.

As a byproduct, BOPP also produced Gaussian process approximations to the log MSE variations,
as shown in Figure~\ref{fig:multi-tau}.  We see that the two problems lead to distinct performance variations.
  Based on the (unshown) uncertainty estimates of these Gaussian processes, we believe these 
  approximations are a close representation of the truth.

\vspace{-4pt}

\section{Applications}
\label{sec:app}

\subsection{Bayesian Experimental Design}
\label{sec:bo-design}

In this section, we show how our results can be used to derive an improved estimator for the problem of
Bayesian experimental design (BED) in the case where the experiment outputs are discrete.  
A summary of our approach is provided here, 
%Due to space restrictions, we provide only a summary here, 
with full details
%a full introduction to BED, the derivation of all
%the introduced estimators, and additional results 
provided in Appendix~\ref{sec:exp-design}.

Bayesian experimental design provides a framework for designing experiments in a manner that is optimal from
an information-theoretic viewpoint \citep{chaloner1995bayesian,sebastiani2000maximum}.  Given a prior
$p(\theta)$ on parameters $\theta$ and a corresponding likelihood $p(y|\theta,d)$ for experiment 
outcomes $y$ given a design $d$, the Bayesian optimal design $d^*$ is given by maximizing the mutual
information between $\theta$ and $y$ defined as follows
\begin{align}
\bar{U}(d)=
\int_{\mathcal{Y}}\int_{\Theta} p(y,\theta | d)\log\left(\frac{p(\theta |y, d)}{p(\theta)}\right)d\theta dy. 
\label{eq:u_bar_main}
\end{align}
Estimating $d^*$ is challenging as $p(\theta |y, d)$ is rarely known in closed-form.  
However, appropriate algebraic manipulation shows that~\eqref{eq:u_bar_main}
is consistently estimated by
\begin{align}
\label{eq:exp-des-nmc-main}
\begin{split}
\hat{U}_{\text{NMC}}(d) =
\frac{1}{N} &\sum_{n=1}^{N} \Bigg[ \log(p(y_n | \theta_{n,0},d))  \\
&- \log \left(\frac{1}{M} \sum_{m=1}^{M}p(y_n | \theta_{n,m},d)\right) \Bigg]
\end{split}
\end{align}
where $\theta_{n,m} \sim p(\theta)$ for each $(m,n) \in \{0,\ldots,M\}\times \{1,\ldots,N\}$, 
and $y_n \sim p(y|\theta=\theta_{n,0}, d)$ for each $n \in \{1,\ldots,N\}$.  
This na\"{i}ve NMC estimator has been implicitly used by \cite{myung2013tutorial} amongst others and gives a convergence
rate of $O(1/N+1/M^2)$ as per Theorem~\ref{the:Repeat}.

\begin{figure}[t]
	\centering
	\begin{subfigure}[b]{0.23\textwidth}
		\centering
		\includegraphics[height=0.88\textwidth]{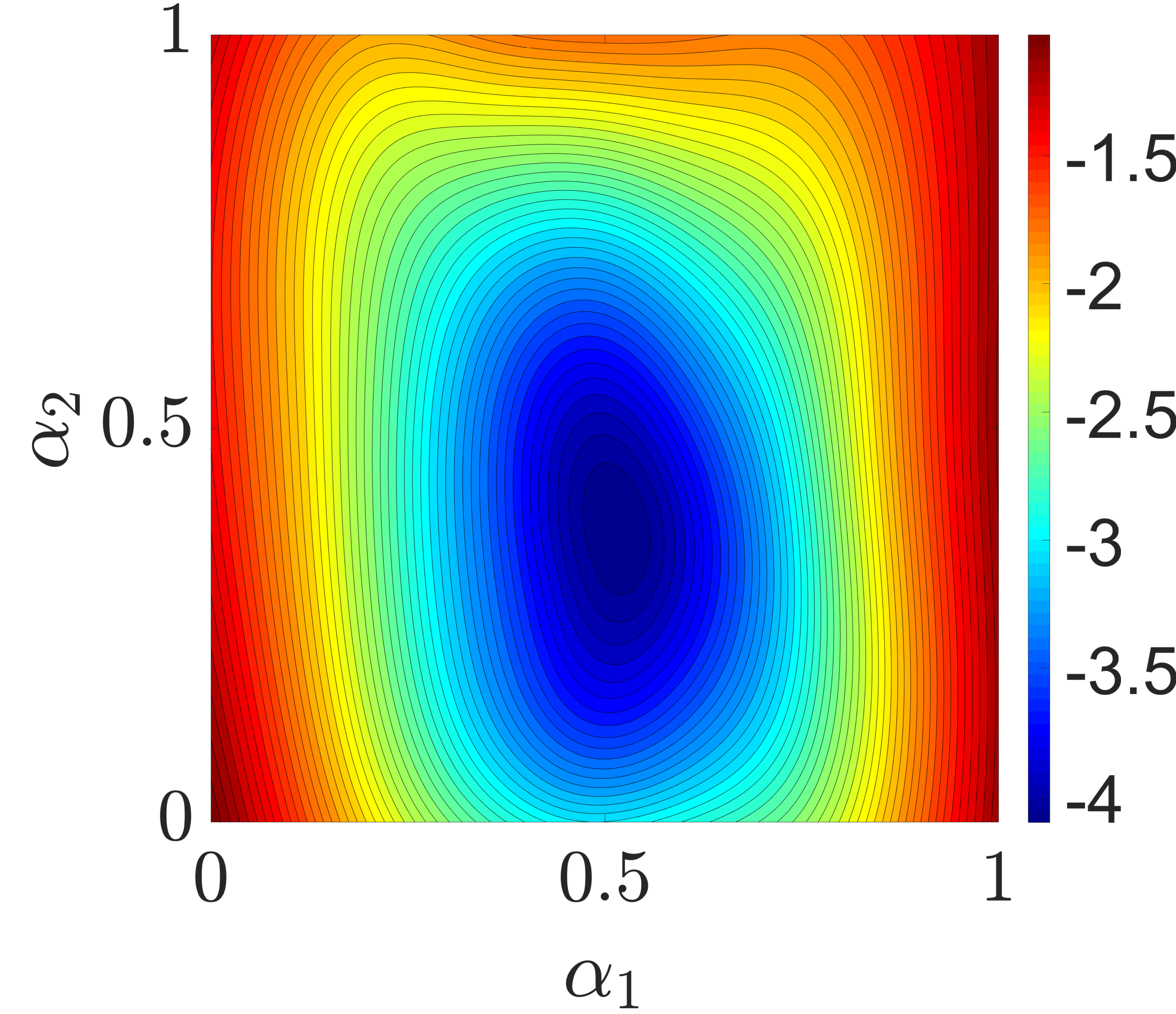}
		\caption{Original}
	\end{subfigure}
	%	~
	%	\begin{subfigure}[b]{0.23\textwidth}
	%		\centering
	%		\includegraphics[height=0.88\textwidth]{tmax_1e6_model_2_contour_plot.pdf}
	%		\caption{First modification}
	%	\end{subfigure}
	~
	\begin{subfigure}[b]{0.23\textwidth}
		\centering
		\includegraphics[height=0.88\textwidth]{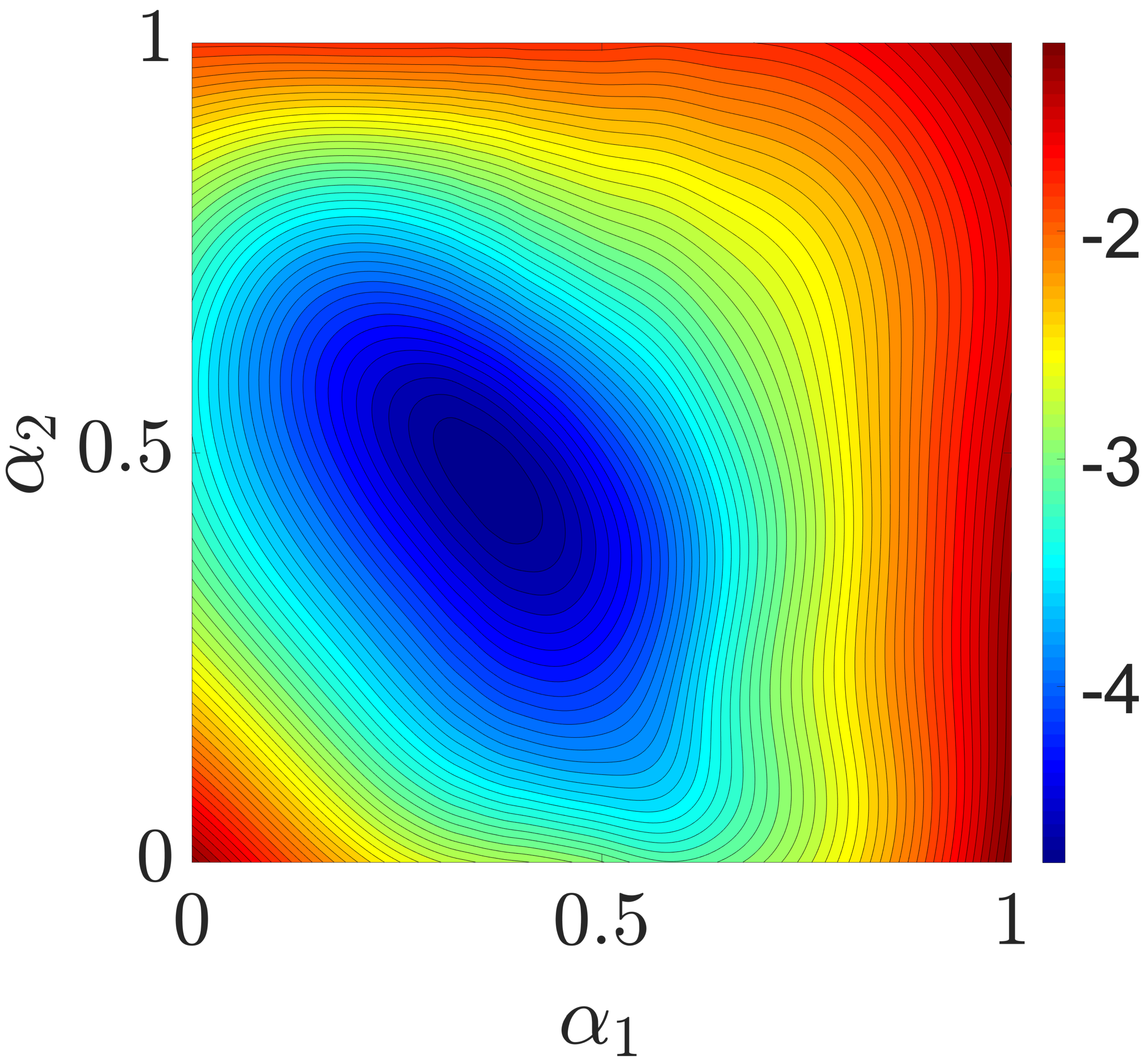}
		\caption{Modification}
	\end{subfigure}
	\vspace{-6pt}
	\caption{Contour plots of $\log_{10} \left(\E \left[(I_0-\gamma_0)^2\right]\right)$ produced by BOPP
		for different allocations of the sample budget $T=10^6$ for the
		problem shown in~\eqref{eq:repeat-nest} and its modified variant.  \vspace{-14pt}
		\label{fig:multi-tau}}
\end{figure}

When $y$ can only take on finitely many realizations $y_{1},\dots,y_c$, we use the ideas introduced
in Section~\ref{sec:discrete} to derive the following improved estimator
\begin{align}
\label{eq:u_bar_MC_main}
&\hat{U}_{R}(d) = \frac{1}{N} \sum_{n=1}^{N} \sum_{c=1}^{C} p(y_c | \theta_n, d) \log\left(p(y_c | \theta_n, d)\right) \\
&- \sum_{c=1}^{C} \left[\left(\frac{1}{N}\sum_{n=1}^{N} p(y_c | \theta_n, d)\right) \log \left(\frac{1}{N} \sum_{n=1}^{N} p(y_c | \theta_n, d)\right) \right] \nonumber
\end{align}
where $\theta_n \sim p(\theta), \forall n \in \{1,\dots,N\}$.  
As $C$ is fixed,~\eqref{eq:u_bar_MC_main} converges at the standard MC error rate of $O(1/N)$.  This constitutes
a substantially faster convergence as~\eqref{eq:exp-des-nmc-main} requires a total of $MN$
samples compared to $N$ for~\eqref{eq:u_bar_MC_main}.  
%To the best of our knowledge, this superior estimator is new to the literature.

We finish by showing that the theoretical advantages of this reformulation also lead to empirical gains.  
For this we consider a model used in psychology experiments introduced by \cite{vincent2016hierarchical}, details of which are given in Appendix~\ref{sec:exp-design}.  
Figure~\ref{fig:exp-conv} demonstrates that the theoretical convergence rates
are observed while results given in Appendix~\ref{sec:exp-design} show that this leads to significant practical gains 
in estimating $\bar{U}(d)$.

\begin{figure}[t]
	\centering
	\includegraphics[width=0.45\textwidth,trim={1.5cm 0 3.5cm 0},clip]{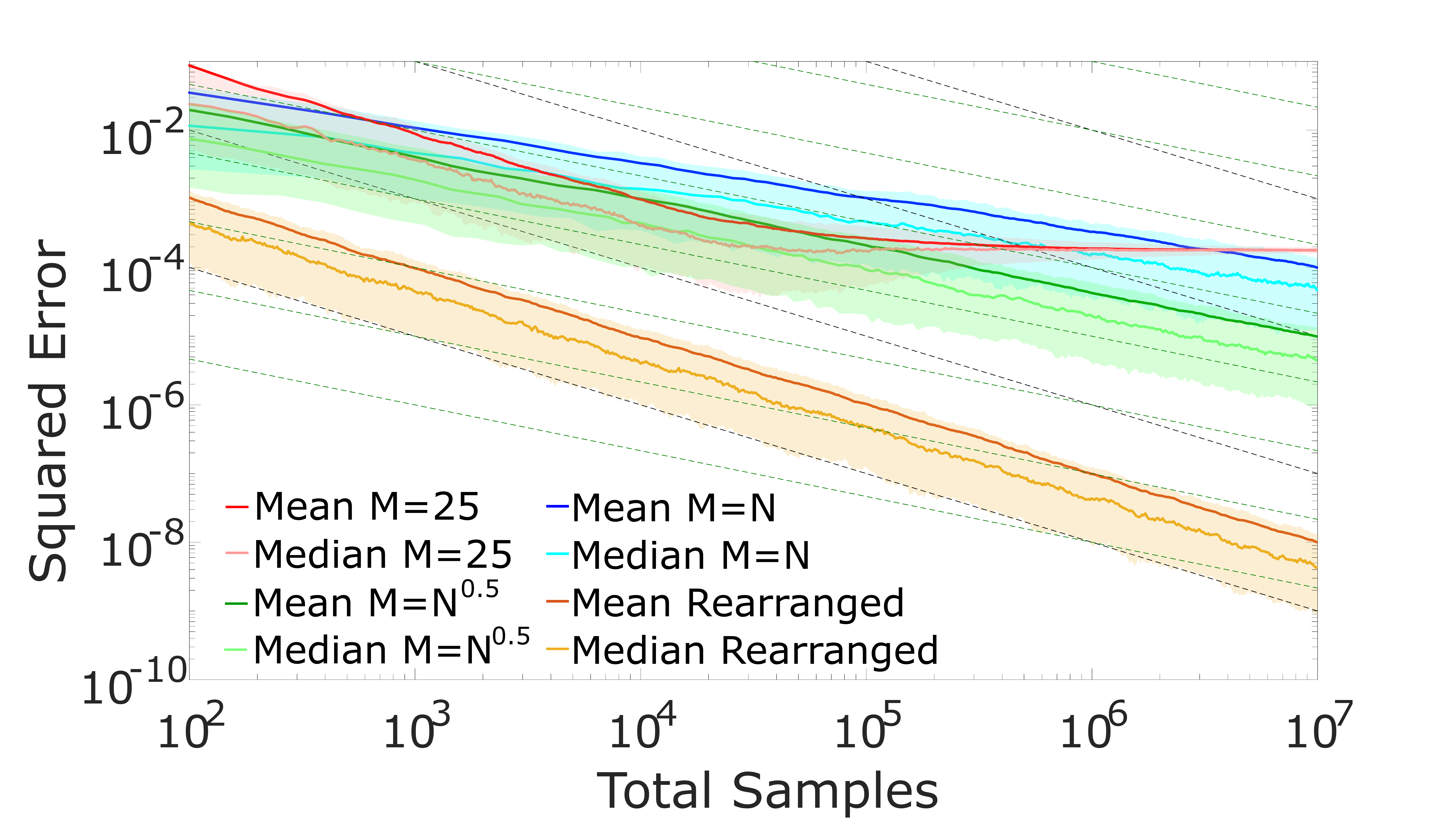}\vspace{-7pt}
	\caption{Convergence of NMC (i.e.~\eqref{eq:exp-des-nmc-main}) and our reformulated
		estimator~\eqref{eq:u_bar_MC_main} for the BED problem.
		Experimental setup and conventions are as per Figure~\ref{fig:emprical-conv}, with a
		ground truth estimate made using a single run of the reformulated
		estimator with $10^{10}$ samples. We
		see that the theoretical convergence rates are observed,
		with the advantages of the reformulated estimator particularly pronounced.\label{fig:exp-conv} \vspace{-20pt}}
\end{figure}
% !TEX root =  main.tex

\subsection{Variational Autoencoders}

To give another example of the applicability of our results, we now use 
Theorem~\ref{the:Repeat} to directly derive a new result for the 
importance weighted autoencoder (IWAE)~\citep{burda2015importance}.  Both the IWAE
and the standard variational autoencoder (VAE)~\citep{kingma2013auto} use 
lower bounds on the model evidence as objectives for train deep generative models and 
employ estimators of the form
\begin{align}
	\label{eq:iwae}
	I_{N,M} = \frac{1}{N}  \sum_{n=1}^{N}  \log \left(\frac{1}{M} \sum_{m=1}^{M} w_{n,m} (\theta)\right)
\end{align}
for some given $\theta$ upon which the random $w_{n,m} (\theta)$ depend.  The IWAE sets $N=1$ and the
VAE $M=1$.  We can view~\eqref{eq:iwae} as a (biased) NMC estimator for 
the evidence $\log \E \left[w_{1,1} (\theta)\right]$, which is
the target one actually wishes to optimize (for the generative network).  We
can now assess the MSE of this biased estimator using~\eqref{eq:cont-single},
noting that this is a special case where $\varsigma_0^2=0$,
giving 
$\E \left[\left(I_{N,M}-I\right)^2\right] \le \frac{C_0 ^2 \varsigma_1^4}{4 M^2}\left(1+\frac{1}{N}\right)
+\frac{K_{0}^2 \varsigma_1^2}{NM}+\frac{C_0 K_0 \varsigma_1^3}{N M^{3/2}}+O(\frac{1}{M^3})$. For a fixed
budget $T=NM$ this becomes $O\left(\frac{1}{M^2}+\frac{1}{T}+\frac{1}{T\sqrt{M}}\right)$.  Given $T$ is fixed, we thus see that 
the higher $M$ is, the lower the error bound.  Therefore,
the lowest MSE is achieved by setting $N=1$ and $M=T$, as is done by the IWAE.  
As we show in~\citet{rainforth2018tighter}, these results further carry over to the reparameterized derivative
estimates $\nabla_{\theta} I_{N,M}$.

\subsection{Nesting Probabilistic Programs}

Probabilistic programming systems (PPSs)
\citep{goodman2008church,wood2014new} provide a strong motivation for the study of NMC methods
because many PPSs allow for
arbitrary nesting of models (or queries, as they are known in the PPS literature), such that it is easy to 
define and run nested inference problems, including cases with multiple layers of
nesting~\citep{stuhlmuller2012dynamic,stuhlmuller2014reasoning}.
Though this ability to nest queries has started to be exploited in application-specific work
\citep{ouyang2016practical,le2016nested}, the resulting nested inference problems fall outside the
scope of conventional convergence proofs and so the statistical validity of the underlying inference
engines has previously been an open question in the field.

As we show in~\citet{rainforth2017thesis,rainforth2017nestpp}, the results presented here
can be brought to bear on assessing the relative correctness
of the different ways PPSs allow model nesting.  
In particular, the correctness of sampling from the conditional distribution of one query within another
follows from Theorem~\ref{the:Repeat},
but only if the computation for each call to the inner query increases the more times that query is
called.  This requirement is not satisfied by current systems.
%, e.g. Anglican's \conditional function~\cite{tolpin2016design}.  
Meanwhile, Theorem~\ref{the:prod} can be used to the assert 
that observing the output of one query inside another leads to convergence at the standard MC rate, provided
that the computation of the inner query instead remains fixed.

\section{Conclusions}

We have introduced a formal framework for NMC estimation and shown that
it can be used to yield a consistent estimator for problems that cannot be tackled
with conventional MC alone.  We have derived convergence rates and considered what minimal continuity assumptions
are required for convergence.
However, we have also highlighted a number of potential pitfalls for
na\"{i}ve application of NMC and provided guidelines for avoiding these, e.g. highlighting the
importance of increasing the number of samples in both the inner  and the outer estimators to ensure convergence.
We have further introduced techniques for converting certain classes of NMC problems to conventional
MC ones, providing improved convergence rates.
Our work has implications throughout machine learning and we hope it will provide
the foundations for exploring this plethora of applications.

\clearpage

\begin{appendices}
	\onecolumn
%	\setlength{\abovedisplayskip}{5pt}
%	\setlength{\belowdisplayskip}{5pt}
%	\setlength{\abovedisplayshortskip}{5pt}
%	\setlength{\belowdisplayshortskip}{5pt}
%	
%	\titlespacing\section{0pt}{12pt plus 2pt minus 2pt}{4pt plus 2pt minus 0pt}
%	\titlespacing\subsection{0pt}{12pt plus 2pt minus 2pt}{4pt plus 2pt minus 0pt}
%	\titlespacing\subsubsection{0pt}{12pt plus 2pt minus 2pt}{4pt plus 2pt minus 0pt}
%	
%	\thispagestyle{empty} 
%	%\vspace{-20pt}
%	\rule{\textwidth}{1pt}
%	\vspace{-6pt}
%\begin{center}
%	\textbf{ \Large On Nesting	Monte Carlo Estimators -- Supplementary Material}
%	\end{center}\vspace{-6pt}
%\rule{\textwidth}{1pt}

% !TEX root =  main.tex

\section{Proof of Theorem~\ref{the:Rate} - Simplified Convergence Rate}
\label{sec:app:rate_single}

\theRate*

\begin{proof}
	Though the Theorem follows directly from Theorem~\ref{the:Repeat}, we also provide the following proof 
	for this simplified case to provide a more accessible intuition behind the result.  Note that the approach
	taken is distinct from the proof of Theorem~\ref{the:Repeat}.
	
	Using Minkowski's inequality, we can bound the mean squared error of $I_{N,M}$ by
	\begin{align} \label{eq:mse-bound-app}\begin{split}
	&\E[(I-I_{N,M})^2] = \norm{I - I_{N,M}}_2^2 \leq {U}^2 + {V}^2 + 2 U V \leq 2\left(U^2 + V^2\right)
	\end{split} \\
	\nonumber
		\text{where} \quad  U &= \norm{I - \frac{1}{N} \sum_{n=1}^N f(y_n, \gamma(y_n))}_2
		\quad \text{and} \quad
		V = \norm{\frac{1}{N} \sum_{n=1}^N f(y_n, \gamma(y_n)) - I_{N,M}}_2.
	\end{align}
	We see immediately that $U = O\left(1 / \sqrt{N}\right)$, since $\frac{1}{N} \sum_{n=1}^N
	f(y_n, \gamma(y_n))$ is a MC estimator for $I$, noting our assumption that
	$f(y_n, \gamma(y_n)) \in L^2$. For the second term,
	\begin{eqnarray*}
		V &=& \norm{\frac{1}{N} \sum_{n=1}^N f(y_n, (\hat{\gamma}_M)_n) - f(y_n, \gamma(y_n))}_2 \\
		&\leq&\frac{1}{N} \sum_{n=1}^{N} \norm{f(y_n, (\hat{\gamma}_M)_n) - f(y_n,
			\gamma(y_n))}_2
		\leq \frac{1}{N} \sum_{n=1}^N K \norm{(\hat{\gamma}_M)_n - \gamma(y_n)}_2
	\end{eqnarray*}
	where $K$ is a fixed constant, again by Minkowski and using the assumption that $f$ is
	Lipschitz. We can rewrite
	\[
	\norm{(\hat{\gamma}_M)_n - \gamma(y_n)}_2^2
	= \E \left[ \E \left[ ((\hat{\gamma}_M)_n - \gamma(y_n))^2 \middle| y_n \right]\right].
	\]
	by the tower property of conditional expectation, and note that
	\begin{align*}
	\E &\left[ ((\hat{\gamma}_M)_n - \gamma(y_n))^2 \middle| y_n \right]
	= \var\left(\frac{1}{M} \sum_{m=1}^M \phi(y_n, z_{n,m}) \middle| y_n \right) 
	= \frac{1}{M} \var\left(\phi(y_n, z_{n,1}) \middle| y_n \right)
	%    &=& \E_{y_n \sim p(y)} \left[ O(1/M) \right] \\
	%    &=& O(1/M),
	\end{align*}
	since each $z_{n,m}$ is i.i.d. and conditionally independent given $y_n$. As
	such
	\begin{align*}
	\norm{(\hat{\gamma}_M)_n - \gamma(y_n)}_2^2
	&= \frac{1}{M} \, \E \left[ \var \left(\phi(y_n, z_{n,1}) \middle| y_n \right)\right] 
	= O(1/M),
	\end{align*}
	noting that $\E \left[ \var \left(\phi(y_n, z_{n,1}) \middle| y_n \right)\right]$ is a
	finite constant by our assumption that $\phi(y_n, z_{n,m}) \in L^2$. Consequently,
	\[
	V \leq \frac{NK}{N} O\left(1/\sqrt{M}\right) = O\left(1/\sqrt{M}\right).
	\]
	Substituting these bounds for $U$ and $V$ in \eqref{eq:mse-bound-app} gives
	\begin{align*}
	\norm{I - I_{N,M}}_2^2
	&\leq 2\left(O\left(1/\sqrt{N}\right)^2 + O\left(1/\sqrt{M}\right)^2\right) 
	= O\left(1/N + 1/M\right)
	\end{align*}
	as desired.
\end{proof}
% !TEX root =  main.tex

\section{The Inevitable Bias of Nested Estimation} 
\label{sec:bias}

In this section we demonstrate formally that NMC schemes must produce biased estimates of $I(f)$ for certain
functions $f$. In fact, our result applies more generally: we show that this holds for any
MC scheme that makes use of imperfect estimates $\hat{\zeta}_n$ of $\gamma(y_n)$,
either via a NMC procedure (e.g. $\hat{\zeta}_n = (\hat{\gamma}_M)_n$), or
when these inner estimates are generated by some other methods
such as a variational approximation~\citep{blei2016variational} or Bayesian
quadrature~\citep{o1991bayes}.
%A characterization of this result is shown in Figure~\ref{fig:bias-plot} and more formally
%in the following Theorem.
\vspace{-2pt}
\begin{restatable}{theorem}{theBias}
	\label{the:bias}
  Suppose that the random variables $\hat{\zeta}_n$ satisfy
    $\mathbb{P}(\hat{\zeta}_n \neq \gamma(y_n)) > 0$.
  Then we can choose $f$ such that if $y_n \sim p(y)$,
    $\E\left[\frac{1}{N} \sum_{n=1}^N f(y_n, \hat{\zeta}_n)\right] \neq I(f)$ for any
    $N$ (including the limit $N\rightarrow\infty$).
\end{restatable}
\vspace{-12pt}
\begin{proof}
	Take $f(y, w) = (\gamma(y) - w)^2$. Then $f(y, \gamma(y)) = 0$, so that $I(f) = 0$.  On
	the other hand, $f(y_n, \hat{\zeta}_n) \geq 0$ since $f$ is non-negative.
	Moreover, $f(y_n, \hat{\zeta}_n) > 0$ on the event $\{\hat{\zeta}_n \neq \gamma(y_n)\}$.
	Since we assumed this event has nonzero probability, it follows that $\E \left[f(y_n, \hat{\zeta}_n)\right] > 0$ and hence
	\[
	\E\left[\frac{1}{N} \sum_{n=1}^N f(y_n, \hat{\zeta}_n)\right]
	= \frac{1}{N} \sum_{n=1}^N \E\left[f(y_n, \hat{\zeta}_n)\right]
	> 0 = I(f)
	\]
	which gives the required result.
\end{proof}
\vspace{-8pt}
It also follows from Jensen's inequality
that \emph{any} strictly convex or concave $f$  entails a biased estimator 
when $\hat{\zeta}_n$ is unbiased but has non-zero variance given $y_n$, e.g.
when $\hat{\zeta}_n$ is a MC estimate.  More formally we have
\begin{theorem}
	\label{the:bias-conv}	
	Suppose that $y_n \sim p(y)$ and that each $\hat{\zeta}_n$ satisfies $\E \left[\hat{\zeta}_n \middle| y_n\right] = \gamma(y_n)$.
	Define $\mathcal{A} \subseteq \mathcal{Y}$ as
	$\mathcal{A} = \left\{y \in \mathcal{Y} \;\middle|\; \mathrm{Var}\left(\hat{\zeta}_n \middle| y_n=y \right)>0\right\}$
	and assume that ~$\mathbb{P}(y_n \in \mathcal{A}) > 0$.
	Then for any $f$ that is strictly convex in its second argument,
	\[ 
	\E\left[\frac{1}{N} \sum_{n=1}^N f(y_n, \hat{\zeta}_n)\right] > I(f).
	\]
	Similarly for any $f$ that is strictly concave in its second argument, 
	\[
	\E\left[\frac{1}{N} \sum_{n=1}^N f(y_n, \hat{\zeta}_n)\right] < I(f).
	\]
\end{theorem}
\begin{proof}
	We prove our claim for the case that $f$ is strictly convex; our proof for the other concave case  
	is symmetrical. We have
	\begin{align*}
		\E\left[\frac{1}{N} \sum_{n=1}^N f(y_n, \hat{\zeta}_n)\right] = \E\left[f(y_1, \hat{\zeta}_1)\right]
		= \E\left[\E\left[f(y_1, \hat{\zeta}_1) \middle| y_1\right]\right] 
		\ge \E\left[f\left(y_1, \E\left[\hat{\zeta}_1 \middle| y_1\right]\right)\right] = I(f)
		%        \E&\left[\frac{1}{N} \sum_{n=1}^N f(y_n, \hat{\zeta}_n)\right] = \E\left[f(y_1, \hat{\zeta}_1)\right] \\
		%        & \quad \quad \quad \quad = \E\left[\E\left[f(y_1, \hat{\zeta}_1)|y_1\right]\right] \\
		%        & \quad \quad \quad \quad = \E\left[\mathbb{I}(y_1 \in \mathcal{A}) \E\left[f(y_1, \hat{\zeta}_1)|y_1\right]\right] \\
		%        & \quad \quad \quad \quad \quad + \E\left[\mathbb{I}(y_1 \notin \mathcal{A}) \E\left[f(y_1, \hat{\zeta}_1)|y_1\right]\right] \\
		%        & \quad \quad \quad \quad \ge \E\left[\mathbb{I}(y_1 \in \mathcal{A}) f\left(y_1, \E\left[\hat{\zeta}_1 \middle| y_1\right]\right)\right] \\
		%        & \quad \quad \quad \quad \quad + \E\left[\mathbb{I}(y_1 \notin \mathcal{A}) \E\left[f(y_1, \hat{\zeta}_1)|y_1\right]\right] \\    
		%        & \quad \quad \quad \quad = \E\left[\E\left[f(y_1, \hat{\zeta}_1)|y_1\right]\right] = I(f)
	\end{align*}
	where the $\ge$ is a result of Jensen's inequality on the inner expectation.  
	Since $f$ is strictly convex and therefore non-linear, equality holds if and only if
	$\hat{\zeta}_1$ is almost surely constant given $y_1$. 
	This is violated whenever $y_1 \in \mathcal{A}$, which by assumption
	has a non-zero probability of occurring.  Consequently,
	the inequality must be strict, giving the desired result.
\end{proof}

In addition to some special cases discussed in the Section~\ref{sec:special_cases},
it may still be possible to develop unbiased estimation
schemes for certain non-linear $f$ using Russian roulette sampling~\citep{lyne2015russian} or other debiasing techniques.
However, these induce their own complications: for some problems the resultant estimates
have infinite variance \citep{lyne2015russian} and as shown by \cite{jacob2015nonnegative}, there are 
no general purpose ``$f$-factories'' that produce both non-negative and
unbiased estimates for non-constant, positive output functions $f : \real \rightarrow \real^+$,
given unbiased estimates for the inputs.

\section{Proof of Theorem~\ref{the:Consistent} - ``Almost almost sure'' convergence}
\label{sec:app:consistent}

\theConsistent*

\begin{proof}
	For all $N, M$, we have by the triangle inequality that
		\begin{align*}
	\left|I_{N,M} - I\right| \leq V_{N,M} + U_N, &\quad \text{where} \\
		V_{N,M} = \left|\frac{1}{N} \sum_{n=1}^N f(y_n, \gamma(y_n)) - I_{N,M} \right| 
		\quad \text{and}& \quad
		U_N = \left|I - \frac{1}{N} \sum_{n=1}^N f(y_n, \gamma(y_n)) \right|.
	\end{align*}
	A second application of the triangle inequality then allows us to write
	\[
    V_{N,M} \leq \frac{1}{N} \sum_{n=1}^N (\epsilon_M)_n
	\]
	where we recall that $(\epsilon_M)_n = |f(y_n, \gamma(y_n)) - f(y_n, \hat{\gamma}_n)|$.
	Now, for all fixed $M$, each $(\epsilon_M)_n$ is i.i.d. Furthermore, since
        $\E
	\left[(\epsilon_M)_1\right] \to 0$ as $M \to \infty$ by our assumption and 
	$(\epsilon_M)_n$ is nonnegative, there exists some $L \in \N$
	such that $\E\left[\left|(\epsilon_M)_n \right|\right] < \infty$ for all $M \geq L$.
	Consequently, the strong law of large numbers means that as $N \to \infty$ then for all $M \geq L$
	\begin{align}
	\label{eq:epas}
	\frac{1}{N} \sum_{n=1}^N (\epsilon_M)_n \asto \E\left[(\epsilon_M)_1\right].
	\end{align}      
        For any fixed $\delta > 0$ then by repeatedly applying Egorov's theorem to each $M \geq L$,
        we can find a sequence of events 
        \[
                B_{L}, B_{L+1}, B_{L+2}, \ldots
        \]
        such that for every $M \geq L$,
        \[
                \mathbb{P}(B_{M}) < \frac{\delta}{4} \cdot \frac{1}{2^{M-L}}
        \]
        and outside of $B_{M}$, the sequence $\frac{1}{N} \sum_{n=1}^N (\epsilon_M)_n$
        converges \emph{uniformly} to $\E\left[(\epsilon_M)_1\right]$. 
        This uniform convergence (as opposed to just the piecewise convergence implied
        by~\eqref{eq:epas}) now guarantees that we can define some function
        $\tau^1_\delta : \N \to \N$ such that 
	\begin{align}
	\label{eq:ombo}
                \left|\frac{1}{M'} \sum_{n=1}^{M'} (\epsilon_M)_n(\omega) - \E\left[(\epsilon_M)_1\right]\right| < \frac{1}{M} 
        \end{align}
        for all $M \geq L$, $M' \geq \tau^1_\delta(M)$, and $\omega \not\in B_M$, remembering that $\omega$ is
        a point in our sample space.  We further have that~\eqref{eq:ombo} holds for all 
        $M \geq M_0$, $M' \geq \tau^1_\delta(M)$, and $\omega \not\in B_\delta := \bigcup_{M \geq L} B_M$.
%        Such a $\tau^1_\delta$ exists
%        because of the uniform convergence guarantee that we just mentioned. 
%        Defining 
%        \[
%                B_\delta := \bigcup_{M \geq L} B_M,
%        \]
%        then~\eqref{eq:ombo} hold for all $M \geq M_0$, $M' \geq \tau^1_\delta(M)$, and $\omega \not\in B_\delta$.
        Consequently, for all such $M$, $M'$ and $\omega$,
        \begin{equation} 
                \label{eqn:almostsure:1}
                V_{M',M}(\omega)
                \leq
                \frac{1}{M'} \sum_{n=1}^{M'} (\epsilon_M)_n(\omega) 
                < 
                \frac{1}{M} + \E\left[(\epsilon_M)_1\right],
        \end{equation}
        while we also have
        \begin{equation} 
                \label{eqn:almostsure:2}
                \mathbb{P}(B_\delta) 
                %=\mathbb{P}\left(\bigcup_{M \geq L} B_M\right)
                \leq 
                \sum_{M \geq L} \mathbb{P}\left(B_M\right)
                < \sum_{M \geq L} \frac{\delta}{4} \cdot \frac{1}{2^{M-L}}
                = \frac{\delta}{2}.
        \end{equation}
      To complete the proof, we must remove the dependence of $U_N$ on $N$ as well. This is
	straightforward once we observe that $U_N \asto 0$ as $N \to \infty$ by the strong law of 
        large numbers. So, by Egorov's theorem again, there exists an event $C_\delta$ such that
        \begin{equation} 
                \label{eqn:almostsure:3}
                \mathbb{P}(C_\delta) < \frac{\delta}{2}
        \end{equation}
        and outside of $C_\delta$, the sequence $U_N$ converges uniformly to $0$. This uniform convergence,
        in turn, ensures the existence of a function $\tau^2_\delta : \N \to \N$ such that
        \begin{equation} 
                \label{eqn:almostsure:4}
                U_{M'}(\omega) < \frac{1}{M} 
        \end{equation}
        for all $M \in \N$, $M' \geq \tau^2_\delta(M)$, and $\omega \not\in C_\delta$.
	
        We can now define $\tau_\delta(M) = \max(\tau^1_\delta(M), \tau^2_\delta(M))$, and 
        $A_\delta = B_\delta \cup C_\delta$. By inequalities in \eqref{eqn:almostsure:2}
        and \eqref{eqn:almostsure:3},
        \[ 
                \mathbb{P}(A_\delta) \leq \mathbb{P}(B_\delta) + \mathbb{P}(C_\delta) < \delta.  
        \] 
        Also, by the inequalities in \eqref{eqn:almostsure:1} and \eqref{eqn:almostsure:4},
	\[ 
                \left|I - I_{\tau_\delta(M),M}(\omega)\right| 
                \leq
                V_{\tau_\delta(M),M}(\omega) 
                +
                U_{\tau_\delta(M)}(\omega) 
                \leq 
                \frac{1}{M} + \frac{1}{M} + \E\left[(\epsilon_M)_1\right]
	\]
        for all $M \geq L$ and $\omega \notin A_\delta$. Since $\E\left[(\epsilon_M)_1\right] \to 0$, we have here that 
        $I_{\tau_\delta(M),M}(\omega) \to I$ as desired.
\end{proof}

\section{Proof of Theorem~\ref{the:Repeat} - Convergence for Repeated Nesting}
\label{sec:app:repeat}
%
%Below we provide a proof of Theorem~\ref{the:Repeat} from the main paper.  We note that this covers 
%Theorem~\ref{the:Rate} as a special case.  A separate proof for Theorem~\ref{the:Rate} was provided in
%Appendix~\ref{sec:app:rate_single}.  This may make easier reading on a first pass as 
%it is substantially less involved that the following, more general case.

\theRepeat*
\begin{proof}
As this is a long and involved proof, we start by defining a number of useful terms that will be 
used throughout.  Unless otherwise stated, these definitions hold for all $k \in \left\{0,\dots,D\right\}$.
Note that most of these terms implicitly depend on the number of
samples $N_0,N_1,\dots,N_D$.  However, $s_k$, $\zeta_{d,k}$, and $\varsigma_k$ do not
and are thus constants for a particular problem.
\begin{align}
\intertext{$E_k \left(y^{(0:k-1)}\right)$ is the MSE of the estimator at depth $k$ given $y^{(0:k-1)}$}
E_k \left(y^{(0:k-1)}\right) &:= \E \left[\left(I_{k}\left(y^{(0:k-1)}\right)-
\gamma_{k}\left(y^{(0:k-1)}\right)\right)^2 \middle| y^{(0:k-1)}\right]
\displaybreak[0] \\
\intertext{$\bar{f}_{k} \left(y^{(0:k-1)}\right)$ is the expected value of the estimate at depth
	$k$, or equivalently the expected function output using the estimate of the layer below}
\begin{split}
\bar{f}_{k} \left(y^{(0:k-1)}\right) &:=\E\left[I_{k}\left(y^{(0:k-1)}\right) \middle| y^{(0:k-1)}\right] \;\; \forall k\in\{1,\dots,D-1\} \\
&=\E\left[f_k\left(y^{(0:k)},I_{k+1}\left(y^{(0:k)}\right)\right)
\middle|  y^{(0:k-1)}\right] 
\end{split}
\displaybreak[0] \\ 
\intertext{$v_k^2 \left(y^{(0:k-1)} \right)$ is the variance of the estimator at depth $k$}
    \begin{split}
    	v_k^2 \left(y^{(0:k-1)} \right) &:= 
    	\text{Var}\left[I_{k}\left(y^{(0:k-1)}\right) \middle| y^{(0:k-1)}\right] \\
    	&= \E\left[\left(I_{k}\left(y^{(0:k-1)}\right)- \bar{f}_k 
    	\left(y^{(0:k-1)}\right) \right)^2 \middle| y^{(0:k-1)}\right]
    \end{split}
\displaybreak[0]   \\ 
\intertext{$\beta_k \left(y^{(0:k-1)} \right)$ is the bias of the estimator at depth $k$}
   \begin{split}
   	\label{eq:bias-def}
   	\beta_k \left(y^{(0:k-1)} \right) &:= 
   	\E  \left[I_{k}\left(y^{(0:k-1)}\right)-
   	\gamma_{k}\left(y^{(0:k-1)}\right) \middle| y^{(0:k-1)}\right] \\
   	&=\bar{f}_{k} \left(y^{(0:k-1)}\right)-\gamma_{k}\left(y^{(0:k-1)}\right) \\
   	&=
   	\E \left[f_k\left(y^{(0:k)}, I_{k+1}\left(y^{(0:k)}\right)\right)
   	- f_k\left(y^{(0:k)}, \gamma_{k+1}\left(y^{(0:k)}\right)\right)
   	\middle|  y^{(0:k-1)} \right]
   \end{split}
\displaybreak[0] \\ 
\intertext{$s_k^2 \left(y^{(0:k-1)}\right)$ is the variance at depth $k$ of the true function output}
s_k^2 \left(y^{(0:k-1)}\right) &:= \E \left[\left(f_k\left(y^{(0:k)},\gamma_{k+1}
\left(y^{(0:k)}\right) \right)-\gamma_k\left(y^{(0:k-1)}\right)\right)^2 \middle|
y^{(0:k-1)}	\right]
\\ \displaybreak[0]
s_D^2 \left(y^{(0:D-1)}\right) &:= \E \left[\left(f_D\left(y^{(0:D)}\right)-\gamma_D\left(y^{(0:D)}\right)\right)^2 \middle| y^{(0:D-1)}	\right]
\\
\intertext{$\zeta_{d,k}^2\left(y^{(0:k-1)}\right)$ is expectation of $s_d^2 \left(y^{(0:d-1)}\right)$ over 
	$y^{(k:d-1)}$}
\begin{split}
	\zeta_{d,k}^2\left(y^{(0:k-1)}\right) &:= 
	\E \left[ s_{d}^2 \left(y^{(0:d-1)}\right) \middle|
	y^{(0:k-1)}\right] \\
	&=\E \left[\left(f_d\left(y^{(0:d)},\gamma_{d+1}
	\left(y^{(0:d)}\right) \right)-\gamma_d\left(y^{(0:d-1)}\right)\right)^2 \middle|
	y^{(0:k-1)}	\right]
\end{split}
\displaybreak[0] \\
\intertext{$\varsigma_{k}^2 $ is expectation of $s_k^2 \left(y^{(0:k-1)}\right)$
	over all $y^{(0:k-1)}$}
\varsigma_{k}^2  
&:=\zeta_{k,0}^2=\E \left[\left(f_k\left(y^{(0:k)},\gamma_{k+1}
\left(y^{(0:k)}\right) \right)-\gamma_k\left(y^{(0:k-1)}\right)\right)^2\right]
 \displaybreak[0] \\
 \intertext{$A_k \left(y^{(0:k-1)}\right)$ is the MSE in the function output from
 	using the estimate of the next layer, rather than the true value $\gamma_{k+1}\left(y^{(0:k)}\right)$,
 	we fix $A_D:=0$}
A_k \left(y^{(0:k-1)}\right):=& \E \left[\left(f_k\left(y^{(0:k)}, I_{k+1}\left(y^{(0:k)}\right)\right)
- f_k\left(y^{(0:k)}, \gamma_{k+1}\left(y^{(0:k)}\right)\right)\right)^2
\middle|  y^{(0:k-1)} \right] 
%\\\displaybreak[0]
%A_D &:=0 
\displaybreak[0]\\
 \intertext{$\sigma_k^2 \left(y^{(0:k-1)}\right)$
 	is the variance in the function output from using the estimate of the next layer,
 	we fix $\sigma_D^2 \left(y^{0:D-1}\right) := s_D^2 \left(y^{0:D-1}\right)$ }
 \begin{split}
 \sigma_k^2 \left(y^{(0:k-1)}\right) &:= 
 \text{Var}\left[f_k\left(y^{(0:k)},I_{k+1}\left(y^{(0:k)}\right)\right) \middle| y^{(0:k-1)}\right] \\
 &= \E\left[\left(f_k\left(y^{(0:k)},I_{k+1}\left(y^{(0:k)}\right)\right)
 - \bar{f}_k 
 \left(y^{(0:k-1)}\right) \right)^2 \middle| y^{(0:k-1)}\right]
 \end{split}
\displaybreak[0]  \\ 
 \intertext{$\omega_k\left(y^{(0:k-1)} \right)$ is the expectation over $y^{(k)}$ of the MSE for the next layer,
 	we fix $\omega_D \left(y^{(0:D-1)}\right) := 0$}
      \begin{split}
      \omega_k\left(y^{(0:k-1)} \right) &:=  \E \left[E_{k+1} 
      \left(y^{(0:k)}\right) \middle|  y^{(0:k-1)} \right] \\
      &=\E \left[
      \left(I_{k+1}\left(y^{(0:k)}\right) - \gamma_{k+1}\left(y^{(0:k)}\right)\right)^2
      \middle|  y^{(0:k-1)} \right] 
      \end{split}
  \displaybreak[0]
 \\
  \intertext{$\lambda_k \left(y^{(0:k-1)} \right)$ is the expectation over $y^{(k)}$ of the magnitude of the bias for the next layer,
  	we fix $\lambda_D \left(y^{(0:D-1)}\right) := 0$ and note that $\lambda_{D-1} \left(y^{(0:D-2)}\right) := 0$ also
  	as the innermost layer is an unbiased}
    \begin{split}
   \lambda_k \left(y^{(0:k-1)} \right) &:= \E \left[\left|\beta_{k+1} 
   \left(y^{(0:k)}\right) \right| \Bigg|   y^{(0:k-1)} \right]\\
   &=
   \E \left[ \left|\E \left[
   \left(I_{k+1}\left(y^{(0:k)}\right) - \gamma_{k+1}\left(y^{(0:k)}\right)\right)
   \middle|  y^{(0:k)} \right] \right| \Bigg|  y^{(0:k-1)} \right]
   \end{split}
\end{align}

\vspace{5pt}
\noindent\textbf{\large Lipschitz Continuous Case}
\vspace{5pt}

\noindent Given these definitions, we start by breaking the error down into a variance and bias term.  
Using the standard bias-variance decomposition we have
\begin{align}
E_k \left(y^{(0:k-1)}\right) &= \E \left[\left(I_{k}\left(y^{(0:k-1)}\right)-
\gamma_{k}\left(y^{(0:k-1)}\right)\right)^2 \middle| y^{(0:k-1)}\right]
\nonumber \\
&=v_k^2 \left(y^{(0:k-1)} \right)
+\left(\beta_k \left(y^{(0:k-1)} \right)\right)^2 \label{eq:bias-var-decomp}
\end{align}
It is immediately clear from its definition in \eqref{eq:bias-def} that the bias term
$\left(\beta_k \left(y^{(0:k-1)} \right)\right)^2$ is independent of 
$N_0$.  On the other hand, we will show later that the
dominant components of the variance term for large $N_{0:D}$ depend only
on $N_0$.  We can thus think of increasing $N_0$ as reducing the variance of
the estimator and increasing $N_{1:D}$ as reducing the bias.

We first consider the variance term
\begin{align*}
v_k^2 \left(y^{(0:k-1)} \right) &= \E \left[\left(
\frac{1}{N_k} \sum_{n=1}^{N_k} f_k\left(y^{(0:k)}_n, I_{k+1}\left(y^{(0:k)}_n\right)\right)-
\bar{f}_k 
\left(y^{(0:k-1)}\right) \right)^2 \middle| y^{(0:k-1)}\right]   \displaybreak[0] \\
&=\frac{1}{N_k} \E \left[\left(
f_k\left(y^{(0:k)}, I_{k+1}\left(y^{(0:k)}\right)\right)-
\bar{f}_k 
\left(y^{(0:k-1)}\right) \right)^2 \middle| y^{(0:k-1)}\right]
\end{align*}
with the equality following because the $y_n^{(0:k)}$ being drawn i.i.d. and the
expectation of each
$f_k\left(y^{(0:k)}, I_{k+1}\left(y^{(0:k)}\right)\right)$ equaling $\bar{f}_k 
\left(y^{(0:k-1)}\right)$ means that all the cross terms are zero.
%&= \E \left[\frac{1}{N_k^2} \sum_{n=1}^{N_k} \left(
% f_k\left(y^{(0:k)}_n, I_{k+1}\left(y^{(0:k)}_n\right)\right)-
%\bar{f}_k 
%\left(y^{(0:k-1)}\right) \right)^2 \middle| y^{(0:k-1)}\right] \\
%&\phantom{=} +\E \left[\frac{1}{N_k^2} \sum_{n=1}^{N_k} \sum_{m=1, m\neq n}^{N_k}\left(
%f_k\left(y^{(0:k)}_n, I_{k+1}\left(y^{(0:k)}_n\right)\right)-
%\bar{f}_k 
%\left(y^{(0:k-1)}\right) \right) \right.\\
%&\quad\quad\quad\quad\quad\quad
%\cdot\left(
%f_k\left(y^{(0:k)}_m, I_{k+1}\left(y^{(0:k)}_m\right)\right)-
%\bar{f}_k 
%\left(y^{(0:k-1)}\right) \right)
% \Bigg| y^{(0:k-1)}\Bigg]
%\end{align*}
%Now noting that each $y_n^{(0:k)}$ are drawn i.i.d., the cross terms are
%independent and the distribution is the same for each $m$ and $n$
% and so we have
%\begin{align*}
%v_k^2 \left(y^{(0:k-1)} \right) &= 
%\frac{1}{N_k} \E \left[\left(
%f_k\left(y^{(0:k)}, I_{k+1}\left(y^{(0:k)}\right)\right)-
%\bar{f}_k 
%\left(y^{(0:k-1)}\right) \right)^2 \middle| y^{(0:k-1)}\right]\\
%&\phantom{=}+\frac{N_k-1}{N_k}
%\left(\E \left[
%f_k\left(y^{(0:k)}, I_{k+1}\left(y^{(0:k)}\right)\right)-
%\bar{f}_k 
%\left(y^{(0:k-1)}\right)\middle| y^{(0:k-1)}\right]\right)^2,
%\end{align*}
%where the second term is zero by the definition of $\bar{f}_k$ and so
By the definition of $\sigma_k^2$ we now have
\begin{align}
	\label{eq:lln}
v_k^2 \left(y^{(0:k-1)} \right)  &= \frac{\sigma_k^2 \left(y^{(0:k-1)} \right)}{N_k}.
\end{align}
By using Minkowski's inequality and the definition of $A_k$ it also follows that
\begin{align}
\sigma_k & \left(y^{(0:k-1)} \right) \le 
\left(A_k \left(y^{(0:k-1)}\right)\right)^{\frac{1}{2}}+
\left(\E \left[\left(
f_k\left(y^{(0:k)}, \gamma_{k+1}\left(y^{(0:k)}\right)\right)
-\bar{f}_k\left(y^{(0:k-1)}\right)\right)^2 \middle| y^{(0:k-1)} \right] \right)^{\frac{1}{2}}.\label{eq:sigle1}
\end{align}
Using a bias-variance decomposition on the second term above and noting that 
$s_k^2 \left(y^{(0:k-1)} \right)$ and $\bar{f}_k\left(y^{(0:k-1)}\right)-\beta_k \left(y^{(0:k-1)} \right)$
are respectively the variance and expectation of $f_k\left(y^{(0:k)}, \gamma_{k+1}\left(y^{(0:k)}\right)\right)$,
%	and by substituting $\bar{f}_k=\gamma_k\left(y^{(0:k-1)}\right)-\beta_k \left(y^{(0:k-1)} \right)$,
%expanding the square, and noting that 
%$\E \left[f_k\left(y^{(0:k)},\gamma_{k+1}
%\left(y^{(0:k)}\right) \right)-\gamma_k\left(y^{(0:k-1)}\right) \middle| y^{(0:k-1)} \right]=0$ 
we can rearrange the right-hand size of~\eqref{eq:sigle1} to give
%&=
%\left(A_k \left(y^{(0:k-1)}\right)\right)^{\frac{1}{2}} +
%\left( \E \left[\left(f_k\left(y^{(0:k)},\gamma_{k+1}
%\left(y^{(0:k)}\right) \right)-\gamma_k\left(y^{(0:k-1)}\right)+\beta_k \left(y^{(0:k-1)} \right)\right)^2 \middle|
%y^{(0:k-1)}	\right]\right)^{\frac{1}{2}} \nonumber \\
\begin{align}
\sigma_k \left(y^{(0:k-1)} \right) &\le\left(A_k \left(y^{(0:k-1)}\right)\right)^{\frac{1}{2}}
+\left(s_k^2 \left(y^{(0:k-1)} \right) +
\left(\beta_k \left(y^{(0:k-1)} \right)\right)^2 
 \right)^{\frac{1}{2}}. \label{eq:sigma_bound}
\end{align}
Here $s_k^2$ is independent of the number of samples used at any level of the estimate,
while $A_k$ and $\beta_k^2$ are independent of $N_d \; \forall d\le k$.
Now by Jensen's inequality, we have that
\begin{align}
\label{eq:AleB}
\left(\beta_k \left(y^{(0:k-1)}\right)\right)^2 \le
A_k \left(y^{(0:k-1)}\right)
\end{align}
noting that the only difference in the definition of $\left(\beta_k \left(y^{(0:k-1)}\right)\right)^2$
and $A_k \left(y^{(0:k-1)}\right)$ is
whether the squaring occurs inside or outside the expectation.
Therefore, presuming that $A_k$
does not increase with $N_d  \; \forall d>k$, neither will $\sigma_k^2 \left(y^{(0:k-1)} \right)$, and so
the variance term will converge to zero with rate $O(1/N_k)$.  
Further, if ${A_k}\rightarrow 0$ as $N_{k+1},\dots,N_D \rightarrow \infty$,
then for a large number of inner samples $\sigma_k^2 \rightarrow s_k^2$ and thus we will have
$ v_k^2 \left(y^{(0:k-1)} \right) \le \frac{s_k^2}{N_k} +
O\left(\epsilon\right)$ where $O\left(\epsilon\right)$ represents higher order
terms that are dominated in the limit $N_k,\dots,N_D \rightarrow \infty$.
Provided this holds, we will also, therefore, have that
\begin{align}
\label{eq:E_decomp}
E_k \left(y^{(0:k-1)}\right) 
=\frac{\sigma_k^2 \left(y^{(0:k-1)}\right)}{N_k}+ \beta_k^2 \left(y^{(0:k-1)} \right)
&=\frac{s_k^2 \left(y^{(0:k-1)}\right)}{N_k}+ \beta_k^2 \left(y^{(0:k-1)} \right) +O(\epsilon).
\end{align}

We now show that Lipschitz continuity is sufficient for ${A_k}\rightarrow0$ and derive a
concrete bound on the variance by bounding ${A_k}$.  By definition of Lipschitz continuity,
we have that
\begin{align}
\left(A_k \left(y^{(0:k-1)}\right)\right)^{\frac{1}{2}} &\le
\left(\E \left[K_k^2 \left(I_{k+1}\left(y^{(0:k)}\right)-
 \gamma_{k+1}\left(y^{(0:k)}\right)\right)^2\middle| y^{(0:k-1)}
 \right] \right)^{\frac{1}{2}} \nonumber \\
 &= K_k \left(\omega_k \left(y^{(0:k-1)}\right)\right)^{\frac{1}{2}}
\end{align}
where we remember that $\omega_k \left(y^{(0:k-1)}\right) 
=\E \left[E_{k+1} 
\left(y^{(0:k)}\right) \middle|  y^{(0:k-1)} \right]$ is the expected MSE
of the next level estimator.  Once we also have an expression for the 
bias, we will thus be able to use this bound on $A_k$ along with~\eqref{eq:bias-var-decomp},~\eqref{eq:lln},
and~\eqref{eq:sigma_bound} to inductively derive a bound on the error.

For the case where we only assume Lipschitz continuity then we will simply
use the bound on the bias given by~\eqref{eq:AleB} leading to
\begin{align}
 E_k &\left(y^{(0:k-1)}\right) 
\le \frac{\sigma_k^2 \left(y^{(0:k-1)}\right)}{N_k} + A_k \left(y^{(0:k-1)}\right) \displaybreak[0] \\
&\le \frac{s_k^2 \left(y^{(0:k-1)}\right) +
2A_k \left(y^{(0:k-1)}\right)
+2\left(A_k \left(y^{(0:k-1)}\right)\right)^{\frac{1}{2}}
\left(s_k^2 \left(y^{(0:k-1)}\right) + A_k \left(y^{(0:k-1)}\right)\right)^{\frac{1}{2}}}
{N_k} +A_k \left(y^{(0:k-1)}\right)\nonumber \displaybreak[0] \\
&= \frac{s_k^2 \left(y^{(0:k-1)}\right) +
	2K_k^2 \omega_k \left(y^{(0:k-1)}\right)}{N_k} +K_k^2 \omega_k \left(y^{(0:k-1)}\right)\nonumber\\
&\phantom{==} 	+\frac{2K_k\left(\omega_k \left(y^{(0:k-1)}\right)\right)^{\frac{1}{2}}
	\left(s_k^2 \left(y^{(0:k-1)}\right) +  K_k^2 \omega_k \left(y^{(0:k-1)}\right)\right)^{\frac{1}{2}}}
{N_k}\nonumber \displaybreak[0] \displaybreak[0]\\
&\le \frac{s_k^2 \left(y^{(0:k-1)}\right) +
	4 K_k^2 \omega_k \left(y^{(0:k-1)}\right)
	+2 K_k \left(\omega_k \left(y^{(0:k-1)}\right)\right)^{\frac{1}{2}}
	s_k \left(y^{(0:k-1)}\right)}{N_k} +K_k^2 \omega_k \left(y^{(0:k-1)}\right)
\label{eq:general-bound-lip}
\end{align}
which fully defines a bound on conditional the variance of one layer given the mean squared error of the layer below.
In particular as $\omega_D \left(y^{(0:D-1)}\right) = 0$ we now have
\[
E_D \left(y^{(0:D-1)}\right) \le \frac{s_D^2 \left(y^{(0:D-1)}\right)}{N_D} = 
\frac{\E \left[\left(f_D\left(y^{(0:D)}\right)-\gamma_D\left(y^{(0:D)}\right)\right)^2 \middle| y^{(0:D-1)}	\right]}{N_D}
\]
which is the standard error for Monte Carlo convergence.  
We
further have 
\[
\omega_{D-1} \left(y^{(0:D-2)}\right) = 
\E \left[E_{D} 
\left(y^{(0:D-1)}\right) \middle|  y^{(0:D-2)} \right]
=
\frac{\zeta^2_{D,D-1}
	\left(y^{(0:D-2)}\right) }{N_D}.
\]
and thus
\begin{align}
\begin{split}
E_{D-1} \left(y^{(0:D-2)}\right) \le&  \frac{s_{D-1}^2 \left(y^{(0:D-2)}\right)}{N_{D-1}} +
	\frac{4 K_{D-1}^2 \zeta^2_{D,D-1}
		\left(y^{(0:D-2)}\right)}{N_D N_{D-1}} \\
&+ \frac{2 K_{D-1}s_{D-1} \left(y^{(0:D-2)}\right)
		\zeta_{D,D-1}
		\left(y^{(0:D-2)}\right)}{N_{D-1} \sqrt{N_D}}
+\frac{K_{D-1}^2 \zeta^2_{D,D-1}\left(y^{(0:D-2)}\right)}{N_D}.
\end{split}
\end{align}
This leads to the following result for the single nesting case
\begin{align}
E_0 \le \frac{\varsigma^2_0}{N_0}+\frac{4 K_{0}^2 \varsigma_1^2}{N_0 N_{1}}
+\frac{2 K_{0}\varsigma_{0} \varsigma_1}{N_{0} \sqrt{N_1}}+\frac{K_0 ^2 \varsigma_1^2}{N_1}
\end{align}
$\approx \frac{\varsigma^2_0}{N_0}+\frac{K_0 ^2 \varsigma_1^2}{N_1} = O\left(\frac{1}{N_0}+\frac{1}{N_1}\right)$
where the approximation becomes exact as $N_0,N_1 \rightarrow \infty$.
Note that there is no $O\left(\epsilon\right)$ term as this bound is exact
in the finite sample case.

Things quickly get messy for double nesting and beyond so we will
ignore non-dominant terms in the limit $N_0,\dots,N_D \rightarrow \infty$
and resort to using $O(\epsilon)$ for these instead. 
We first note that removing dominated terms from~\eqref{eq:general-bound-lip} gives
\begin{align}
	\label{eq:Eklip}
E_k \left(y^{(0:k-1)}\right) \le 
\frac{s_k^2}{N_k} + K_k^2 \omega_k \left(y^{(0:k-1)}\right) + O(\epsilon)
\end{align}
as $s_k^2$ does not decrease with increasing $N_{k+1:D}$ whereas the other
terms do.  We therefore also have
\begin{align}
\omega_k \left(y^{(0:k-1)}\right) &= \E \left[E_{k+1} 
\left(y^{(0:k)}\right) \middle|  y^{(0:k-1)} \right] \nonumber \\ &\le \E \left[ \frac{s_{k+1}^2\left(y^{(0:k)}\right)}
{N_{k+1}} + K^2_{k+1} \omega_{k+1} \left(y^{(0:k)}\right) \middle|
y^{(0:k-1)}\right] + O(\epsilon) \label{eq:omega-bound}
\end{align}
and therefore by recursively substituting~\eqref{eq:omega-bound} into itself
we have
\begin{align}
	\label{eq:Komega}
	K_k^2\omega_k \left(y^{(0:k-1)}\right) &\le
	\sum_{d=k+1}^{D} \frac{\left(\prod_{\ell=k}^{d-1} K_{\ell}^2\right)
		\E \left[ s_{d}^2 \left(y^{(0:d-1)}\right) \middle|
		y^{(0:k-1)}\right]}{N_{d}}+ O(\epsilon).
\end{align}
Now noting that $\zeta_{d,k}^2\left(y^{(0:k-1)}\right) =
\E \left[ s_{d}^2 \left(y^{(0:d-1)}\right) \middle| y^{(0:k-1)}\right]$, substituting
\eqref{eq:Komega} back into~\eqref{eq:Eklip} gives
\begin{align}
E_k \left(y^{(0:k-1)}\right) 
%&\le  \frac{s_k^2\left(y^{(0:k-1)}\right)}{N_k} +
%\sum_{d=k+1}^{D} \frac{\left(\prod_{\ell=k}^{d-1} K_{\ell}^2\right)
%	\E \left[ s_{d}^2 \left(y^{(0:d-1)}\right) \middle|
%	y^{(0:k-1)}\right]}{N_{d}}+ O(\epsilon) \\
&= \frac{s_k^2\left(y^{(0:k-1)}\right)}{N_k} +
\sum_{d=k+1}^{D} \frac{\left(\prod_{\ell=k}^{d-1} K_{\ell}^2\right)
	\zeta_{d,k}^2\left(y^{(0:k-1)}\right)}{N_{d}}+ O(\epsilon).
\label{eq:final-lip-bound}
\end{align}
By definition we have that
$\zeta_{0,0}^2 = s_0^2 =\varsigma_0^2$ and $\zeta_{d,0}^2 = \varsigma_d^2$ and
as~\eqref{eq:final-lip-bound} holds in the case $k=0$, 
the mean squared error of the overall estimator is as follows
\begin{align}
\E \left[\left(I_0-\gamma_0\right)^2\right] 
= E_0 \le 
\frac{\varsigma_{0}^2}{N_0} +
\sum_{k=1}^{D} \frac{\left(\prod_{\ell=0}^{k-1} K_{\ell}^2\right)
	\varsigma_{k}^2}{N_{k}}+ O(\epsilon)
\end{align}
and we have the desired result for the Lipschitz case.

\vspace{15pt}
\noindent\textbf{\large Continuously Differentiable Case}
\vspace{5pt}

\noindent We now revisit the bound for the bias in the continuously differentiable case to
show that a tighter overall bound can be found.  We first remember that
\[
\beta_k \left(y^{(0:k-1)} \right) =
\E \left[f_k\left(y^{(0:k)}, I_{k+1}\left(y^{(0:k)}\right)\right)
- f_k\left(y^{(0:k)}, \gamma_{k+1}\left(y^{(0:k)}\right)\right)
\middle|  y^{(0:k-1)} \right].
\]
Taylor's theorem implies that for any continuously differentiable $f_k$ we can write
\begin{align}
	\begin{split}
f_k&\left(y^{(0:k)}, I_{k+1}\left(y^{(0:k)}\right)\right)
	- f_k\left(y^{(0:k)}, \gamma_{k+1}\left(y^{(0:k)}\right)\right) \\
	&\quad\quad\quad\quad= \frac{\partial f_k \left(y^{(0:k)},\gamma_{k+1}(y^{(0:k)})\right)}{\partial \gamma_{k+1}}
		\left(I_{k+1}\left(y^{(0:k)}\right) - \gamma_{k+1}\left(y^{(0:k)}\right)\right) \\
	&\quad\quad\quad\quad\phantom{=}+\frac{1}{2}
	\frac{\partial f_k^{2} \left(y^{(0:k)},\gamma_{k+1}(y^{(0:k)})\right)}{\partial \gamma_{k+1}^{2}}
	\left(I_{k+1}\left(y^{(0:k)}\right) - \gamma_{k+1}\left(y^{(0:k)}\right)\right)^{2} \\
	&\quad\quad\quad\quad\phantom{=}+h_3\left(I_{k+1}\left(y^{(0:k)}\right) \right) 
	\left(I_{k+1}\left(y^{(0:k)}\right) - \gamma_{k+1}\left(y^{(0:k)}\right)\right)^{3}
	\end{split}
\end{align}
where $\lim_{x\rightarrow\gamma_{k+1}\left(y^{(0:k)}\right)}h_3(x)=0$.  Consequently, the
last term is $O\left(\left(I_{k+1}\left(y^{(0:k)}\right) - \gamma_{k+1}\left(y^{(0:k)}\right)\right)^{3}\right)$
and so will diminish in magnitude faster than the first two terms provided that the derivatives are
bounded, which is guaranteed by our assumptions.  We will thus use $O(\epsilon)$ for this term and note that
it is dominated in the limit.

Now defining
\begin{align*}
	%\delta_{1,k} &= \E \left[\frac{\partial f_k \left(y^{(0:k)},\gamma_{k+1}(y^{(0:k)})\right)}{\partial \gamma_{k+1}}
	%\left(I_{k+1}\left(y^{(0:k)}\right) - \gamma_{k+1}\left(y^{(0:k)}\right)\right)
	%\middle|  y^{(0:k-1)} \right] \\
	\delta_{\ell,k} &= \E \left[ \frac{\partial f_k^{\ell} \left(y^{(0:k)},\gamma_{k+1}(y^{(0:k)})\right)}{\partial \gamma_{k+1}^{\ell}}
	\left(I_{k+1}\left(y^{(0:k)}\right) - \gamma_{k+1}\left(y^{(0:k)}\right)\right)^{\ell}
	\middle|  y^{(0:k-1)} \right] 
\end{align*}
then we have
\begin{align*}
	\beta_k^2 \left(y^{(0:k-1)} \right)
	=&\delta_{1,k}^2+\frac{\delta_{2,k}^2}{4}+\delta_{1,k}\delta_{2,k} +O(\epsilon).
\end{align*}
%We now use a Taylor expansion for $f_k\left(y^{(0:k)}, I_{k+1}\left(y^{(0:k)}\right)\right)$ about the point
%$f_k\left(y^{(0:k)}, \gamma_{k+1}\left(y^{(0:k)}\right)\right)$
%as follows where all derivatives are with respect to the second input.  First we
%define
%\begin{align*}
%%\delta_{1,k} &= \E \left[\frac{\partial f_k \left(y^{(0:k)},\gamma_{k+1}(y^{(0:k)})\right)}{\partial \gamma_{k+1}}
%%\left(I_{k+1}\left(y^{(0:k)}\right) - \gamma_{k+1}\left(y^{(0:k)}\right)\right)
%%\middle|  y^{(0:k-1)} \right] \\
%\delta_{\ell,k} &= \E \left[ \frac{\partial f_k^{\ell} \left(y^{(0:k)},\gamma_{k+1}(y^{(0:k)})\right)}{\partial \gamma_{k+1}^{\ell}}
%\left(I_{k+1}\left(y^{(0:k)}\right) - \gamma_{k+1}\left(y^{(0:k)}\right)\right)^{\ell}
%\middle|  y^{(0:k-1)} \right] 
%\end{align*}
%then we have from the Taylor expansion
%\begin{align*}
%\beta_k^2 \left(y^{(0:k-1)} \right) =& 
%\left(\E \left[f_k\left(y^{(0:k)}, \gamma_{k+1}\left(y^{(0:k)}\right)\right)
%\middle|  y^{(0:k-1)} \right]+
%\delta_{1,k}+\frac{\delta_{2,k}}{2}
%+O(\epsilon) \right.\\
%&\left.\quad-\E \left[f_k\left(y^{(0:k)}, \gamma_{k+1}\left(y^{(0:k)}\right)\right)
%\middle|  y^{(0:k-1)} \right]\right)^2 \\
%=&\delta_{1,k}^2+\frac{\delta_{2,k}^2}{4}+\delta_{1,k}\delta_{2,k} +O(\epsilon).
%\end{align*}
By using the tower property we further have that
\begin{align*}
\delta_{\ell,k} &= \E \left[ \E \left[\frac{\partial f_k^{\ell} \left(y^{(0:k)},\gamma_{k+1}(y^{(0:k)})\right)}{\partial \gamma_{k+1}^{\ell}}
\left(I_{k+1}\left(y^{(0:k)}\right) - \gamma_{k+1}\left(y^{(0:k)}\right)\right)^{\ell}
\middle|  y^{(0:k)} \right] \middle|  y^{(0:k-1)} \right] \displaybreak[0]\\
&= \E \left[ \frac{\partial f_k^{\ell} \left(y^{(0:k)},\gamma_{k+1}(y^{(0:k)})\right)}{\partial \gamma_{k+1}^{\ell}} \E \left[
\left(I_{k+1}\left(y^{(0:k)}\right) - \gamma_{k+1}\left(y^{(0:k)}\right)\right)^{\ell}
\middle|  y^{(0:k)} \right] \middle|  y^{(0:k-1)} \right] \displaybreak[0]\\
&\le \E \left[ \left|\frac{\partial f_k^{\ell} \left(y^{(0:k)},\gamma_{k+1}(y^{(0:k)})\right)}{\partial \gamma_{k+1}^{\ell}} \right| 
\left|\E \left[
\left(I_{k+1}\left(y^{(0:k)}\right) - \gamma_{k+1}\left(y^{(0:k)}\right)\right)^{\ell}
\middle|  y^{(0:k)} \right] \right| \; \middle|  y^{(0:k-1)} \right]  \displaybreak[0]\\
&\le \left(\sup_{y^{(0)}} \left|
\frac{\partial^{\ell} f_k \left(y^{(0:k)},\gamma_{k+1}(y^{(0:k)})\right)}{\partial \gamma^{\ell}_{k+1}} \right| \right)
\E \left[\left| \E \left[
\left(I_{k+1}\left(y^{(0:k)}\right) - \gamma_{k+1}\left(y^{(0:k)}\right)\right)^{\ell}
\bigg| y^{(0:k)} \right]
\right| \; \Bigg| y^{(0:k-1)} \right].
\end{align*}
Now for the particular cases of $\ell=1$ and $\ell=2$ then the derivative terms where 
defined in the theorem and the expectations correspond respectively to our definitions of $\lambda_k$ and $\omega_k$
giving
\begin{align*}
\delta_{1,k} &\le  K_k \lambda_k\left(y^{(0:k-1)} \right) \\
\delta_{2,k} &\le C_k \omega_k\left(y^{(0:k-1)} \right)
\end{align*}
and therefore
%\begin{align}
%\beta_k^2 \left(y^{(0:k-1)} \right)  \le 
%K_k^2 \omega_k\left(y^{(0:k-1)} \right) + 
%\frac{C_k^2}{4} \omega_k^2\left(y^{(0:k-1)} \right)
%+K_k C_k \omega_k^{\frac{3}{2}}\left(y^{(0:k-1)} \right)+O(\epsilon)
%\end{align}
\begin{align}
\beta_k^2 \left(y^{(0:k-1)} \right)  &\le K_k^2 \lambda_k^2 \left(y^{(0:k-1)} \right) +\frac{C_k^2}{4} \omega_k^2\left(y^{(0:k-1)} \right) 
+K_k\; C_k \; \lambda_k \left(y^{(0:k-1)} \right) \omega_k \left(y^{(0:k-1)} \right) + O(\epsilon) \nonumber\\
&=\left(K_k \lambda_k \left(y^{(0:k-1)} \right) +
\frac{C_k}{2} \omega_k\left(y^{(0:k-1)} \right) \right)^2+O(\epsilon). \label{eq:beta-k-cont}
\end{align}

Remembering~\eqref{eq:E_decomp}
%$E_k \left(y^{(0:k-1)}\right) 
%= \frac{\sigma_k^2 \left(y^{(0:k-1)}\right)}{N_k} + \beta_k^2 \left(y^{(0:k-1)}\right)$,
we can recursively define the error bound in the same manner as the Lipschitz case.  We can immediately see that,
as $\beta_D =0$ without any nesting, we recover the bound from the Lipschitz case for the inner most estimator as expected.  
As the innermost estimator is unbiased we also have
$\lambda_{D-1} \left(y^{(0:D-2)} \right)=0$ and so
\begin{align*}
\beta_{D-1}^2 \left(y^{(0:D-2)} \right) &\le \frac{C_{D-1}^2}{4} \omega^2_{D-1} \left(y^{(0:D-2)}\right) + O(\epsilon) \\
&\le \frac{C_{D-1}^2}{4} 
\left(\E \left[\frac{s_D^2\left(y^{(0:D-1)}\right)}{N_D}
\middle|  y^{(0:D-2)} \right]\right)^2+ O(\epsilon) \\
&= \frac{C_{D-1}^2 \; \zeta^4_{D,D-1}
	\left(y^{(0:D-2)}\right) }{4N_D^2}+ O(\epsilon).
\end{align*}
Going back to our original bound on $\sigma_{D-1}^2 \left(y^{(0:D-2)}\right)$ given in~\eqref{eq:sigma_bound}
and substituting for $\beta_{D-1} \left(y^{(0:D-2)} \right)$ we now have
\begin{align}
\sigma_{D-1} \left(y^{(0:D-2)} \right) &\le\left(A_{D-1} \left(y^{(0:D-2)}\right)\right)^{\frac{1}{2}}
+\left(s_{D-1}^2 \left(y^{(0:D-2)} \right) +
\frac{C_{D-1}^2 \; \zeta^4_{D,D-1}
	\left(y^{(0:D-2)}\right) }{4N_D^2}+ O(\epsilon)
\right)^{\frac{1}{2}}.
\end{align}
There does not appear to be tighter bound for $A_{D-1} \left(y^{(0:D-2)}\right)$ than in the Lipschitz continuous case
and so using the same bound of $A_{D-1} \left(y^{(0:D-2)}\right) \le K_{D-1}^2 \zeta^2_{D,D-1}	\left(y^{(0:D-2)}\right) / N_{D-1}$ we have
\begin{align}
 E_{D-1}\left(y^{(0:D-2)}\right) \le& \frac{\sigma_{D-1}^2 \left(y^{(0:D-2)} \right) }{N_{D-1}}+
 \frac{C_{D-1}^2 \; \zeta^4_{D,D-1}	\left(y^{(0:D-2)}\right) }{4N_D^2}+ O(\epsilon) \displaybreak[0] \nonumber\\
 \begin{split}
 \le&  \frac{s_{D-1}^2 \left(y^{(0:D-2)}\right)}{N_{D-1}} +
\frac{K_{D-1}^2 \zeta^2_{D,D-1}	\left(y^{(0:D-2)}\right)}{N_D N_{D-1}}
+\frac{C_{D-1}^2 \; \zeta^4_{D,D-1}  	\left(y^{(0:D-2)}\right) }{4N_D^2}\left(1+\frac{1}{N_{D-1}}\right) \\
&+ \frac{2 K_{D-1}\zeta_{D,D-1}	\left(y^{(0:D-2)}\right)}{N_{D-1} \sqrt{N_D}} \left(s_{D-1} \left(y^{(0:D-2)}\right)^2+
\frac{C_{D-1}^2 \; \zeta^4_{D,D-1}  	\left(y^{(0:D-2)}\right) }{4N_D^2}\right)^{\frac{1}{2}}
+ O(\epsilon).
\end{split}
\end{align}
Therefore for the single nesting case, we now have
\begin{align}
E_0 \le \frac{\varsigma^2_0}{N_0}+\frac{K_{0}^2 \varsigma_1^2}{N_0 N_{1}}
+\frac{2 K_{0}\varsigma_1}{N_{0} \sqrt{N_1}}\sqrt{\varsigma_0^2+\frac{C_0 ^2 \varsigma_1^4}{4 N_1^2}}+\frac{C_0 ^2 \varsigma_1^4}{4 N_1^2}\left(1+\frac{1}{N_0}\right)
+ O\left(\frac{1}{N_1^3}\right)
\end{align}
$\approx \frac{\varsigma^2_0}{N_0}+\frac{C_0 ^2 \varsigma_1^4}{4 N_1^2} = O\left(\frac{1}{N_0}+\frac{1}{N_1^2}\right)$
where again the approximation becomes tight when $N_0,N_1 \rightarrow \infty$.
Here we have used the fact that the only $O(\epsilon)$ term comes from the Taylor expansion and is equal to $O\left(\frac{1}{N_1^3}\right)$ 
because we have $\delta_{1,D-1}=0$ and therefore
\begin{align*}
O(\epsilon)=& O\left(\delta_{2,D-1}\delta_{3,D-1}+\delta_{2,D-1}\delta_{4,D-1}\right)\\
=& O\left(\delta_{2,D-1}\E \left[
\left(I_{1}\left(y^{(0)}\right) - \gamma_{1}\left(y^{(0)}\right)\right)^{3}\middle| y^{(0)}\right]\right)
+ O\left(\delta_{2,D-1}\E \left[
\left(I_{1}\left(y^{(0)}\right) - \gamma_{1}\left(y^{(0)}\right)\right)^{4}\middle| y^{(0)}\right]\right) \\
=& O\left(\frac{1}{N_1}\E \left[
\left(\frac{1}{N_1}\sum_{n=1}^{N_1} f_1\left(y^{(0:1)}_n\right)- \E \left[f_1\left(y^{(0:1)}\right) \middle| y^{(0)}\right]\right)^{3}\middle| y^{(0)}\right]\right) \\
&+ O\left(\frac{1}{N_1}\E \left[
\left(\frac{1}{N_1}\sum_{n=1}^{N_1} f_1\left(y^{(0:1)}_n\right)- \E \left[f_1\left(y^{(0:1)}\right) \middle| y^{(0)}\right]\right)^{4}\middle| y^{(0)}\right]\right) \\
\intertext{now noting that the $y^{(0:1)}_n$ are i.i.d., and that $\E\left[f_1\left(y^{(0:1)}_1\right)- \E \left[f_1\left(y^{(0:1)}\right) \middle| y^{(0)}\right]\middle| y^{(0)}\right]=0$ such many of the cross terms when expanding the brackets are zero, we have}
=& O\left(\frac{1}{N_1^4} \sum_{n=1}^{N_1} \E \left[
\left( f_1\left(y^{(0:1)}_1\right)- \E \left[f_1\left(y^{(0:1)}\right) \middle| y^{(0)}\right]\right)^{3}\middle| y^{(0)}\right]\right) \\
&+O\left(\frac{1}{N_1^5} \sum_{n=1}^{N_1} \E \left[
\left( f_1\left(y^{(0:1)}_1\right)- \E \left[f_1\left(y^{(0:1)}\right) \middle| y^{(0)}\right]\right)^{4}\middle| y^{(0)}\right]\right) \\
&+O\left(\frac{3}{N_1^5} \sum_{n=1}^{N_1} \sum_{m=1,m\neq n}^{N_1} \left( \E \left[
\left( f_1\left(y^{(0:1)}_1\right)- \E \left[f_1\left(y^{(0:1)}\right) \middle| y^{(0)}\right]\right)^{2}\middle| y^{(0)}\right]\right)^2\right) \displaybreak[0] \\
=& O\left(\frac{1}{N_1^3}\right) + O\left(\frac{1}{N_1^4}\right)+ O\left(\frac{1}{N_1^3}\right) = O\left(\frac{1}{N_1^3}\right)
\end{align*}
as required.
%as $\E \left[
%\left(I_{1}\left(y^{(0)}\right) - \gamma_{1}\left(y^{(0)}\right)\right)^4
%\right]=O\left(\frac{1}{N_1^{2}}\right)$.  Therefore $\frac{\delta_{2,k} \delta_{3,k}}{3}
%=O\left(\frac{1}{N_1^{5/2}}\right)$ which confirms that the $O(\epsilon)$ term is
%asymptotically negligible as required and allows us to better characterize the
%bound in the single nesting case as we have
%\begin{align}
%E_0 \le \frac{\varsigma^2_0}{N_0}+\frac{4 K_{0}^2 \varsigma_1^2}{N_0 N_{1}}
%+\frac{2 K_{0}\varsigma_{0} \varsigma_1}{N_{0} \sqrt{N_1}}+\frac{C_0 ^2 \varsigma_1^4}{4 N_1^2}
%+ O\left(\frac{1}{N_1^{5/2}}\right).
%\end{align}

Returning to calculating the bound for the repeated nesting case then by 
substituting~\eqref{eq:beta-k-cont} into~\eqref{eq:E_decomp} we have more generally
\begin{align}
	E_k \left(y^{(0:k-1)}\right) 
	&\le \frac{s_k^2 \left(y^{(0:k-1)}\right)}{N_k}+ \left(K_k \lambda_k \left(y^{(0:k-1)} \right) +
	\frac{C_k}{2} \omega_k\left(y^{(0:k-1)} \right) \right)^2+O(\epsilon). \label{eq:Ekcont}
	\end{align}
Now remembering that
$\omega_k \left(y^{(0:k-1)}\right) = \E \left[E_{k+1} 
\left(y^{(0:k)}\right) \middle|  y^{(0:k-1)} \right] $ from~\eqref{eq:E_decomp} we have
\begin{align}
	\omega_k\left(y^{(0:k-1)} \right) &= \E \left[\frac{s_{k+1}^2 \left(y^{(0:k)}\right)}{N_{k+1}} + \beta^2_{k+1} 
	\left(y^{(0:k)}\right) \middle|  y^{(0:k-1)} \right] + O(\epsilon) \nonumber \\
	&= \frac{\zeta_{k+1,k}^2}{N_{k+1}}+\E \left[\beta^2_{k+1} 
	\left(y^{(0:k)}\right) \middle|  y^{(0:k-1)} \right] + O(\epsilon).  \label{eq:omega}
\end{align}
We also have that except at $k=D-1$ and $k=D$ (for which both $\lambda_k$ and $\beta_{k+1}$ are zero), then
\[
\lambda_k\left(y^{(0:k-1)} \right) = \E \left[\left|\beta_{k+1} 
\left(y^{(0:k)}\right) \right| \Bigg|  y^{(0:k)}\right] \gg
\E \left[\beta^2_{k+1} 
\left(y^{(0:k)}\right) \middle|  y^{(0:k-1)} \right]
\]
for sufficiently large $N_{k+1},\dots,N_D$.
This means that when we substitute~\eqref{eq:omega} into~\eqref{eq:Ekcont},
the second term in \eqref{eq:omega} becomes dominated giving
\begin{align}
	E_k \left(y^{(0:k-1)}\right) 
	\le\frac{s_k^2 \left(y^{(0:k-1)}\right)}{N_k}
	+\left(K_k \lambda_k \left(y^{(0:k-1)}\right) 
	+\frac{C_k \zeta_{k+1,k}^2}{2 N_{k+1}}\right)^2 + O(\epsilon). \label{eq:E_k_for_sub}
\end{align}
 Now as $\beta_{k+1}^2 \left(y^{(0:k)}\right) =E_{k+1} 
 \left(y^{(0:k)}\right) -\frac{s_{k+1}^2 \left(y^{(0:k)}\right)}{N_{k+1}}$ we have
 \begin{align*}
  \lambda_k \left(y^{(0:k-1)}\right)  &=
  \E \left[\sqrt{E_{k+1} 
  \left(y^{(0:k)}\right) -\frac{s_{k+1}^2 \left(y^{(0:k)}\right)}{N_{k+1}}}
  \middle|  y^{(0:k-1)}  \right] +O(\epsilon)\\
  \intertext{and substituting in~\eqref{eq:E_k_for_sub} gives}
 \lambda_k \left(y^{(0:k-1)}\right)  &\le \E \left[K_{k+1} \lambda_{k+1} \left(y^{(0:k)}\right) 
  +\frac{C_{k+1} \zeta_{k+2,k+1}^2}{2 N_{k+2}} \middle|  y^{(0:k-1)} \right] 
  +O(\epsilon) \displaybreak[0]\\
 &= \frac{C_{k+1} \zeta_{k+2,k}^2}{2 N_{k+2}}
 +K_{k+1} \E \left[\lambda_{k+1} \left(y^{(0:k)}\right)  \middle|  y^{(0:k-1)} \right] +O(\epsilon) \displaybreak[0]\\
  &\le \frac{C_{k+1} \zeta_{k+2,k}^2}{2 N_{k+2}}+
  \sum_{d=k+1}^{D-2}  \E \left[\left(\prod_{\ell=k+1}^{d} K_{\ell}\right)
  \frac{C_{d+1} \zeta^2_{d+2,d}}{2 N_{d+2}} \middle|  y^{(0:k-1)} \right] +O(\epsilon) \displaybreak[0]\\
 &\le \frac{C_{k+1} \zeta_{k+2,k}^2}{2 N_{k+2}}+
 \sum_{d=k+1}^{D-2}  \left(\prod_{\ell=k+1}^{d} K_{\ell}\right)
 \frac{C_{d+1} \zeta^2_{d+2,k}}{2 N_{d+2}}+O(\epsilon)
 \end{align*}
 and thus
 \begin{align*}
 E_k \left(y^{(0:k-1)}\right)  \le \frac{s_k^2 \left(y^{(0:k-1)}\right)}{N_k}
 +\frac{1}{4}\left(
 \frac{C_k \zeta_{k+1,k}^2}{N_{k+1}}
 +\sum_{d=k}^{D-2}  \left(\prod_{\ell=k}^{d} K_{\ell}\right)
  \frac{C_{d+1} \zeta^2_{d+2,k}}{N_{d+2}}
 \right)^2 + O(\epsilon).
 \end{align*}
 and therefore
  \begin{align*}
  \E \left[\left(I_0-\gamma_0\right)^2\right] 
  = E_0  \le \frac{\varsigma_0^2}{N_0}
  +\frac{1}{4}\left(
  \frac{C_0 \varsigma_{1}^2}{N_{1}}
  +\sum_{k=0}^{D-2}  \left(\prod_{d=0}^{k} K_{d}\right)
  \frac{C_{k+1} \varsigma^2_{k+2}}{N_{k+2}}
  \right)^2 + O(\epsilon)
  \end{align*}
  as required and we are done.

\end{proof}
%\input{bias_proof}
% !TEX root =  main.tex

\section{Proof of Theorem~\ref{the:finite-res} - Convergence Rate for Finite Realisations of $y$}
\label{sec:app:finite-res}

%\begin{theorem}
%	If $f$ is Lipschitz continuous, then the mean squared error of 
%	\[
%	I_N = \sum_{c=1}^C (\hat{P}_N)_c \, (\hat{f}_N)_c
%	\]
%	as an estimator for $I$ converges at rate $O(1/N)$.
%\end{theorem}
\thefiniteres*

\begin{proof}
	Denote
	$P_c = P(y = y_c)$ and
	$f_c = f(y_c, \gamma(y_c))$
	noting that as the $y_c$ are fixed values, so are $P_c$ and $f_c$.
	Then, Minkowski's inequality allows us to bound the mean squared error as
	\begin{eqnarray*}
		\E\left[(I_N - I)^2\right] = \norm{I_N - I}_2^2
		\leq \left(\sum_{c=1}^C W_c \right)^2 \quad \text{where} \quad W_c := \norm{(\hat{P}_N)_c \, (\hat{f}_N)_c - P_c \, f_c}_2.
	\end{eqnarray*}
	Moreover, again by Minkowski, we have $W_c \leq U_c + V_c$
	where 
	\[
	U_c = \norm{(\hat{P}_N)_c \, (\hat{f}_N)_c - (\hat{P}_N)_c \, f_c}_2, \quad
	V_c = \norm{(\hat{P}_N)_c \, f_c - P_c \, f_c}_2.
	\]
	Factoring out $(\hat{P}_N)_c$ in $U_c$ and noting that each $y_n$ and $z_{n,c}$ are sampled independently gives
	\begin{eqnarray*}
		U_c = \sqrt{\E\left[(\hat{P}_N)_c^2 \, \left((\hat{f}_N)_c - f_c\right)^2\right]} 
		= \sqrt{\E\left[(\hat{P}_N)_c^2\right]} \sqrt{\E\left[\left((\hat{f}_N)_c - f_c\right)^2\right]}.
	\end{eqnarray*}
	Using Minkowski's
	inequality, we may write the first right-hand term as
	\begin{align*}
		\sqrt{\E\left[(\hat{P}_N)_c^2\right]} = \norm{(\hat{P}_N)_c}_2
		&\leq \frac{1}{N} \sum_{n=1}^N \norm{\mathbbm{1}(y_n = y_c)}_2 
		=\frac{1}{N} \sum_{n=1}^N \E \left[{\mathbbm{1}(y_n = y_c)}^2\right] 
		= \frac{1}{N} \sum_{n=1}^N P_c
		= P_c.
	\end{align*}
	For the second term, note that by Lipschitz continuity, we have for some constant $K >
	0$
	\begin{eqnarray*}
		\sqrt{\E\left[\left((\hat{f}_N)_c - f_c\right)^2\right]} = \norm{(\hat{f}_N)_c - f_c}_2 \leq
		K \, \norm{\frac{1}{N} \sum_{n=1}^N \phi(y_c, z_{n,c}) - \gamma(y_c)}_2  
		= K \cdot O(1/\sqrt{N}) 
		= O(1/\sqrt{N}),
	\end{eqnarray*}
	since $\frac{1}{N} \sum_{n=1}^N \phi(y_c, z_{n,c})$ is a Monte Carlo estimator for
	$\gamma(y_c)$. Altogether then, we have that
	\[
	U_c = P_c \cdot O(1 / \sqrt{N}) = O(1 / \sqrt{N}).
	\]
	We can also factor out $f_c$ in $V_c$ to obtain
	\[
	V_c = |f_c| \cdot \norm{(\hat{P}_N)_c - P_c}_2 = |f_c| \cdot O(1/\sqrt{N}) = O(1/\sqrt{N}),
	\]
	since $(\hat{P}_N)_c$ is a Monte Carlo estimator for $P_c$.
	Now by noting that 
	$(A+B)^2 \le 2(A^2+B^2)$
	for any $A, B \in \mathbb{R}$, an inductive argument shows that 
	\[
	\left(\sum_{\ell=1}^L A_\ell\right)^2 \leq 2^{\lceil \log_2 L \rceil} \sum_{\ell=1}^L A_\ell^2
	\]
	for all $A_1, \cdots, A_L \in \mathbb{R}$.  We can now show that
	our asymptotic bounds for $U_c$ and $V_c$ entail that our overall mean squared
	error satisfies
	\begin{eqnarray}
	\E\left[(I_N - I)^2\right] &\leq& 2^{\lceil \log_2 C \rceil} \sum_{c=1}^C W_c^2 \label{eq:square-inequal-1} \notag
	\leq 2^{\lceil \log_2 C \rceil} \sum_{c=1}^C (U_c + V_c)^2 \notag
	\leq 2^{\lceil \log_2 C \rceil + 1} \sum_{c=1}^C U_c^2 + V_c^2 \label{eq:square-inequal-2} \notag \\
	&=& 2^{\lceil \log_2 C \rceil + 1} \sum_{c=1}^C O(1/N) + O(1/N) \notag
	= O(1/N), \notag
	\end{eqnarray}
	as desired.
\end{proof}
% !TEX root =  main.tex

\section{Proof for Theorem~\ref{the:prod} - Products of Expectations}
\label{sec:app-prod}

\theprod*

\begin{proof}
	Consider fixed sizes of nested sample sets, $\{M_\ell\}_{\ell = 1:L}$.
	For each $y \in \mathcal{Y}$ and 
        \[
                x = \{\{z_{\ell,m}'\}_{m=1:M_{\ell}}\}_{\ell=1:L} \in \mathcal{X} 
                = \mathcal{Z}_1^{M_{1}} \otimes \dots \otimes \mathcal{Z}_L^{M_L},
        \]
	define 
	\[
	\eta(y,x) = f\left(y,\prod_{\ell=1}^{L} \frac{1}{M_{\ell}} \sum_{m_{\ell}=1}^{M_{\ell}} \psi_{\ell}(y,z_{\ell,m}')\right).
	\]
	Now, $I_{N} = \frac{1}{N} \sum_{n=1}^{N} \eta(y_n,x_n)$ is a standard
	MC estimator on the space $\mathcal{Y} \otimes \mathcal{X}$. Thus,
	$I_{N} \asto \mathbb{E}[I_{N}]$ with convergence properties and rate as per standard MC.  
	We finish the proof by showing that $\mathbb{E}[I_{N}]=I$ when $f$ is linear:
	\begin{align*}
	\displaybreak[0]
	\mathbb{E}[I_{N}] &= \mathbb{E} \left[\frac{1}{N} \sum_{n=1}^N f\left(y_n,\prod_{\ell=1}^{L} \frac{1}{M_{\ell}}  \sum_{m=1}^{M_{\ell}} \psi_{\ell}(y_n,z_{n,\ell,m}') \right)\right] 
	= \mathbb{E}\left[ \mathbb{E}\left[ f\left( y_1,\prod_{\ell=1}^{L} \frac{1}{M_{\ell}}  \sum_{m=1}^{M_{\ell}} \psi_{\ell}(y_1,z_{1,\ell,m}')\right) \middle| y_1 \right]\right],
	\intertext{now using the linearity of $f$}
	\displaybreak[0]
	\phantom{\mathbb{E}[I_{N}]} &= \mathbb{E}\left[ f\left( y_1,\mathbb{E}\left[ \prod_{\ell=1}^{L} \frac{1}{M_{\ell}}  \sum_{m=1}^{M_{\ell}} \psi_{\ell}(y_1,z_{1,\ell,m}') \middle| y_1 \right] \right) \right],
	\intertext{and using the fact that terms for different $\ell$ are by construction independent}
	\displaybreak[0]
	\phantom{\mathbb{E}[I_{N}]} 
	&= \mathbb{E}\left[ f\left( y_1, \prod_{\ell=1}^L\mathbb{E}\left[\frac{1}{M_{\ell}}  \sum_{m=1}^{M_{\ell}} \psi_{\ell}(y_1,z_{1,\ell,m}') \middle| y_1 \right] \right) \right] 
	= \mathbb{E}\left[ f\left( y_1, \prod_{\ell=1}^{L}\mathbb{E}\left[\psi_{\ell}(y_1,z'_{1,\ell,1}) \middle| y_1 \right]\right)\right] 
	= I,
	\end{align*}
	as required.
\end{proof}

%\input{convex.tex}
% !TEX root =  main.tex

\section{Optimizing the Convergence Rates}
\label{sec:app:opt-conv}

We have shown that the mean squared error converges at a rate \[
O\left(\sum_{k=0}^{D} \frac{1}{N_k}\right) \quad \mathrm{or} \quad 
O\left(\frac{1}{N_0} +\left(\sum_{k=1}^{D} \frac{1}{N_k}\right)^2\right)
\]
 depending on the smoothness assumptions that can be made about
$f$.  Here we show that given a sample budget for the inner most estimator $T=\prod_{k=0}^{D} N_k$, then these bounds are
optimized by setting $N_0 \propto N_1 \propto \dots \propto N_D$ and 
$N_0 \propto N_1^2 \propto \dots \propto N_D^2$ respectively for the two cases and that this
gives bounds of $O\left(1/T^{\frac{1}{D+1}}\right)$ and $O\left(1/T^{\frac{2}{D+2}}\right)$respectively.  For the single
nested case, this gives bounds of $O(1/\sqrt{T})$ and $O(1/T^{2/3})$ respectively.

We start by explaining why $T$ is an appropriate measure of the overall computational cost.  
First note that for each sample of $y^{(0:k)}$,
 the NMC estimator requires $N_k$ samples of $y^{(k+1)}$.  
Thus there are $N_0$ samples of the outermost level, $N_0 \times N_1$ of the next level, and 
$T = \prod_{k=0}^{D} N_k$ samples of the innermost level, regardless of the setup.
In other words, each individual estimate of the innermost level uses $N_D$ 
samples and we generate $\prod_{k=0}^{D-1} N_k = T/ N_D$ of these estimates 
because we need to generate one estimate for each sample of the layer above.  Thus what we can vary 
for a fixed $T$ is whether we use more estimates each using fewer samples, or fewer estimates each using more samples.

Now the total cost of generating $I_0$ scales with sum the costs of each individual layer, namely
\[
\text{Cost} = \sum_{k=0}^{D} c_{k} \prod_{\ell=0}^{k} N_\ell
\]
where $c_{k}$ is the per sample cost local computations made at the $k^{\text{th}}$ layer (i.e. sampling
$y^{(0:k)}$ and evaluating $f_k$ for given inputs), which is independent of the $N_k$.  For large $N_D$, we see that the dominant
cost is that of the inner most layer, namely $c_{T} \prod_{\ell=0}^{D} N_\ell = c_{T}T$, and
we asymptotically spend 100\% of our time dealing with the innermost estimator.  To give intuition to this,
 think about writing the estimator 
as a hierarchy of nested for loops; as the length of the loops increases we spend an increasing proportion of our time 
inside the innermost loop.  Consequently, in the asymptotic regime, our computational cost is $O(T)$ and we can
use $T$ is an appropriate measure of the overall computational cost.  

To derive the optimal rates, we first consider the single nested case where $D=1$, $N_0=N$, and $N_1=M$.
Consider setting $N = \tau(M)$ then $T = \tau(M) \cdot M$ and our bounds become $O(R)$, where
\[
  R = 1/\tau(M) + 1/M \quad \text{and} \quad R = 1/\tau(M) + 1/M^2.
\]
for the two cases respectively.

In this first case supposing $\tau(M) = O(M)$ easily gives
\[
  T = M \tau(M)
    = O\left(M^2\right)
\]
and as such 
\begin{equation} \label{eq:rm}
  R = O\left(\frac{1}{M}\right) = O\left(\frac{1}{\sqrt{T}}\right)
\end{equation}
as $M \to \infty$.  In contrast, consider the case that $\tau(M) \gg M$ as $M \to \infty$. We then have
$\frac{1}{\sqrt{M}} \gg \frac{1}{\sqrt{\tau(M)}}$ as $M \to \infty$, so that
\[
  R = O\left(\frac{1}{M}\right) \gg \frac{1}{\sqrt{M}} \frac{1}{\sqrt{\tau(M)}} = \frac{1}{\sqrt{T}}.
\]
Comparing with \eqref{eq:rm}, we observe that, for the same total
budget of samples $T$, this choice of $\tau$ provides a strictly weaker convergence
guarantee than in the previous case. When $M \gg \tau(M)$ also then we have
$\frac{1}{\sqrt{\tau(M)}} \gg \frac{1}{\sqrt{M}}$ as $M \to \infty$ and so
\[
R = O\left(\frac{1}{\tau(M)}\right) \gg \frac{1}{\sqrt{M}} \frac{1}{\sqrt{\tau(M)}} = \frac{1}{\sqrt{T}}
\]
which is again a weaker bound.  We thus see that the $O(1/N + 1/M)$ bound is optimized when
$N \propto M$, giving a convergence rate of $O(1/\sqrt{T})$.

In the second case suppose that $\tau(M) = O(M^2)$ as $M \to \infty$.  This now gives
\[
  T = M \tau(M)
  = O\left(M^3\right)
\]
and therefore
\[
R = O\left(\frac{1}{M^2}\right) = O\left(\frac{1}{T^{2/3}}\right)
\]
as $M \to \infty$.  Now considering the cases $\tau(M) \gg M^2$ leads to $\frac{1}{M^{4/3}}  \gg \frac{1}{\tau(M)^{2/3}}$ and thus
\[
R = O\left(\frac{1}{M^2}\right) \gg \frac{1}{M^{2/3}} \frac{1}{\tau(M)^{2/3}} = \frac{1}{T^{2/3}}.
\]
Similarly, if $\tau(M) \ll M^2$ then $\frac{1}{\tau(M)^{1/3}}\gg \frac{1}{M^{2/3}}$ and thus
\[
R = O\left(\frac{1}{\tau(M)}\right) \gg \frac{1}{M^{2/3}} \frac{1}{\tau(M)^{2/3}} = \frac{1}{T^{2/3}}.
\]
Both of these cases lead to weaker bounds and so we see that the $O(1/N + 1/M^2)$ bound
is tightest when $N \propto M^2$, giving a convergence rate of $O(1/T^{2/3})$.

We now consider the repeated nesting case without continuously differentiability such that our
bound is $O\left(\sum_{k=0}^{D} \frac{1}{N_k}\right)$.  Here we can immediately see that
$N_0 \propto N_1 \propto \dots \propto N_D$ leads to $N_k \propto T^\frac{1}{D+1}$ and thus
 $O\left(1/T^{\frac{1}{D+1}}\right)$ convergence.  If we were to set any $N_k \ll T^{\frac{1}{D+1}}$ then this term would
 dominate the sum and lead to a worse converge.  Thus the result from the single nested case
 trivially extends to the multiple nested case, giving the required result.
 
 Finally considering repeated nesting for the bound $O\left(\frac{1}{N_0} +\left(\sum_{k=1}^{D} \frac{1}{N_k}\right)^2\right)$
 then we have from the previous result that $N_1 \propto N_2 \propto \dots \propto N_D$
 is required for optimality as otherwise one of the terms in the summation dominates the
 other terms.  If we now define $M=\prod_{k=1}^{D} N_k = T/N_0$ then we get a convergence
 rate of $O(1/N_0+1/M^2)$ which is identical to the single nesting case for this tighter
 bound.  We, therefore, have that the optimal configuration must be
 $N_0 \propto N_1^2 \propto \dots \propto N_D^2$ giving a bound of $O\left(1/T^{\frac{2}{D+2}}\right)$
 as it gives $N_0 \propto T^{\frac{2}{D+2}}$.
%\input{repeat-nest}
% !TEX root =  main.tex

\vspace{-5pt}

\section{Additional details pertaining to cancer simulator}
\label{sec:cancer_sim_app}

In this section, we elucidate some more details about the cancer simulator described in the manuscript, provide more rigorous mathematical definitions for the relevant terms using the same nomenclature, and also include more results figures.

\vspace{-5pt}

\subsection{Simulator details}
\label{sec:cancer_sim_details}

We define $I(T_{\text{treat}})$ to be the expected proportion of patients who receive treatment.
A particular patient is represented by $y\in\mathcal{R}^d$. Specifically, $y$ consists of only a single real number ($d=1$) representing the size of the tumor upon discovery. Initial tumor size is drawn from a scaled Rayleigh distribution.
The outcome of the simulator is then $\phi(y,z)\in\left\lbrace 0,1\right\rbrace $, and is the binary outcome of whether that particular patient and sample of unobserved parameters yield an expected tumor size below the threshold, $T_{opp}$, after a fixed time duration, $t_{max}$. The simulator is a pair of coupled, parameterized differential equations for the action of an anti-tumor treatment such as chemotherapy, as described in~\citet{enderling2014cancer}:
\begin{align}
&\frac{dc}{dt} = -\lambda c \log\big(\frac{c}{K}\big) - \xi c \\
&\frac{dK}{dt} = \phi c - \psi K c^{2/3},
\end{align}
where $c(t,x)\in\mathcal{R}_+$ represents tumor size, with initial size $y_n$.
Similarly, $K(t,x)\in\mathcal{R}_+$ represents the notion of a carrying capacity, with the initial carrying capacity, $K(0,z)$, set to a known constant $K_0$.
The magnitude of the patient response to an anti-tumor treatment (such as chemotherapy) is represented by $\xi\in[0,1]$, drawn from a beta distribution.
$\{\lambda, \psi, \phi\}\in\mathcal{R}_+^3$ represent the parameters of the simulator.
We also define $x _{n,m}= \{\lambda, \psi, \phi\, K_0, \xi\}$ and $z_{n,m}=\{x_{n,m}, T_{\text{opp}}, t_{\text{max}}\}$, where all but £$\xi$ are set to constant values. Expanding this to condition all values on $y_n$ is trivial given domain knowledge. Alternatively, they could also be drawn at random, but not be conditioned on $y_n$. Such relations are omitted here for simplicity.

We can now fully define $\phi$ as:
\begin{equation}
\phi(y_n,z_{n,m}) = \mathbbm{1}(c(t_{\text{max}}, x_{n,m})<T_{\text{opp}}).
\end{equation}
Taking the expectation of $\phi$ over $M$ different realizations of $z$ yields the estimate $(\hat{\gamma}_M)_n$. 
This value is the probability that treatment will be successful for a particular patient, marginalizing over possible unobserved dynamics. This is the point at which clinician decides whether initiate the treatment plan.
This decision is represented $f(y_n,(\hat{\gamma}_M)_n) \in [0,1]$ as:
\begin{equation}
f(y_n,(\hat{\gamma}_M)_n) = \mathbbm{1}((\hat{\gamma}_M)_n>T_{\text{treat}})
\end{equation}
where $T_{\text{treat}}$ is the minimum probability of success required for that patient to receive the treatment, and again, could be conditioned on $y$ also.
Taking the expectation of $f$ over patients yields the expected frequency with which the treatment will be delivered, given a value of $T_{\text{treat}}$. The hospital wishes to estimate the value $T_{\text{treat}}$ that maximizes the number of patients treated, while only treating those patients with the highest probability of success, and (in expectation) staying within the budgetary constraint.

\begin{figure}[t]
	\centering
	\includegraphics[width=0.48\textwidth]{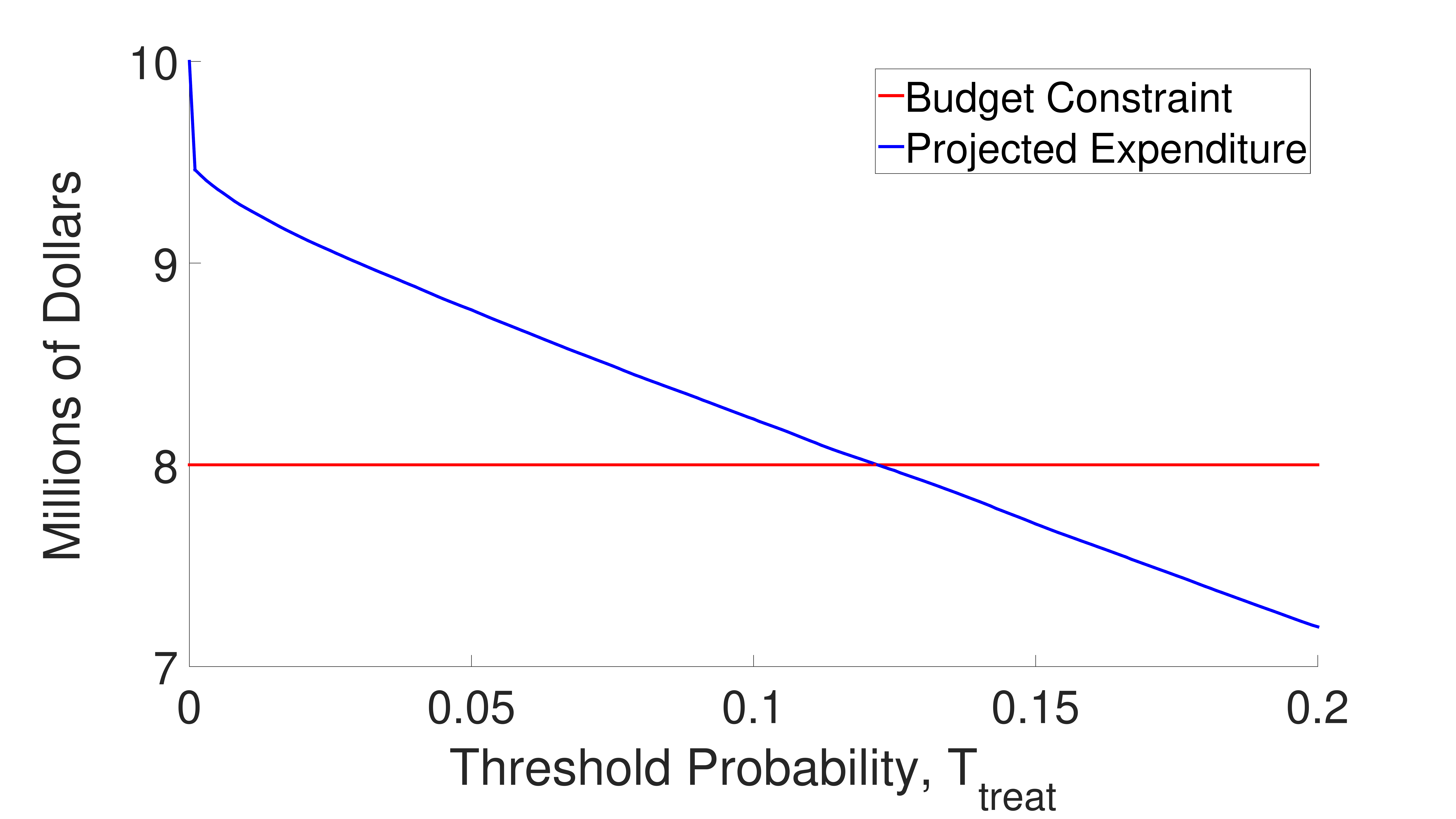}
	\caption{Projected expenditure (proportional to $I_{N,M}$) evaluated at different values of $T_{\mathrm{treat}}$. The budget constraint is shown by the horizontal red line. The optimal value of $T_{\mathrm{treat}}$ is found by the intersection 
		and occurs at $T_{\mathrm{treat}} = 12.5\%$. Evaluated was carried out at 100 $T_{\mathrm{treat}} \i$. Only the bottom 20\% is pictured as this is the operating range for most treatment centers. }
	\label{fig:emperical-cost}
\end{figure}

The model is completed by the definition of the following distributions and parameters.
\begin{align*}
& K_0 = 100000000, \quad
\phi = 0.001, \quad
\psi = 0.05, \quad
\lambda = 0.5, \quad
\xi \sim \text{Beta}(5, 2), \\
& c_0 \sim 1000*\text{Rayleigh}(10), \quad
T_{\text{opp}} = 2000,\quad
T_{\text{treat}} = 0.35, \quad
t_{\text{max}} = 250, \quad
t_{\text{step}} = 0.01 
\end{align*}

\subsection{Budget result}
\label{sec:cancer_sim_result}

In the example outlined above, the treatment center is not actually attempting to evaluate the value of $I$, but to find the optimal value of $T_{\text{treat}}$ subject to a budgetary constraint. A simplistic way of evaluating the optimal value is to perform a dense search over different values of the parameter, each time evaluating the estimated expenditure, and select the best performing value.

Figure~\ref{fig:emperical-cost} shows the variation of predicted expenditure against the threshold probability, as well as the budget constraint. The intersection of these curves is the optimal setting of $T_{opp}$, here evaluated to be 12.5\%. From the blue line, it is clear that the relationship between expenditure and treatment probability is non-linear, especially at the extrema of the distribution, and hence the use of NMC was necessarily for evaluating the optimal value.

% !TEX root =  main.tex

\section{Bayesian Experimental Design}
\label{sec:exp-design}

Bayesian experimental design provides a framework for designing experiments in a manner that is optimal from
an information-theoretic viewpoint \citep{chaloner1995bayesian,sebastiani2000maximum}.  By minimizing the entropy in the posterior distribution of the 
parameters of interest, one can maximize the information gathered by the experiment.

Let the parameters of interest be denoted by $\theta \in \Theta$ for which we define a prior distribution $p(\theta)$.
Let the probability of achieving outcome $y\in\mathcal{Y}$, given parameters $\theta$ 
and a design $d \in \mathcal{D}$, be defined by likelihood model $p(y | \theta, d)$.
Under our model, the outcome of the experiment given a chosen $d$ is distributed according to
\begin{equation}
\label{eq:marginal_def}
p(y | d) = \int_{\Theta} p(y,\theta | d) d\theta = \int_{\Theta} p(y | \theta, d) p(\theta) d\theta.
\end{equation}
where we have used the fact that $p(\theta)=p(\theta|d)$ because $\theta$ is independent of the design.
Our aim is to choose the optimal design $d$ under some criterion. 
We, therefore, define a utility function, $U(y,d)$, representing the utility of choosing a design $d$ 
and getting a response $y$.
Typically our aim is to maximize information gathered from the experiment, and so we set 
$U(y,d)$ to be the gain in Shannon information between the prior and the posterior:
\begin{align}
\label{eq:shannon_inf}
U(y,d) &= \int_{\Theta} p(\theta |y, d) \log(p(\theta |y, d)) d\theta -\int_{\Theta} p(\theta) \log(p(\theta))d\theta
\end{align}
However, we are still uncertain about the outcome. Thus, we use the expectation of $U(y,d)$ with respect to $p(y | d)$
as our target:
\begin{align}
\bar{U}(d) & = \int_{\mathcal{Y}} U(y,d) p(y|d) dy\nonumber\\
&=\int_{\mathcal{Y}}\int_{\Theta} p(y,\theta | d)\log(p(\theta |y, d)) d\theta dy - \int_{\Theta} p(\theta) \log(p(\theta)) d\theta \nonumber\\
&=\int_{\mathcal{Y}}\int_{\Theta} p(y,\theta | d)\log\left(\frac{p(\theta |y, d)}{p(\theta)}\right)d\theta dy. 
\label{eq:u_bar_1}
\end{align}
noting that this corresponds to the mutual information between the parameters $\theta$ and
the observations $y$.  The Bayesian-optimal design is then given by
\begin{equation}
\label{eq:d_star}
d^* = \argmax_{d \in \mathcal{D}} \bar{U}(d).
\end{equation}

Finding $d^*$ is challenging because the posterior $p(\theta |y, d)$ is rarely known in closed form.  
To solve the problem, we proceed by rearranging~\eqref{eq:u_bar_1} using Bayes' rule (remembering that
$p(\theta)=p(\theta|d)$):
\begin{align}
\label{eq:u_bar_2}
\begin{split}
\bar{U}(d) = & \int_{\mathcal{Y}}\int_{\Theta} p(y,\theta | d) \log\left(\frac{p(\theta | y, d)}{p(\theta)}\right) d\theta dy \\
=& \int_{\mathcal{Y}}\int_{\Theta} p(y,\theta | d) \log\left(\frac{p(y | \theta, d)}{p(y |d)}\right) d\theta dy \\
=&\int_{\mathcal{Y}}\int_{\Theta} p(y,\theta | d) \log(p(y | \theta, d)) d\theta dy - \int_{\mathcal{Y}} p(y | d) \log(p(y | d))dy.
\end{split}
\end{align}
The first of these terms can now be evaluated using standard MC approaches as the integrand is analytic.  
In contrast, the second term is not directly amenable to standard MC estimation
as the marginal $p(y|d)$ represents an expectation
and taking its logarithm represents a non-linear functional mapping.  

To derive an estimator, we will now consider these terms separately.  Starting
with the first term,
\begin{align}
\bar{U}_1(d) = &\int_{\mathcal{Y}}\int_{\Theta} p(y,\theta | d) \log(p(y | \theta, d)) d\theta dy
\approx \frac{1}{N} \sum_{n=1}^{N} \log(p(y_n | \theta_n, d)) \label{eq:U1_MC}
\end{align}
where $\theta_n \sim p(\theta)$ and $y_n \sim p(y|\theta=\theta_n, d)$.  We
note that evaluating~\eqref{eq:U1_MC} involves both
sampling from $p(y | \theta, d)$ and directly evaluating it point-wise.
The latter of these cannot be avoided, but in the scenario where we
do not have direct access to a sampler for $p(y | \theta, d)$, we can
use the standard importance sampling trick, sampling instead
$y_n \sim q(y|\theta=\theta_n, d)$ and weighting the samples in~\eqref{eq:U1_MC}
by $w_n = \frac{p(y_n|\theta_n, d)}{q(y_n|\theta_n, d)}$.

Now considering the second term we have
\begin{align}
\bar{U}_2(d) = &\int_{\mathcal{Y}} p(y | d) \log(p(y | d))dy
\approx \frac{1}{N} \sum_{n=1}^{N} \log \left(\frac{1}{M} \sum_{m=1}^{M} p(y_n | \theta_{n,m}, d)\right) \label{eq:U2_MC}
\end{align}
where $\theta_{n,m} \sim p(\theta)$ and $y_n \sim p(y | d)$.  Here we can sample the latter by first sampling an otherwise unused $\theta_{n,0} \sim p(\theta)$ and 
then sampling $y_n \sim p(y | \theta_{n,0}, d)$.  Again we can use importance sampling if 
we do not have direct access to a sampler for $p(y | \theta_{n,0}, d)$.

Putting~\eqref{eq:U1_MC} and~\eqref{eq:U2_MC} together (and renaming
$\theta_n$ from~\eqref{eq:U1_MC} as $\theta_{n,0}$ for notational
consistency with ~\eqref{eq:U2_MC})  we now have the following complete estimator
given in the main paper and implicitly used by \cite{myung2013tutorial} amongst others
\begin{align}
\label{eq:exp-des-nmc}
\bar{U}(d) 
& \approx  
\frac{1}{N} \sum_{n=1}^{N} \left[ \log(p(y_n | \theta_{n,0},d)) 
- \log \left(\frac{1}{M} \sum_{m=1}^{M}p(y_n | \theta_{n,m},d)\right) \right]
\end{align}
where $\theta_{n,m} \sim p(\theta) \; \forall m \in 0:M, \;n \in 1:N$ and $y_n \sim p(y|\theta=\theta_{n,0}, d) \; \forall n \in 1:N$.

We now show that if $y$ can only take on one of $C$ possible values ($y_1, \ldots, y_C$), 
we can achieve significant improvements in the convergence rate by using a similar to that
introduced in Section 3.2 to convert to single MC estimator:
\begin{align}
\bar{U}(d)=&\int_{\mathcal{Y}}\int_{\Theta} p(y,\theta | d) \log(p(y | \theta, d)) d\theta dy - \int_{\mathcal{Y}} p(y | d) \log(p(y | d))dy \nonumber\\
=& \int_{\Theta} \left[\sum_{c=1}^{C} p(\theta) p(y_c|\theta, d) \log(p(y_c | \theta, d)) \right] d\theta
-\sum_{c=1}^{C} p(y_c | d)\log(p(y_c | d)) \nonumber \\
\approx& \label{eq:u_bar_MC}
\frac{1}{N} \sum_{n=1}^{N} \sum_{c=1}^{C} p(y_c | \theta_n, d) \log\left(p(y_c | \theta_n, d)\right) 
- \sum_{c=1}^{C} \left[\left(\frac{1}{N}\sum_{n=1}^{N} p(y_c | \theta_n, d)\right) \log \left(\frac{1}{N} \sum_{n=1}^{N} p(y_c | \theta_n, d)\right) \right]
\end{align}
where $\theta_n \sim p(\theta) \quad \forall n \in 1,\dots,N$.  As $C$ is a fixed constant,
the MSE for first term clearly converges at the standard MC error rate of $O(1/N)$.  Similarly each
$\hat{P}_N(y_c | d) = \frac{1}{N}\sum_{n=1}^{N} p(y_c | \theta_n, d)$ term also converges at a rate 
$O(1/N)$ to $p(y_c | d)$.  Now noting that $\hat{P}_N(y_c | d) \le 1$ and that $f(x) = x \log x$ is Lipschitz
continuous in the range $(0,1]$, each $\hat{P}_N(y_c | d) \log \left(\hat{P}_N(y_c | d)\right)$ 
term must also converge at the MC error rate if $p(y_c | d)>0 \; \forall c=1,\dots,C$.  
Finally if we assume that when $p(y_c | d)=0$ then $\hat{P}_N(y_c | d)=0$
almost surely for sufficiently large $N$, then the second term also converges at the MC error when
$p(y_c | d)=0$.  We now have a finite sum of terms which each convergence to $\bar{U}(d)$ with MC
MSE rate $O(1/N)$, and so the overall estimator~\eqref{eq:u_bar_MC} must also converge at this rate.
This compares to $O(1/T^{2/3})$ for~\eqref{eq:exp-des-nmc} (assuming we take $N \propto M^2$), noting that generating $T$ samples for~\eqref{eq:exp-des-nmc} has the same cost up to a constant factor as generating $N$ for~\eqref{eq:u_bar_MC}.
To the best of our knowledge, this is the first introduction of this superior estimator in the literature.

\begin{figure}[t]
	\centering
	\includegraphics[width=0.49\textwidth,trim={1.5cm 0 3.5cm 0},clip]{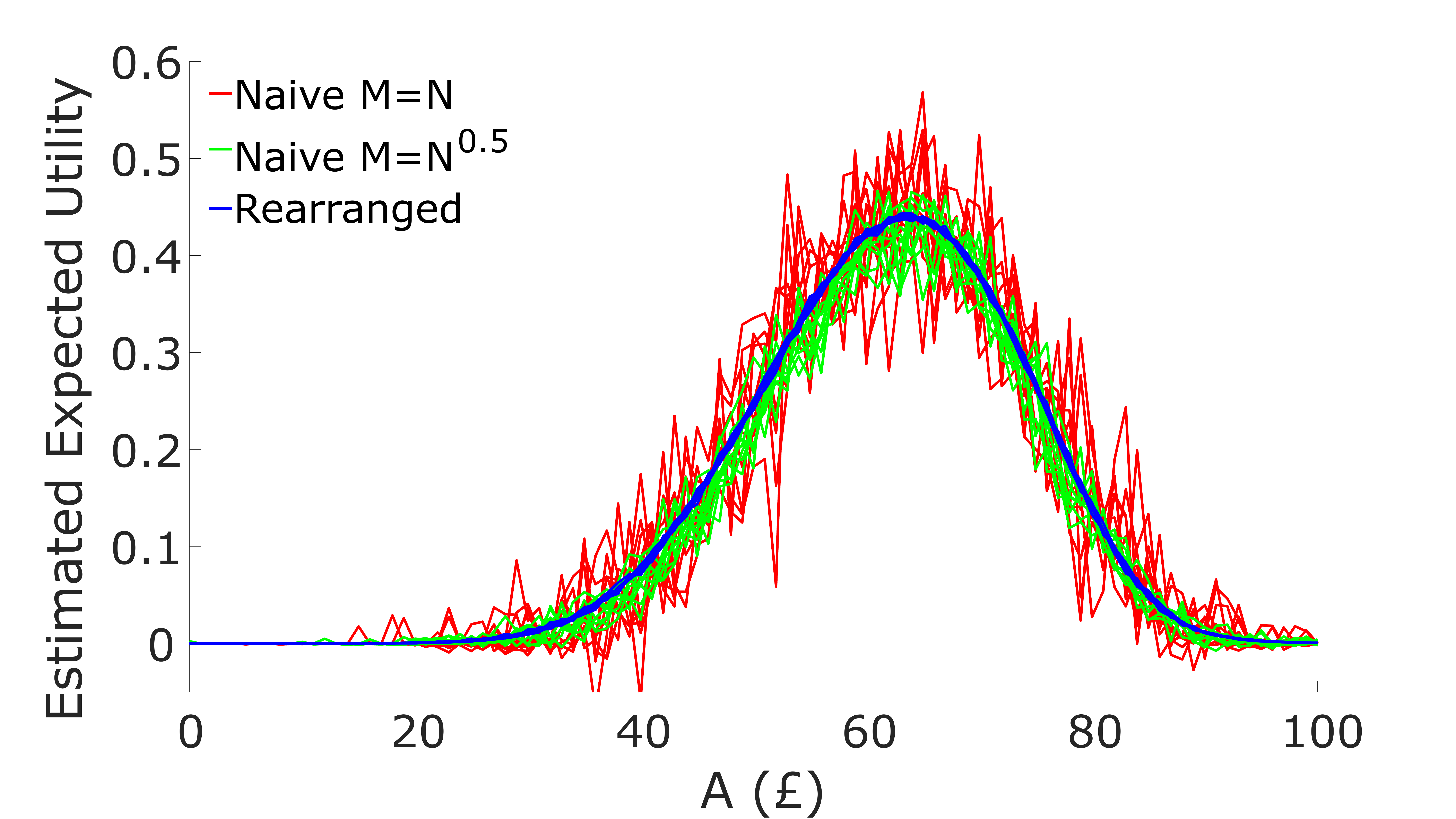}
	\caption{Estimated expected utilities $\bar{U}(d)$for 
		different values of one of the design parameters $A \in \{1,2,\dots,100\}$ given a fixed total
		sample budget of $T=10^4$.  Here the lines correspond to 10 independent runs, showing
		that the variance of \eqref{eq:exp-des-nmc} is far higher than \eqref{eq:u_bar_MC}.\label{fig:exp-d-scan}}
\end{figure}

We finish by showing that the theoretical advantages of this reformulation also leads to empirical gains in the estimation of $\bar{U}(d)$.  For this, we consider a model used in psychology experiments for delay discounting introduced by \cite{vincent2016hierarchical,vincent2017darc}.  Our experiment comprises of asking questions of the form \emph{``Would you prefer $\pounds A$ now, or $\pounds B$ in $D$ days?''} and we wish to choose the question  variables $d = \{A,B,D\}$ in the manner that will give the most incisive questions.  The target participant is presumed to have parameters $\theta=\{k,\alpha\}$ and the following response model
\begin{equation}
y \sim \mathrm{Bernoulli} \left(0.01 + 0.98 \cdot \Phi\left(\frac{1}{\alpha} \left(\frac{B}{1+e^k D}-A\right)\right)\right)
\end{equation}
where $y=1$ indicates choosing the delayed response and $\Phi$ represents the cumulative normal distribution.  As more questions are asked, the distribution over the parameters $\theta$ is updated, such
that the most optimal question to ask at a particular time depends on the previous questions
and responses.  For the sake of brevity, when comparing the performance of~\eqref{eq:exp-des-nmc} and~\eqref{eq:u_bar_MC} we will neglect the problem of how best to optimize the design, and consider
only the problem of evaluating $\bar{U}(d)$.  We will further consider the case where $B=100$ and $D = 50$ are fixed and we are only choosing the delayed value $A$.
We presume the following distribution on the parameters
\begin{align*}
k &\sim \mathcal{N}(-4.5,0.5^2) \\
\alpha &\sim \Gamma(2,2).
\end{align*}
We first consider convergence in the estimate of $\bar{U}(d)$ for the case $A=70$ for our suggested method~\eqref{eq:u_bar_MC} and the na\"{i}ve solution~\eqref{eq:exp-des-nmc}, the results of which are shown in Figure~2a in the main paper. 
Here we see that the convergence rates of the two methods are both as expected and that our suggested method offers significant empirical performance improvements.  

We next consider setting a total sample budget $T=10^4$ and look at the variation in the estimated values of $\bar{U}(d)$ for different values of $A$ for the two methods as shown in Figure~\ref{fig:exp-d-scan}.
This shows that the improvement in MSE leads to clearly visible improvements in the characterization of $\bar{U}(d)$ that
will translate to improvements in seeking the optimum.

\end{appendices}

\section*{Acknowledgements}

Tom Rainforth's research leading to these results has received funding from the
European Research Council under the European Union's Seventh Framework
Programme (FP7/2007-2013) ERC grant agreement no. 617071.  However, the
majority of this work was undertaken while he was in the 
Department of Engineering Science, University of Oxford, and was supported 
by a BP industrial grant. Robert Cornish is supported by an NVIDIA scholarship. Hongseok Yang is supported by
an Institute for Information \& communications Technology Promotion (IITP)
grant funded by the Korea government (MSIP) (No.R0190-16-2011, Development of Vulnerability Discovery Technologies for IoT Software Security).  Frank Wood is supported under DARPA PPAML through the U.S. AFRL under Cooperative Agreement FA8750-14-2-0006, Sub Award number 61160290-111668. 

{\bibliography{refsnestedmc}}

\end{document}